\documentclass{article}
\usepackage[utf8]{inputenc}
\usepackage[english]{babel}
\usepackage[T1]{fontenc}
\usepackage{algorithm}
\usepackage{amsmath,amssymb,amsthm}
\usepackage{anyfontsize}
\usepackage{array}
\usepackage{authblk}
\usepackage{bbm}
\usepackage{bm}
\usepackage{blindtext}
\usepackage{booktabs}
\usepackage{diagbox}
\usepackage{dsfont}
\usepackage{enumitem}
\usepackage{fancyhdr} 
\usepackage{float}
\usepackage{gensymb}
\usepackage{glossaries}
\usepackage{graphicx}
\usepackage{ifthen}
\usepackage{lipsum}
\usepackage{listings}
\usepackage{mathrsfs}
\usepackage{mathtools}
\usepackage{multirow}
\usepackage{relsize}
\usepackage{stmaryrd}
\usepackage{subcaption}
\usepackage{textcomp}
\usepackage{ulem}
\usepackage{url}
\usepackage{appendix}
\usepackage{verbatim}
\usepackage[dvipsnames]{xcolor}
\definecolor{blue}{HTML}{1F77B4}
\definecolor{orange}{HTML}{FF7F0E}
\definecolor{green}{HTML}{2CA02C}
\definecolor{red}{HTML}{D62728}
\definecolor{purple}{HTML}{9467BD}
\definecolor{brown}{HTML}{8C564B}
\definecolor{pink}{HTML}{E377C2}
\definecolor{grey}{HTML}{7F7F7F}
\definecolor{yellow}{HTML}{BCBD22}
\definecolor{cyan}{HTML}{17BECF}
\definecolor{turquoise}{HTML}{3FE0D0}

\usepackage[section]{placeins}
\usepackage[authoryear]{natbib}
\usepackage[unicode=true,bookmarks=true,bookmarksnumbered=true,bookmarksopen=true,bookmarksopenlevel=1,breaklinks=true,pdfborder={0 0 1},backref=false,colorlinks=true,linkcolor=black, citecolor=blue, urlcolor=blue, filecolor=blue,pdfpagelayout=OneColumn, pdfnewwindow=true,pdfstartview=XYZ, plainpages=false,, hyperindex=true]{hyperref}
\usepackage[top=27mm, bottom=27mm, left=25mm, right=25mm]{geometry}
\usepackage{tikz}
\usetikzlibrary{shapes}
\usetikzlibrary{trees}
\usetikzlibrary{calc}

\newcolumntype{R}[1]{>{\raggedleft\arraybackslash }b{#1}}
\newcolumntype{L}[1]{>{\raggedright\arraybackslash }b{#1}}
\newcolumntype{C}[1]{>{\centering\arraybackslash }b{#1}}

\newcommand{\R}{\mathbb{R}}

\newcommand{\N}{\mathbb{N}}
\newcommand{\E}{\mathbb{E}}

\newcommand{\C}{\mathbb{C}}
\newcommand{\Prob}{\mathbb{P}}

\newcommand{\Tau}{\mathcal{T}}

\DeclarePairedDelimiterX{\expectarg}[1]{(}{)}{%
  \ifnum\currentgrouptype=16 \else\begingroup\fi
  \activatebar#1
  \ifnum\currentgrouptype=16 \else\endgroup\fi
}

\newcommand{\innermid}{\nonscript\;\delimsize\vert\nonscript\;}
\newcommand{\activatebar}{%
  \begingroup\lccode`\~=`\|
  \lowercase{\endgroup\let~}\innermid 
  \mathcode`|=\string"8000
}

\newcommand\tab[1][0.5cm]{\hspace*{#1}}

\theoremstyle{plain}

\theoremstyle{definition}
\newtheorem{defn}{\protect\definitionname}
\theoremstyle{plain}
\newtheorem{lem}{\protect\lemmaname}
\theoremstyle{plain}
\newtheorem{thm}{\protect\theoremname}
\theoremstyle{remark}

\theoremstyle{plain}
\newtheorem{prop}{\protect\propositionname}
\theoremstyle{plain}
\newtheorem{cor}{\protect\corollaryname}

\usepackage[algo2e,ruled,vlined]{algorithm2e}
\definecolor{algoColorKeyword}{named}{blue}
\definecolor{algoColorComment}{named}{olive}

\SetKwComment{Comment}{$\triangleright$\ }{}

\usepackage{xparse} \DeclarePairedDelimiterX{\Iintv}[1]{\llbracket}{\rrbracket}{\iintvargs{#1}}
\NewDocumentCommand{\iintvargs}{>{\SplitArgument{1}{,}}m}
{\iintvargsaux#1} %
\NewDocumentCommand{\iintvargsaux}{mm} {#1\mkern1.5mu..\mkern1.5mu#2}

\DeclareMathOperator*{\argmin}{arg\,min} % Jan Hlavacek
\newcommand{\Var}{\mathbb{V}\hspace{-0.08334em}\mathrm{ar}}
\newcommand{\Cov}{\mathbb{C}\mathrm{ov}}
\newcommand{\cqfd}{\hspace*{\fill}$\blacksquare$}

\makeatother
\providecommand{\assumptionname}{Assumption}
\providecommand{\corollaryname}{Corollary}
\providecommand{\definitionname}{Definition}
\providecommand{\lemmaname}{Lemma}
\providecommand{\propositionname}{Proposition}
\providecommand{\remarkname}{Remark}
\providecommand{\theoremname}{Theorem}

\usepackage{authblk}
\usepackage{orcidlink}
\definecolor{orcidlogocol}{HTML}{2E7E9F}

\title{\huge Fast simulation of Volterra processes using\\ random Fourier features with application to\\ the log-stationary fractional Brownian motion}

\author[1]{\orcidlinki{Othmane Zarhali}{0009-0007-5413-3342}}
\author[2]{\orcidlinki{Nicolas Langren\'e}{0000-0001-7601-4618}\thanks{Corresponding author, nicolaslangrene@bnbu.edu.cn}}

\affil[1]{\normalsize Ceremade, CNRS-UMR 7534, Université Paris-Dauphine PSL, Place du Maréchal de Lattre de Tassigny, 75016 Paris, France}
\affil[2]{\normalsize Guangdong Provincial/Zhuhai Key Laboratory of Interdisciplinary Research and Application for Data Science, Beijing Normal-Hong Kong Baptist University, Zhuhai, China} 

\date{\today}

\begin{document}

\maketitle

\begin{abstract}

A fast simulation framework for stochastic Volterra processes based on Random Fourier Features (RFF) approximation of the kernel is developed. After recalling the main properties of Volterra processes and reviewing existing numerical simulation methods, an accelerated scheme is introduced that relies on a spectral representation of the kernel. A particular attention is devoted to sampling from the kernel spectral density using Hamiltonian Monte Carlo, whose efficiency and stability bring more convenience than alternative sampling procedures. Quantitative guarantees for the proposed method are established, including moment estimates and strong error bounds. The approach is further compared with the kernel approximation by sum of exponentials commonly used in the literature, emphasizing the broader generality of the present framework. As a primary application, Volterra processes associated with the Stationary fractional Brownian Motion (S-fBM) kernel are investigated. A spectral density representation is derived in closed form using hypergeometric functions, a condition for positive definiteness is established and explicit truncation as well as Monte Carlo error bounds are provided for the RFF approximation in this setting. Numerical experiments in dimensions one and two illustrate the accuracy of the kernel approximation, the reliable recovery of model parameters and the competitiveness of the accelerated simulation scheme in terms of computational efficiency and both weak and strong error performance.

\end{abstract}


%\tableofcontents

\section{Introduction}

Stochastic Volterra equations (SVEs) have emerged as a central framework 
for modelling non-Markovian dynamics with memory effects, with important developments in both theory and applications. Memory effects and non-Markovian dynamics play an important role across several applied domains. In finance, Volterra-type processes and rough models have been widely used to capture persistence and microstructural effects \citep{JaissonRosenbaum2016, ElEuchFukasawaRosenbaum2018, abijaber2019affine}. In biology, stochastic models with memory naturally arise in population dynamics \citep{Allen2010, bressloff2021stochastic1, bressloff2021stochastic2, anderson2015stochastic}. These works illustrate the broad relevance of Volterra-type frameworks for modelling complex systems with memory.
In contrast to classical Itô stochastic differential equations, 
the future evolution of an SVE depends on the entire past trajectory 
through a deterministic kernel, typically of convolution type. 
This non-local structure allows one to capture rough behavior, long-range dependence and hereditary phenomena arising in physics,  biology and more recently in mathematical finance.

A general stochastic Volterra equation of convolution type can be written as:
\begin{equation}\label{eq:intro_SVE}
X_t = X_0 + \int_0^t K(t-s)b(X_s)\,ds 
      + \int_0^t K(t-s)\sigma(X_s)\,dW_s,
\end{equation}
where $K$ is a deterministic kernel and $W$ is a Brownian motion. When the kernel exhibits a fractional singularity of the form $K(t) = t^{H-\frac12}$ with $H\in(0,1)$, 
the solution typically fails to be a semimartingale for $H<\frac12$, placing the model outside the scope of classical Itô calculus. A prominent example is the rough Heston model introduced in \citet{ElEuchRosenbaum2019}, which combines a fractional Volterra-type variance process with affine transform techniques. The model preserves tractability through a fractional Riccati equation while reproducing key stylized facts of implied volatility surfaces, thus establishing a deep connection between stochastic Volterra equations and modern mathematical finance.

A systematic weak solution theory for stochastic Volterra equations of convolution type has been developed in \citet*{AbiJaberCuchieroLarssonPulido2019}.  Their work establishes general existence, stability and approximation results and provides a robust framework for analyzing non-Markovian dynamics under minimal structural assumptions on the kernel.
In particular, they introduce techniques based on resolvent measures and affine transform methods that have become fundamental in the area.

A major breakthrough in connecting SVEs with tractable probabilistic structures was achieved in the theory of affine Volterra processes \citep{abijaber2019affine}. This work extends the classical affine process framework to the Volterra setting, establishing exponential–affine transform formulas and Riccati–Volterra equations that characterize the Fourier–Laplace transform of the solution. These results provide the theoretical backbone for rough volatility models such as the rough Heston model.

From a numerical and approximation perspective, 
\citet{bayer2023markovian} proposed and analyzed multi-factor Markovian approximations of stochastic Volterra equations with fractional kernels. They prove strong convergence rates and show how singular Volterra dynamics can be approximated by finite-dimensional Markovian systems, 
thus bridging the gap between non-Markovian theory and computational practice. More recently, significant analytical progress has been made on  pathwise uniqueness and strong solvability for singular kernels. \citet{Promel2025} establish pathwise uniqueness  for a broad class of singular stochastic Volterra equations, extending classical Yamada–Watanabe arguments to the non-Markovian setting. 
Their results clarify the interplay between kernel singularity and coefficient regularity in determining strong well-posedness.\\

Besides this singular power kernel, it can be of interest to study SVEs for other general classes of kernels.  The focus of this paper is the class of \textit{positive definite} kernels.  By definition,  a shift-invariant, isotropic kernel $K:[0,\infty)\rightarrow\mathbb{R}$ is said to be \textit{positive definite} in $\mathbb{R}^{d}$
if for any $N\geq1$, $(\mathbf{x}_{1},\ldots,\mathbf{x}_{N})\in\mathbb{R}^{d\times N}$
and $(z_{1},\ldots,z_{N})\in\mathbb{R}^{N}$,
\begin{equation}
\sum_{i=1}^{N}\sum_{j=1}^{N}z_{i}z_{j}K(\left\Vert \mathbf{x}_{i}-\mathbf{x}_{j}\right\Vert )\geq0.\label{eq:positive_definite}
\end{equation}

Let $\Phi_{d}$ be the set of kernels that are positive definite in $\mathbb{R}^{d}$ and let $\Phi_{\infty}$ be the set of kernels that are in $\Phi_{d}$ for all $d\geq1$.  The construction and characterization of positive definite kernels have been extensively studied in the literature, since they coincide with covariance functions \citep{yaglom1987correlation1,yaglom1987correlation2}.  According to Bochner's theorem \citep[Theorem~1.7.3]{sasvari2013characteristic},
$K\in\Phi_{d}$ if and only if it is the inverse Fourier transform of a finite nonnegative measure on $[0,\infty)$.
In particular, $K\in\Phi_{d}$ admits the probabilistic representation:
\begin{equation}
\forall\mathbf{u}\in\mathbb{R}^{d},\tab K(\Vert\mathbf{u}\Vert)=K(0)\mathbb{E}\left[\cos(\bm{\eta}.\mathbf{u})\right]\  \label{eq:Ecos}
\end{equation}
where the density of the random vector $\boldsymbol{\eta}$ is the
inverse Fourier transform of $K(\left\Vert \mathbf{.}\right\Vert )/K(0)$.
Such a kernel \eqref{eq:Ecos} can be well approximated by a sum of cosine terms; this is the \textit{random Fourier features} methodology \citep{rahimi2007random}, also known as  \textit{spectral sampling}.  It provides the following explicit kernel approximation: 
 \begin{align}
 & K(\Vert\mathbf{x}_{i}-\mathbf{x}_{j}\Vert)\approx\varphi(\mathbf{x}_{i})^{\top}\varphi(\mathbf{x}_{j}),\ \mathrm{where,\ for\ any\ }\mathbf{x}\in\mathbb{R}^{d},\label{eq:feature_mapping_1}\\
 & \varphi(\mathbf{x}):=\frac{\sqrt{K(0)}}{\sqrt{M}}\left[\begin{array}{cccccc}
\!\cos(\bm{\eta}_{1}^{\top}\mathbf{x})\! & \!\ldots\! & \!\cos(\bm{\eta}_{M}^{\top}\mathbf{x})\! & \!\sin(\bm{\eta}_{1}^{\top}\mathbf{x})\! & \!\ldots\! & \!\sin(\bm{\eta}_{M}^{\top}\mathbf{x})\!\end{array}\right]^{\top}\!\in\mathbb{R}^{2M}\label{eq:feature_mapping_2}
\end{align}
The main objective of our paper is to leverage the kernel approximation~\eqref{eq:feature_mapping_1} in the context of stochastic Volterra equations~\eqref{eq:intro_SVE}, in order to accelerate the numerical simulation of the solution of such SVEs.

This proposed methodology can be applied whenever the kernel function $K$ is positive definite. A specific example of interest is the so-called stationary fractional Brownian motion (S-fBM) kernel, defined by:
\begin{equation}
K(\tau)=\dfrac{\nu^2}{2} \left( 1 - \left( \dfrac{|\tau|}{T} \right)^{2H} \right)\mathbbm{1}_{\{|\tau|\leq T\}},\tab  \tau \in \R\label{eq:sfbm_kernel}
\end{equation}
The S-fBM covariance kernel appears in the context of the Log S-fBM model introduced in \citet{wu2022rough} to bridge the rough volatility regime for non vanishing Hurst exponent $H$ and the multifractal regime for vanishing Hurst exponent, characterized by the multifractal random walk model (see \citealt{muzy2000modelling,muzy2002multifractal,bacry2003log}). The Log S-fBM model has been used as a building block for the Nested Stationary fractional factor model (NSfFM, \citealt{zarhali2025rougher}) which accounts for the empirical findings in \citet{wu2022rough} that the Hurst exponent of single stocks are of order $H\simeq 0.01$ while the one of indices such as the S\&P500 are of order $H\simeq 0.1$. A multivariate extension of Log S-fBM was developed recently in \citet{zarhali2026multidimlogsfbm} by considering a multidimensional analogue of the Hurst exponent (the co-Hurst matrix) and of the so-called intermittency parameter (the co-intermittency matrix) to define a multidimensional stationary fractional Brownian motion (mS-fBM). The model is able to interpolate between multidimensional rough volatility and multifractal regimes.

Volterra processes with such kernel are worth studying for multiple reasons. This kernel encodes rough, scale-dependent dependence and reproduces the kind of rough behavior observed in many empirical time series (notably volatility in finance), while the truncation at scale $T$ introduces a finite memory horizon. This combination is realistic: it captures short-time roughness together with long-time decorrelation, which standard kernels  do not represent simultaneously. Furthermore, the compact support $|\tau| \leq T$ makes the model computationally tractable yet non-Markovian. Unlike classical Markov processes, Volterra processes retain memory, which is crucial for modelling persistence and clustering effects. As a further contribution to this literature, our paper is going to define, analyze and simulate numerically a Volterra process with the S-fBM kernel~\eqref{eq:sfbm_kernel}.\\

Specifically, our paper makes the following contributions:
\begin{itemize}
\item We introduce a fast simulation scheme for Volterra processes based on random Fourier features.  We propose to perform the spectral sampling using the Hamiltonian Monte Carlo method.  Moment estimates and strong error bounds of the proposed numerical scheme are established.
\item We compare our kernel decomposition by random Fourier features to the traditional multifactor approximation of Volterra kernels by a sum of exponential functions.  In particular, we prove mathematically that our proposed decomposition can be successfully applied to a much larger class of kernel functions, including kernels with compact support such as the S-fBM kernel.
\item We apply the proposed numerical scheme to Volterra processes defined by an S-fBM convolution kernel.  In doing so,  we establish a new formula for the spectral distribution of isotropic S-fBM kernels in terms of hypergeometric functions,  along with an explicit condition to ensure its positive definiteness.  Moreover,  we establish an explicit error bound of the RFF simulation scheme tailored to the case of the S-fBM kernel. 
\item We illustrate numerically the competitiveness of fast Volterra simulation by RFF in terms of both computational runtime and error (weak and strong).
\end{itemize}

The paper is organized as follows.  Section~\ref{sec:volterra} recalls the main properties of Volterra processes and the existing methods to simulate them.  Section~\ref{sec:simulation} introduces the fast Volterra simulation scheme built upon the random Fourier features approximation of kernel functions and Section~\ref{sec:multifactor} compares it to multifactor approximation of kernel functions by sums of exponentials. Section~\ref{sec:logsfbm} defines the log-stationary fractional Brownian motion model and Section~\ref{sec:spectralsfbm} establishes spectral properties of the S-fBM kernel. Section~\ref{sec:numerical} provides detailed numerical results about the RFF Volterra scheme. Finally,  Section~\ref{sec:conclusion} concludes the paper.

\section{Fast simulation of Volterra processes using Random Fourier Features\label{sec:simulation}}

\subsection{Volterra processes\label{sec:volterra}}

Let $(\Omega,\mathcal{F},(\mathcal{F}_t)_{t\ge0},\mathbb{P})$ be a filtered probability space and $(W_t)_{t\ge0}$ be a standard $(\mathcal{F}_t)$-Brownian motion. Stochastic Volterra processes form a broad class of non-Markovian models defined through integral representations of the form:
\begin{eqnarray}
\label{eq:volterradynamic}
X_t =X_0+\int_0^t K(t,s) \sigma(s,X_s)\,\mathrm{d}W_s
\footnotemark
\end{eqnarray}
\footnotetext{$^\ast$ The initial condition of the dynamics of Eq.~\eqref{eq:volterradynamic} as well as all the dynamics introduced in the paper are assumed to be random variables in  $\mathcal{L}^{\infty}$.}
where \(\sigma \colon \mathbb{R} \to \mathbb{R}\) is a globally Lipschitz continuous function and \(K \colon \mathbb{R}_{+}^2 \to \mathbb{R}_+\) is a nonnegative, integrable kernel function. These processes naturally capture memory effects and long-range dependence that arise in turbulence, finance, statistical physics, and other application fields.

A prominent example is the Riemann-Liouville fractional Brownian motion (fBM) with Hurst exponent $H \in (0,1)$, which admits the following Volterra representation, also known as the Mandelbrot-Van-Ness representation:
\begin{eqnarray}
B_t^H = \int_0^t K_H(t,s)\,\mathrm{d}W_s,
\end{eqnarray}
where $K_H$ is a singular kernel behaving like $(t-s)^{H-\frac12}$, see \citet{mandelbrot1968fractional, decreusefond1998fractional}. This representation has motivated the study of rough paths and stochastic calculus for Volterra-type processes \citep{coutin2000rough}.

In mathematical finance, stochastic Volterra processes gained significant attention with the advent of rough volatility models, which generalize classical Markovian stochastic volatility models using Volterra-driven processes. For example, the rough Bergomi model \citep*{bayer2016pricing} defines the log-volatility process $\xi$ as:
\begin{eqnarray}
\log \left(\xi_t\right) = \log \left(\xi_0\right) + \eta \int_0^t (t-s)^{H-\frac12}\,\mathrm{d}W_s,
\end{eqnarray}
with $H<\frac12$, capturing the empirically observed roughness of volatility \citep{bayer2016pricing, gatheral2018volatility}. More generally, affine and non-affine stochastic Volterra equations (SVEs) of the form:
\begin{eqnarray}\label{eq:sve}
X_t = X_0 + \int_0^t K(t,s) b(X_s)\,\mathrm{d}s + \int_0^t K(t,s) \sigma(X_s)\,\mathrm{d}W_s
\end{eqnarray}
where \(b \colon  \mathbb{R} \to \mathbb{R}\), have been studied extensively. Existence, uniqueness and regularity results are provided in \citet{abijaber2019affine}, where applications to rough volatility and term-structure modelling are developed.\\

The motivation for studying stochastic Volterra processes arises from their ability to capture memory, roughness and long-range dependence (see \citealt{mandelbrot1968fractional, coutin2000rough}), which are ubiquitous features in natural and financial systems. Traditional Markovian stochastic models, such as the classical Black-Scholes or Heston models, fail to reproduce empirical stylized facts such as volatility clustering, roughness of realized volatility and persistence of correlations (see \citealt{cont2001empirical, gatheral2018volatility}). Stochastic Volterra processes, through their kernel-based integral representation, naturally incorporate these phenomena by allowing the current state to depend on the entire history of the driving noise \citep{coutin2000rough,biagini2008stochastic}. In particular, the representation of fractional Brownian motion as a Volterra integral highlights how singular kernels can generate rough paths with Hurst exponent $H < 1/2$, providing a mathematically rigorous foundation for rough volatility models \citep{mandelbrot1968fractional, bayer2016pricing}. Moreover, the flexibility of Volterra processes in multidimensional and non-linear settings enables modelling of cross-asset dependence, multifractality and other features observed in financial markets. This combination of theoretical richness and empirical relevance has driven extensive research into stochastic Volterra processes in recent years.

\subsection{Simulation of Volterra processes: standard approaches}

From a numerical perspective, Volterra kernels pose challenges for simulation. Classical schemes based on discretization incur high computational costs, motivating fast methods such as hybrid schemes \citep{bennedsen2017hybrid}, Markovian approximations \citep{abijaber2019multifactor,zhu2021rough,bayer2023markovian} and spectral or Fourier-based approaches \citep{harms2019strong}. These methods aim to preserve the covariance structure while achieving tractability.\\
%In the sequel, we are interested in Gaussian Volterra process and expand the multiple methods to simulate their sample paths advocating their benefits and limitations. We start from the standard Cholesky decomposition of the covariance matrix, then we evoke the Euler discretization of the dynamics of Eq.~\eqref{eq:volterradynamic}. Finally, we dedicate the last section to develop the simulation method based on Random Fourier Features emphasizing their advantages with respect to the previous two methods.

\subsubsection{The Gaussian Volterra case: standard covariance matrix factorization}

Consider first the simpler case of Gaussian Volterra processes \citep{mishura2022gaussian,dinunno2022volterra}, defined by:
\begin{eqnarray}
X_t = X_0 + \int_0^t K(t,s) b(s)\,\mathrm{d}s + \int_0^t K(t,s) \sigma(s)\,\mathrm{d}W_s
\end{eqnarray}
where \(\sigma \colon \mathbb{R}_+  \to \mathbb{R}\).
We are interested in the case where $K:\mathcal{D} \to \mathbb{R}$ is a measurable kernel function and $\mathcal{D}\subset \R_+^2$  as for instance in \citet{biagini2008stochastic}. Since $b$ and $\sigma$ do not depend on $X_s$ here, unlike the more general SVE~\eqref{eq:sve}, the resulting process $X$ has a Gaussian distribution. Being Gaussian, $X$ is fully characterized by its mean and covariance:
\begin{eqnarray}
    \begin{cases}\label{eq:gaussian_volterra}
    \mathbb{E}[X_t] =  \mathbb{E}[X_0]+ \int_0^t K(t,u) b(u)\,\mathrm{d}u\\
    \Cov(X_t,X_s) = \int_0^{t \wedge s} K(t,u)K(s,u)\sigma(u)^2\, \mathrm{d}u.
\end{cases}
\end{eqnarray}
Gaussian Volterra processes have become central in quantitative finance, particularly in rough volatility modelling because they lead to tractable analytical computations of covariance and correlation structures as exploited in \citet{bayer2016pricing, gatheral2018volatility}. Their Gaussianity ensures that linear functionals, such as option pricing kernels, remain analytically or numerically manageable, while the kernel choice provides flexibility to calibrate to market data or introduce desired memory properties.\\
Many methodologies exist to simulate Gaussian processes \citep{liu2019advances}, provided the mean and covariance~\eqref{eq:gaussian_volterra} can be computed analytically. 
The most standard simulation approach relies on discretizing the time interval, constructing the associated covariance matrix, factorizing it using the Cholesky decomposition and multiplying it by a vector of i.i.d. Gaussian simulations. This method is both straightforward and exact, making it attractive for small to moderate grid sizes and for benchmarking purposes. However, the Cholesky-based simulation suffers from significant computational and practical disadvantages. First, its computational complexity scales cubically in the number of time steps, $\mathcal{O}(N^3)$, while memory requirements scale as $\mathcal{O}(N^2)$, rendering it infeasible for fine discretizations or long time horizons commonly required in applications such as Monte Carlo pricing or calibration. Second, the method lacks flexibility in adaptive or non-uniform grids, since any change in the time discretization requires recomputation of the full covariance matrix and its factorization. Third, numerical stability issues may arise when the covariance matrix becomes ill-conditioned, as is often the case for kernels with strong singularities near the diagonal, such as those appearing in fractional or rough Volterra models. These limitations have motivated the development of alternative simulation techniques, including hybrid schemes, convolution-based methods and Markovian approximations, which trade exactness for substantial gains in efficiency while maintaining controlled accuracy. 

\subsubsection{Euler discretization scheme }
\label{sec:eulerschemeSVE}
Another common numerical approach for simulating stochastic Volterra  differential equations of the form of Eq.~\eqref{eq:volterradynamic}
is based on an Euler-type discretization of the stochastic integral, where the process is evaluated on a time grid and the integral is approximated by a weighted sum of Brownian increments. In this scheme, the kernel $K(t,s)$ is discretized explicitly and the diffusion coefficient $\sigma(X_s)$ is frozen at previous grid points, leading to a natural extension of the Euler--Maruyama method to non-Markovian settings. This can be formulated as follows:\\
\RestyleAlgo{ruled}

\SetKwComment{Comment}{/* }{ */}
\SetKwFor{For}{\textbf{\color{black}for}}{\textbf{\color{black}do}}{\textbf{\color{black}end}}

\begin{algorithm}[hbt!]
\caption{Euler scheme of Volterra process of Eq.~\eqref{eq:volterradynamic}}
\label{alg:slow_euler}

\textbf{\color{black}Input:} Time grid $t_{0}=0<t_{1}<t_{2}<\cdots<t_{N}=T$;\newline
\textbf{\color{black}Input:} $N$ i.i.d. standard Gaussian random variables
$G_{1}, G_{2}, \ldots, G_{N}$\;

\textbf{\color{black}Output:} One simulated path $(\tilde{X}_{t_n})_{n=0}^{N}$\;

$\tilde{X}_{t_{0}} \gets x_{0}$\;

\For{$t_{n}=t_{0}, t_{1}, \ldots, t_{N-1}$}{
    $\tilde{X}_{t_{0},t_{n}} \gets x_{0}$\;
    
    \For{$t_{i}=t_{0}, t_{1}, \ldots, t_{n-1}$}{
        $\tilde{X}_{t_{i+1},t_{n}}
        \gets
        \tilde{X}_{t_{i},t_{n}}
        +
        K(t_{n},t_{i})\,\sigma(t_{i},\tilde{X}_{t_{i}})
        \sqrt{t_{i+1}-t_{i}}\,G_{i+1}$\;
    }
    
    $\tilde{X}_{t_{n}} \gets \tilde{X}_{t_{n},t_{n}}$\;
}

\end{algorithm}

The main advantage of this approach lies in its simplicity and broad applicability, as it accommodates general kernels and nonlinear diffusion coefficients without requiring Gaussianity or explicit covariance structures. Nevertheless, the Euler scheme for Volterra equations exhibits several important disadvantages. First, the method typically suffers from low strong convergence rates, which are further deteriorated by the presence of kernel singularities near the diagonal, as encountered in fractional or rough volatility models. In such cases, convergence rates may be significantly slower than the classical $1/2$ rate of standard SDEs. Second, the scheme is computationally expensive, as each time step requires summation over all past increments, leading to $\mathcal{O}(N^2)$ complexity for a single path simulation, which becomes prohibitive for large-scale Monte Carlo applications. Third, error propagation is more severe than in Markovian systems, since inaccuracies introduced at early time steps persist and influence the entire future trajectory through the memory structure of the Volterra kernel. Moreover, stability issues may arise when $\sigma$ is highly nonlinear or only locally Lipschitz, potentially leading to bias or divergence unless very fine time discretizations are employed. These limitations have spurred the development of improved schemes, such as higher-order discretizations, truncated memory methods and hybrid or multilevel approaches, which aim to balance computational efficiency with accuracy while preserving the essential non-Markovian features of the Volterra dynamics.

\subsection{Random Fourier Features representation}
\label{subsec:rff_representation_sfbm}

In the context of SVEs, RFFs allow the kernel-driven memory term to be expressed as a finite sum of stochastic integrals with random weights, providing a tractable and accurate Monte Carlo approximation. Here, we incorporate the RFF in the discretization scheme of stochastic Volterra equations. After exposing the classical Euler scheme in Section \ref{sec:eulerschemeSVE} and evoking its limitation particularly when it comes to generate long sample paths, the focus is placed in the acceleration scheme with RFF as well as the associated estimates. For that sake, we start by presenting the Random Fourier Features representation of the kernel formulated as a proper Monte Carlo estimate sampled from the associated spectral density, then we present the sampling procedure used in the upcoming theoretical error bounds for an arbitrary kernel as well as numerical experiments.
\\\\
Random Fourier features (RFF), introduced in \citet{rahimi2007random}, provide a scalable approach to approximate shift-invariant kernels by representing them as expectations over randomized feature maps. By sampling a finite number of random frequencies, the kernel can be approximated efficiently, reducing the computational complexity of kernel methods from quadratic to linear in the number of data points. Subsequent works have extended this idea, including the Fastfood algorithm for faster computation \citep{le2013fastfood} and rigorous error bounds for Monte Carlo approximations \citep{sutherland2015error}. RFF have been applied for Gaussian processes \citep{lazaro2010sparse} and signal processing \citep{vedaldi2012efficient}, enabling low-rank approximations and efficient simulation of high-dimensional stochastic systems.\\\\
In RFF simulations, the spectral density (for which we have derived the main properties in Section~\ref{subsec:fourier_transform} for the S-fBM kernel) plays a central role by determining how the random frequencies used to approximate a shift-invariant kernel are sampled. By Bochner’s theorem, any continuous, positive-definite, shift-invariant kernel can be expressed as the Fourier transform of a nonnegative measure, whose density is the kernel’s spectral density. Sampling frequencies from this spectral density ensures that the expected inner product of the resulting random features matches the target kernel function. Consequently, the shape of the spectral density directly controls which frequency components are emphasized in the approximation, influencing smoothness, length-scale behavior and approximation accuracy of the simulated kernel. A correct choice and sampling of the spectral density are therefore essential for the Random Fourier Features map to provide a consistent kernel approximation.\\\\
Quantitatively speaking, consider any positive definite kernel \(K\) extended to the entire real line symmetrically: for all \(u \ge 0\), we set \(K(-u) = K(u)\). This ensures that \(K\) remains well-defined and symmetric. Under these assumptions, it is known that the kernel admits the following RFF representation (see \citealt{rahimi2007random}):
\begin{eqnarray}
K(u) = K(0) \, \mathbb{E}[\cos(\eta u)],
\label{eq:random_fourier_features_expression}
\end{eqnarray}
where the density of the random variable $\eta$ is precisely the spectral density of $\frac{K(.)}{K(0)} $. The latter quantity can be approximated
by Monte Carlo. The classical approach would be to consider $\eta_{1}$, $\eta_{2}$, $\ldots$, $\eta_{M}$
be $M$ independent and identically distributed simulations of $\eta$,
independent of $W$. From equation~\eqref{eq:random_fourier_features_expression},
a Monte Carlo approximation of the kernel $K$ is given by:
\begin{eqnarray}
\label{eq:MCrandom_fourier_features_expression}
    \Hat{K}_M(u)=\frac{K(0)}{M}\sum_{m=1}^{M}\cos(\eta_{m}u)
\end{eqnarray}
The aim of the upcoming section is to detail the sampling algorithm from the spectral density, exposing theoretical sampling guarantees of the quantity of Eq.~\eqref{eq:MCrandom_fourier_features_expression}. In particular, in Section \ref{subsec:error_bounds_rff_sfbm}, we explicit Monte Carlo error bounds supposing that the sampled random features are not independent but having a Markovian structure (HMC method). This is favorable in cases where the density from which the sampling is performed is known. Furthermore, implementation details are given in the numerical section (see Section \ref{sec:numerical}).

\subsection{Sampling from the spectral density}
\label{subsec:sampling_algo_hmc}

As long as the spectral density of the underlying kernel is available in closed form up to an unknown multiplicative constant, sampling can be performed using either acceptance–rejection (AR) methods or Markov chain Monte Carlo (MCMC) techniques. Although AR sampling is exact and non-iterative, its efficiency relies critically on the construction of a proposal density that uniformly dominates the target density and on the choice of a tight upper-bounding constant. In moderate to high dimensions, these requirements can result in vanishing acceptance probabilities, making the AR approach impractical. MCMC methods, by contrast, are explicitly designed to operate with unnormalized target densities. Within this class, Hamiltonian Monte Carlo (HMC) offers significant advantages over random-walk Metropolis–Hastings (MH). While MH proposals are typically local and suffer from diffusive random-walk behavior—resulting in slow mixing and high autocorrelation particularly in high dimension, HMC exploits gradient information of the log-target density to construct coherent, long-range proposals via approximate Hamiltonian dynamics. This mechanism preserves high acceptance rates while enabling efficient exploration along level sets of the target distribution. As a result, HMC achieves substantially higher effective sample sizes per unit computational cost than both AR and standard MH methods, making it particularly well suited for scalable Monte Carlo when only the unnormalized target density is available.
\subsubsection{Hamiltonian Monte Carlo}
\label{sec:HMC}
HMC was originally introduced in \citet{duane1987hybrid} as a sampling method inspired by Hamiltonian dynamics for lattice field theory. The key idea is to augment the target distribution with auxiliary momentum variables and to simulate a deterministic Hamiltonian flow that preserves the invariant measure, thereby producing distant proposals with high acceptance probability. The practical implementation relies on a numerical integrator, typically the leapfrog scheme, which ensures approximate energy conservation and detailed balance. This foundational work laid the theoretical and algorithmic basis for modern variants of HMC widely used in Bayesian inference.\\

We start by recalling the HMC algorithm. First, we fix some notations. Let \( \pi(x) \) be the target spectral probability density supported in  \(\mathbb{R}^d\). We define the potential energy as:
\begin{eqnarray}
    \mathcal{U}(x) = -\log \pi(x)
\end{eqnarray}
and introduce auxiliary momentum \(p\in\mathbb{R}^d\) with kinetic energy:
\begin{eqnarray}
    \mathcal{E}(p) = \tfrac{1}{2} p^\top M^{-1} p
\end{eqnarray}
where \(M\) is a symmetric positive-definite mass matrix (often \(M = I\)). The Hamiltonian is:
\begin{eqnarray}
    H(x,p) = \mathcal{U}(x) + \mathcal{E}(p)
\end{eqnarray}
Its dynamics follow Hamilton's equations:
\begin{eqnarray}
    \begin{cases}
        \frac{dx}{dt} = \nabla_p H(x,p) \\
        \frac{dp}{dt} = -\nabla_x H(x,p) 
    \end{cases}
\end{eqnarray}
By definition, one has that:
\begin{eqnarray}
    \begin{cases}
         \nabla_p H(x,p) = M^{-1} p\\
        \nabla_x H(x,p) = \nabla \mathcal{U}(x)
    \end{cases},\nonumber
\end{eqnarray}
which means that the Hamiltonian dynamics can be reformulated as:
\begin{eqnarray}
    \begin{cases}
        \frac{dx}{dt} =  M^{-1} p\\
        \frac{dp}{dt} = -\nabla \mathcal{U}(x)
    \end{cases}
\end{eqnarray}
In practice, the Hamiltonian dynamics is discretized using the leapfrog integrator with step size \(\varepsilon>0\) and \(L\) steps formulated as:
\begin{eqnarray}
    \begin{cases}
        p\!\left(t+\tfrac{\varepsilon}{2}\right) = p(t) - \tfrac{\varepsilon}{2}\,\nabla U\big(x(t)\big),\\
x(t+\varepsilon) = x(t) + \varepsilon\, M^{-1} p\!\left(t+\tfrac{\varepsilon}{2}\right)\\
p(t+\varepsilon) = p\!\left(t+\tfrac{\varepsilon}{2}\right) - \tfrac{\varepsilon}{2}\,\nabla \mathcal{U}\big(x(t+\varepsilon)\big)
    \end{cases}
\end{eqnarray}
This is framed in Algorithm~\ref{alg:hmc}.

\RestyleAlgo{ruled}
\SetKwComment{Comment}{/* }{ */}
\SetKwFor{For}{\textbf{\color{black}for}}{\textbf{\color{black}do}}{\textbf{\color{black}end}}
\SetKwFor{If}{\textbf{\color{black}if}}{\textbf{\color{black}then}}
{\textbf{\color{black}end}}

\begin{algorithm}[hbt!]
\caption{Hamiltonian Monte Carlo (HMC)}
\label{alg:hmc}

\textbf{\color{black}Input:} Initial position $x^{(0)}$;\newline
\textbf{\color{black}Input:} Step size $\varepsilon$, number of leapfrog steps $L$;\newline
\textbf{\color{black}Input:} Mass matrix $M$, number of iterations $N$\;

\textbf{\color{black}Output:} Samples $\{x^{(n)}\}_{n=1}^{N}$\;

\For{$n = 0,1,\ldots,N-1$}{
    Sample momentum $p^{\ast} \sim \mathcal{N}(0,M)$\;
    
    $x \gets x^{(n)}$\;
    $p \gets p^{\ast}$\;
    
    \Comment{\textcolor{black}{Leapfrog integration of Hamiltonian dynamics}}
    $p \gets p - \tfrac{\varepsilon}{2}\,\nabla \mathcal{U}(x)$\;
    
    \For{$\ell = 1,2,\ldots,L$}{
        $x \gets x + \varepsilon\,M^{-1}p$\;
        \If{$\ell < L$}{
            $p \gets p - \varepsilon\,\nabla \mathcal{U}(x)$\;
        }
    }
    
    $p \gets p - \tfrac{\varepsilon}{2}\,\nabla \mathcal{U}(x)$\;
    $p \gets -p$ \Comment{\textcolor{black}{Momentum flip for reversibility}}\;
    
    $\alpha \gets \min\!\left(1,\;
    \exp\!\big(-H(x,p) + H(x^{(n)},p^{\ast})\big)\right)$\;
    
    With probability $\alpha$, set $x^{(n+1)} \gets x$;\newline
    otherwise set $x^{(n+1)} \gets x^{(n)}$\;
}
\end{algorithm}

Appendix~\ref{sec:hmc}  analyzes the HMC dynamics by viewing it as a continuous-time Markov process and studying its long-term behavior. First, it derives the infinitesimal generator of the process. The generator naturally splits into two components: a deterministic part corresponding to the Hamiltonian flow, which moves the state smoothly along trajectories that conserve energy and a stochastic jump part corresponding to momentum updates governed by an acceptance rule. Next, the appendix shows that the Hamiltonian itself can be used to construct a Lyapunov function, meaning that the process has a natural tendency to return toward regions of moderate energy. This establishes a drift condition for the generator, which implies stability of the dynamics. Finally, these stability properties are used to prove geometric ergodicity of the Markov process, leading to exponential decay of correlations between samples.

\subsubsection{Discretization scheme with RFF approximation}
Herein, we are interested in the discretization scheme of Eq.~\eqref{eq:volterradynamic} using Random Fourier Features. For any positive definite, symmetric kernel, Eq.~\eqref{eq:random_fourier_features_expression} holds . For many kernels of practical interest, the distribution of the random projection \(\eta\) can be computed explicitly, allowing for straightforward simulation, as is the case of the S-fBM kernel developed in Section~\ref{subsec:rff_representation_sfbm}.\\

Substituting \eqref{eq:random_fourier_features_expression} into \eqref{eq:volterradynamic}, we can express \(X\) as:
\begin{eqnarray}
    X_t = X_0 + K(0) \int_0^t \mathbb{E}_\eta[\cos(\eta(t-s))] \, \sigma(s, X_s) \, dW_s,
\end{eqnarray}
where \(\mathbb{E}_\eta[\cdot] = \mathbb{E}[\,\cdot \mid W]\) denotes the expectation with respect to the distribution of \(\eta\), conditionally on the Brownian motion \(W\). Applying the standard cosine decomposition formula \(\cos(a-b) = \cos(a)\cos(b) + \sin(a)\sin(b)\) allows the integrand to be separated as:
\begin{eqnarray}
    X_t = X_0 + \mathbb{E}_\eta \Bigg[ K(0) \cos(\eta t) \int_0^t \cos(\eta s) \, \sigma(s, X_s) \, dW_s + K(0) \sin(\eta t) \int_0^t \sin(\eta s) \, \sigma(s, X_s) \, dW_s \Bigg].
\end{eqnarray}
Using Eq.~\eqref{eq:MCrandom_fourier_features_expression}, the SVE \eqref{eq:volterradynamic} can be approximated by the process \(X^M\) with the following dynamics:
\begin{eqnarray}
X^M_t = X_0 + \int_0^t \frac{K(0)}{M} \sum_{m=1}^{M} \cos\big(\eta_m (t-s)\big) \, \sigma(s, X^M_s) \, dW_s, \quad t \ge 0.
\label{eq:sve_rff}
\end{eqnarray}
Applying the cosine decomposition formula again, this representation can be rewritten in a more computationally convenient form:
\begin{align}
X^M_t &= X_0 + \int_0^t \frac{K(0)}{M} \sum_{m=1}^{M} \big( \cos(\eta_m t)\cos(\eta_m s) + \sin(\eta_m t)\sin(\eta_m s) \big) \, \sigma(s, X^M_s) \, dW_s \nonumber\\
&= X_0 + \frac{K(0)}{M} \sum_{m=1}^{M} \Bigg[ \cos(\eta_m t) \int_0^t \cos(\eta_m s) \, \sigma(s, X^M_s) \, dW_s + \sin(\eta_m t) \int_0^t \sin(\eta_m s) \, \sigma(s, X^M_s) \, dW_s \Bigg].
\label{eq:sve_rff_mc}
\end{align}
The process $X^M$ can be simulated efficiently, as described in Algorithm~\ref{algo:fast_euler}.
\RestyleAlgo{ruled}
\SetKwComment{Comment}{/* }{ */}
\SetKwFor{For}{\textbf{\color{black}for}}{\textbf{\color{black}do}}{\textbf{\color{black}end}}
\begin{algorithm}[hbt!]
\caption{RFF scheme for the SVE of \eqref{eq:sve_rff_mc}}
\label{algo:fast_euler}

\textbf{\color{black}Input:} Time grid $t_{0}=0<t_{1}<\cdots<t_{N}=T$\;

\textbf{\color{black}Input:} $N$ i.i.d. standard Gaussian random variables $G_{1}, \dots, G_{N}$\;

\textbf{\color{black}Input:} $M$ i.i.d. random variables $\eta_{1}, \dots, \eta_{M}$\;

\textbf{\color{black}Output:} One simulated path $(\tilde{X}_{t_n}^{M})_{n=0}^{N}$\;

Initialize $\tilde{X}_{t_0}^{M} \gets 0$, $\tilde{C}_{t_0}^{m} \gets 0$ and $\tilde{S}_{t_0}^{m} \gets 0$ for $m=1,\dots,M$\;

\For{$n=0,1,\dots,N-1$}{
    \For{$m=1,2,\dots,M$}{
        $\tilde{C}_{t_{n+1}}^{m} \gets \tilde{C}_{t_n}^{m} + \cos(\eta_m t_n) \, \sigma(t_n, \tilde{X}_{t_n}^{M}) \sqrt{t_{n+1}-t_n} \, G_{n+1}$\;\\
        $\tilde{S}_{t_{n+1}}^{m} \gets \tilde{S}_{t_n}^{m} + \sin(\eta_m t_n) \, \sigma(t_n, \tilde{X}_{t_n}^{M}) \sqrt{t_{n+1}-t_n} \, G_{n+1}$\;
    }
    
    $\tilde{X}_{t_{n+1}}^{M} \gets X_0 + \frac{K(0)}{M} \sum_{m=1}^{M} \left( \cos(\eta_m t_{n+1}) \tilde{C}_{t_{n+1}}^{m} + \sin(\eta_m t_{n+1}) \tilde{S}_{t_{n+1}}^{m} \right)$\;
}

\end{algorithm}

The formulation~\eqref{eq:sve_rff_mc} allows for a tractable and efficient Monte Carlo approximation of the SVE, where the high-dimensional memory effect induced by \(K\) is captured via a finite sum of random Fourier features. Furthermore, it is competitive against the two previous methods (covariance factorization and Euler scheme) with a computational complexity of $\mathcal{O}(N\times M)$, where $M$ is the number of spectral points, typically set to a few dozen points in practice \citep{lazaro2010sparse,rudi2017generalization,delbridge2020randomly}. In the next sections, we investigate the theoretical guarantees and error bounds of the RFF-based numerical scheme.
\subsubsection{Quantitative estimates}
This part is devoted to quantitative estimates for the time-discretized Volterra scheme with its RFF approximation for an arbitrary kernel. We first introduce an Euler-type discretization associated with a generic kernel \(K\) on an arbitrary time grid and then define its random-features counterpart obtained by replacing the kernel with a finite-dimensional spectral approximation. Under standard growth, Lipschitz continuity and temporal regularity assumptions on the volatility coefficient, we derive moment bounds, H\"older-type regularity estimates and sharp \(\mathcal{L}^p\) error bounds that quantify both the time-discretization error and the additional approximation error induced by the random features expansion. For that sake, let us consider for each $N \ge 1$ a discrete grid:
\[
\pi_N = \{ 0 = t_0^N < t_1^N < \cdots < t_N^N = T_f \}
\]
and denote
\[
\delta_N := \max_{0 \le k \le N-1} (t_{k+1}^N - t_k^N), 
\qquad 
\zeta_N(s) := t_k^N, \quad \text{for } s \in [t_k^N, t_{k+1}^N), \; k \ge 0.
\]
We consider the underlying scheme process $X_N$ defined as:
\begin{equation}\label{eq:XN}
X_N(t) = X_0
+ \int_0^t K\big(t-\zeta_N(s)\big)\,\sigma\big(\zeta_N(s),X_N(\zeta_N(s))\big)\,dW_s .
\end{equation}
In practice, we only simulate the values of $X_N$ pointwisely on the discrete-time grid: 
\[
\pi_N = \{ t_k, \; k = 0, 1, \ldots, N \}.
\]
This can be achieved by generating the Brownian motion increments: 
\[
\Delta W_{k+1} := W_{t_{k+1}} - W_{t_k}, \qquad k = 0, \ldots, N-1.
\]
Let $\Delta t_{k+1} := t_{k+1} - t_k$ and initialize $X_N(t_0) := X_0$. 
Then, the recursive scheme is given by:
\[
X_N(t_{k+1}) = X_0
+  \sum_{i=0}^{k} K(t_{k+1}- t_i)\, \sigma(t_i, X_N(t_i))\, \Delta W_{i+1}.
\]
We introduce the scheme with random Fourier features scheme denoted $X^M_N$ and defined as:
\begin{equation}\label{eq:XMN}
X^M_N(t) = X^M_0
+ \int_0^t \Hat{K}_M\big(t-\zeta_N(s)\big)\,\sigma\big(\zeta_N(s),X^M_N(\zeta_N(s))\big)\,dW_s .
\end{equation}
In the sequel, we consider the following assumptions on $\sigma$:
\begin{eqnarray}
\label{eq:hypothesis1}
\bm{\bm{\mathcal{H}_{\mathrm{1}}}}:\quad 
 \forall (s,x)\in\mathbb{R}_+\times \mathbb{R},\tab |\sigma(s, x)| \le C_1(1 + |x|)
\end{eqnarray}
\begin{eqnarray}
\label{eq:hypothesis2}
\bm{\bm{\mathcal{H}_{\mathrm{2}}}}:\quad 
 \forall (t,x,y)\in \R_+\times \R^2,\tab \left|\sigma(t, x) - \sigma(t, y)\right|
\le C_2\,|x - y|,\tab
\end{eqnarray}
\begin{eqnarray}
\label{eq:hypothesis3}
\bm{\bm{\mathcal{H}_{\mathrm{3}}}}:\quad 
 \forall (t,s,x)\in \R_+^2\times \R,\tab \left|\sigma(t, x) - \sigma(s, x)\right|
\le C_3\,|t - s|^{\alpha \wedge 1} \,(1 + |x|),\tab \alpha>0
\end{eqnarray}
which correspond respectively to the linear growth and the linear growth in space and H\"older in time conditions.\\\\
We establish moment estimates as well as the H\"older regularity of the schemes~\eqref{eq:XN} and~\eqref{eq:XMN} in the following theorem, proved in Appendix~\ref{subsec:proof_rff_estimate_moment}.
\begin{thm}
\label{thm:quantRFFestimatemoment}
If $\sigma$ satisfies Assumption~\hyperref[eq:hypothesis1]{\(\bm{\mathcal{H}_{\mathrm{1}}}\)} and $K$ satisfies the integrability conditions of Assumptions I and II in \cite{decreusefond2002regularity}, the following claims hold:
\begin{enumerate}
    \item For any compact $\mathcal{K}\subset \R_+$, there exist positive constants $C_1(\mathcal{K})$ and $C_2(\mathcal{K})$ such that:
\begin{eqnarray}
\forall\, p > 0,\ \forall\, M \in \mathbb{N}, \qquad
\begin{cases}
\underset{t \in \mathcal{K}}{\sup}\, \mathbb{E}\!\left[\,|X_N(t)|^p\,\right]
\le C_1(\mathcal{K})\,\Big(\mathbb{E}\!\left[\,|X_0|^p\,\right] + |K(0)|^p\Big), \\[0.6em]
\underset{t \in \mathcal{K}}{\sup}\, \mathbb{E}\!\left[\,|X^M_N(t)|^p\,\right]
\le C_2(\mathcal{K})\,\Big(\mathbb{E}\!\left[\,|X^M_0|^p\,\right] + |K(0)|^p\Big).
\end{cases}
\end{eqnarray}

 \item  There exist positive constants $C_p$ such that:
\begin{eqnarray}
  \forall (t,s)\in \mathcal{K}^2,\tab \forall\, p > 0, \forall\, M\in \N,  \qquad 
\mathbb{E}[|X^M_N(t) - X^M_N(s)|^p] 
\le C_p \Big( |t-s|^{p/2} +1 \Big)
\end{eqnarray}
\end{enumerate}
\end{thm}

Theorem~\ref{thm:quantRFFestimatemoment} provides fundamental a priori estimates for both the original truncated Volterra process \(X_N\) and its random Fourier features approximation \(X_N^M\). 
The first part establishes uniform moment bounds on compact time intervals, showing that the moments of the processes remain controlled independently of the approximation parameter \(M\). 
In particular, these estimates guarantee that neither the truncation procedure nor the random features approximation introduces any explosion of moments over finite horizons. The second part gives quantitative moment estimates for the increments of the approximated process \(X_N^M\). 
The obtained bound exhibits the characteristic Brownian scaling which is consistent with the diffusive nature of the stochastic integral term whereas the additional term ($+1$) can be interpreted as a moment bound.
Altogether, these estimates constitute the probabilistic foundation required for the convergence analysis of the random Fourier features approximation scheme.\\

The following results quantify, for an arbitrary kernel of the form of Section \ref{subsec:rff_representation_sfbm}, the propagation of errors using the scheme of Eq.~\eqref{eq:XMN}. Theorem~\ref{thm:eulerXNXNMRFFerrorbound}, proved in Appendix~\ref{subsec:proof_euler_xnxnm_rff_error_bound}, provides an $\mathcal{L}^p$ bound on the discrepancy between the Euler scheme $X_N$ and its random features approximation $X_N^M$.

\begin{thm}
\label{thm:eulerXNXNMRFFerrorbound}
If $\sigma$ satisfies \hyperref[eq:hypothesis1]{\(\bm{\mathcal{H}_{\mathrm{1}}}\)} and    \hyperref[eq:hypothesis2]{\(\bm{\mathcal{H}_{\mathrm{2}}}\)} and $K$ the integrability conditions of Assumptions I and II in \cite{decreusefond2002regularity}, then for any arbitrary compact $\mathcal{K}\subset \R_+$, there exists a positive constant $C_\mathcal{K}^p>0$ such that:
\begin{eqnarray}
\forall\, M\in \N, \forall\, t\in \mathcal{K},  \qquad \sup_{s \in [0, t]} \mathbb{E}\left[|X_N(s) - X_N^M(s)|^p\right]
\le 
\begin{cases}
   \frac{C_\mathcal{K}^p}{M^{p-1}} e^{C_\mathcal{K}^p t},\tab \text{if $1 \leq p<2$}\\
C_\mathcal{K}^p \Bigg( \frac{2}{M^{p-1}} + \Big(\frac{4^p}{M}\Big)^{\frac{1}{2}} \Bigg) e^{C_\mathcal{K}^p t},\tab\text{if $p \geq 2$}
\end{cases}
\end{eqnarray}
\end{thm}
The above strong error bound guarantees uniform-in-time control (on compact subsets of $\R_+$) of the RFF scheme with respect to its Euler counterpart. The bound is explicit in the discretization parameter $M$ and highlights two different convergence regimes depending on whether $p<2$ or $p\geq 2$. In particular, for $1\leq p<2$, the error behaves essentially like $M^{-(p-1)}$, while for $p\geq 2$ an additional term of order $M^{-1/2}$ appears. The exponential factor $e^{C_\mathcal{K}^p t}$ arises from a Gr\"onwall-type argument and captures the accumulation of local discretization errors over time. \\
One can propagate one step further to explicit the error bound between the true SVE solution and the RFF scheme.

\begin{prop}
\label{prop:eulerXXNMRFFerrorbound}
If $\sigma$ satisfies \hyperref[eq:hypothesis1]{\(\bm{\mathcal{H}_{\mathrm{1}}}\)}, \hyperref[eq:hypothesis2]{\(\bm{\mathcal{H}_{\mathrm{2}}}\)} and \hyperref[eq:hypothesis3]{\(\bm{\mathcal{H}_{\mathrm{3}}}\)} and $K$ the integrability conditions of Assumptions I and II in \cite{decreusefond2002regularity},
then for any $p>0$  and $\varepsilon \in \left(\frac{1}{p}, \alpha \wedge 1\right)$ there exist positive constants $K^1_p$ and $K^2_p$ such that:
\begin{eqnarray*}
\sup_{s\in[0,t]} \mathbb{E}\left[|X_s-X_N^M(s)|^p\right]
\le 
    \begin{cases}
        K^1_p\Bigg(
 \,\Big(1 + \mathbb{E}\left[|X_0|^p\right] \Big)\, \delta_N^{\,p(\alpha \wedge 1) - \varepsilon}
+
\frac{e^{C_K^p t}}{M^{p-1}} 
\Bigg),\tab \text{if $1 \leq p<2$}\\

K^2_p \Bigg(
\,\Big(1 + \mathbb{E}\left[|X_0|^p\right] \Big)\, \delta_N^{\,p(\alpha \wedge 1) - \varepsilon}+ \Bigg( \frac{2}{M^{p-1}} + \Big(\frac{4^p}{M}\Big)^{\frac{1}{2}} \Bigg) e^{C_K^p t},\tab\text{if $p \geq 2$}
    \end{cases}
\end{eqnarray*}
\end{prop}

Indeed, the bound decomposes naturally into two contributions: the first term, of order $\delta_N^{\,p(\alpha \wedge 1)-\varepsilon}$, corresponds to the approximation error induced by the time discretization, while the second term reflects the error associated with the truncation parameter $M$, as already identified in Theorem~\ref{thm:eulerXNXNMRFFerrorbound}. The estimate holds uniformly in time as in the previous result, and two regimes appear depending on whether $1 \leq p<2$ or $p\geq 2$, with an additional $M^{-1/2}$ contribution for higher moments. This proposition therefore quantifies the combined impact of the RFF sampling and the Euler time discretization, yielding explicit convergence rates for both procedures.\\

In Section \ref{subsec:error_bounds_rff_sfbm}, we will leverage the Markovian structure of the HMC sampling to derive explicit Monte Carlo bounds of the RFF kernel approximation and incorporate it in the strong error bounds of Theorem~\ref{thm:eulerXNXNMRFFerrorbound} and Proposition~\ref{prop:eulerXXNMRFFerrorbound} in the case of the S-fBM kernel to derive sharper and explicit bounds.

\subsection{Benchmarking with the multifactor approximation\label{sec:multifactor}}

\subsubsection{Completely monotonic kernels}

\begin{defn}
Let $K:[0,+\infty)\rightarrow\mathbb{R}$ be a continuous, bounded, shift-invariant, infinitely differentiable kernel function. It is
said to be \textit{completely monotonic} if:
\begin{equation}
(-1)^{n}K^{(n)}(x)\geq0\label{eq:completely_monotone}
\end{equation}
for all $n\in\mathbb{N}$ and $x>0$.
\end{defn}
According to the Hausdorff-Bernstein-Widder
theorem \citep[Theorem~3.9.6]{sasvari2013characteristic}, $K$~is completely monotonic if and only if it is the inverse Laplace transform
of a finite nonnegative measure on $[0,\infty)$. In particular, a
completely monotonic $K$ admits the probabilistic representation:
\begin{eqnarray}
\label{eq:LaplaceFeaturesrepresentation}
K(u)=K(0)\mathbb{E}\left[e^{-\eta u}\right]\ ,\ \forall u\geq0\label{eq:Eexp}
\end{eqnarray}
where the density of the random variable $\eta$ is the inverse Laplace
transform of $K(.)/K(0)$. Such a kernel \eqref{eq:Eexp}
can be well approximated by a sum of exponential terms; this is the
\textit{random Laplace features} methodology \citep{yang2014random}.

\subsubsection{The multifactor method}

The multifactor method for simulating stochastic Volterra processes in the sense of \citet{abijaber2019multifactor} is based on approximating a non-Markovian Volterra dynamics by an explicit finite-dimensional Markovian system that is numerically tractable and converges to the original process as the number of factors increases. Considering a stochastic Volterra equation with kernel $K$ on a filtered probability space $(\Omega,\mathcal{F},\mathbb{F},\mathbb{P})$ typically written as:
\[
X_t = u_0(t) + \int_0^t K(t-s)\, b(X_s)\, ds + \int_0^t K(t-s)\, \sigma(X_s)\, dW_s,
\]
where $b$ and $\sigma$ are drift and diffusion coefficients and $W$ is an $\mathbb{F}$-Brownian motion. For rough volatility models (e.g.\ rough Heston), one considers a singular one-dimensional kernel, also known as power-law kernel, of the form:
\[
K(t) = \frac{t^{\alpha-1}}{\Gamma(\alpha)}, 
\qquad \alpha \in \left(0,\frac{1}{2}\right),
\]
which induces non-Markovian and non-semimartingale behavior that is difficult to simulate directly. The multifactor method approximates the kernel $K$ by a finite sum of exponentials seen as a completely monotone approximation justified by the Hausdorff-Bernstein-Widder
theorem \citep[Theorem~3.9.6]{sasvari2013characteristic}, which is nothing else but the empirical counterpart of Eq.~\eqref{eq:LaplaceFeaturesrepresentation}. The approximation takes the form:
\[
K_n(t) = \sum_{i=1}^n c_i e^{-x_i t}, 
\qquad \text{with weights } c_i>0 \text{ and rates } x_i>0,
\]
chosen so that $K_n \to K$ in $L^2$. Under such an approximation, one defines factor processes $(Y^{n,i}_t)_{i=1,\ldots,n}$ satisfying:
\[
dY^{n,i}_t = -x_i\, Y^{n,i}_t\, dt 
+ b(X^n_t)\, dt 
+ \sigma(X^n_t)\, dW_t,
\qquad Y^{n,i}_0 = 0,
\]
and the multifactor approximation of $X$ as:
\[
X^n_t = u_0^n(t) + \sum_{i=1}^n c_i Y^{n,i}_t,
\]
which in integrated form yields:
\[
X^n_t = u_0^n(t) 
+ \int_0^t K_n(t-s)\, b(X^n_s)\, ds
+ \int_0^t K_n(t-s)\, \sigma(X^n_s)\, dW_s.
\]
Each factor $Y^{n,i}$ is Markovian with mean-reversion rate $x_i$, so that $(X^n, Y^{n,1}, \ldots, Y^{n,n})$ is a finite-dimensional Markov process amenable to standard simulation techniques such as the classical Euler scheme (see Section~\ref{sec:eulerschemeSVE}) . Under mild assumptions on $b$, $\sigma$ and the convergence $K_n \to K$, one proves that $X^n$ converges to $X$ in probability or in $\mathcal{L}^p$ uniformly on compacts as $n \to \infty$, providing an efficient simulation framework for rough volatility and more general non-negative Volterra processes \citep{abijaber2019multifactor,AbiJaberBayerBreneis2024}. Extensions to general completely monotone kernels and further approximation results for stochastic Volterra equations are also developed in \citet{AlfonsiKebaier2024}.

\subsubsection{Positive definiteness VS complete monotonicity}

The key connection between these two properties is given by Schoenberg's
theorem for kernels in $\Phi_{\infty}$ \citep[Theorem~3.8.5]{sasvari2013characteristic}:
a kernel $K$ belongs to $\Phi_{\infty}$ if and only if $K(\sqrt{.})$
is completely monotonic. In other words, $K(\left\Vert \mathbf{.}\right\Vert )$
is positive definite in $\mathbb{R}^{d}$ for all $d\geq1$, if and
only if $K(\sqrt{.})$ is the inverse Laplace transform of a
finite nonnegative measure on $[0,\infty)$ \citep[Theorem~3.9.8]{sasvari2013characteristic}.

In addition, \citet{langrene2025mixture} proved that if $K\in\Phi_{\infty}$,
then $K(\left\Vert \mathbf{.}\right\Vert ^{\alpha})$ is also
positive definite in all $\mathbb{R}^{d}$ for all $\alpha\in(0,1]$.
Similarly, if $K$ is completely monotonic, then $K(.^{\alpha})$
is also completely monotonic for all $\alpha\in(0,1]$. Combining
this with Schoenberg's theorem yields the following equivalence: \textit{$K(\left\Vert \mathbf{.}\right\Vert ^{\alpha})$ is positive
definite in all $\mathbb{R}^{d}$ for all $\alpha\in(0,1]$ if and
only if $K(.^{\frac{\alpha}{2}})$ is completely monotonic on
$[0,\infty)$ for all $\alpha\in(0,1]$.}

This result clearly shows that working with positive definite kernels
is more general than working with completely monotonic kernels: if
a kernel $K$ is completely monotonic, then it is also positive
definite, but the converse is not necessarily true. Intuitively, ``half
of the powers are lost'' when going from positive definite kernels
$K(\left\Vert \mathbf{.}\right\Vert ^{\alpha})$ to completely
monotonic kernels $K(.^{\frac{\alpha}{2}})$. For example, the
univariate kernels $K(u)=e^{-u^{\alpha}}$ (exponential power),
$K(u)=(1+u^{\alpha})^{-\beta}$, $\beta>0$ (generalized Cauchy),
or $K(u)=u^{\frac{\alpha\beta}{2}}\mathcal{K}_{\beta}(\sqrt{2\beta}u^{\frac{\alpha}{2}}$),
$\beta>0$ (generalized Mat\'ern, where $\mathcal{K}_{\beta}$ is
the modified Bessel function), are not completely monotonic when $\alpha\in(1,2]$,
yet are still positive definite for such values of $\alpha$. As a
result, such kernels can be well approximated by a sum of cosine functions
(random Fourier features), but cannot be well approximated by a sum
of exponential functions (random Laplace features).

Finally, if one considers kernel functions in $\Phi_{d}\backslash\Phi_{\infty}$, such functions, with either compact support or taking negative values, can still be well approximated by spectral sampling \citep{emery2006tbsim} but not by a sum of exponential functions, which would necessarily be positive with unbounded support.

In particular, the isotropic S-fBM kernel $K(\mathbf{u})=(1-\left\Vert \mathbf{u}\right\Vert ^{\alpha})\mathbbm{1}_{\{\left\Vert \mathbf{u}\right\Vert \leq1\}}$
belongs to $\Phi_{d}$ only for $d=1$ (when $\alpha\in(0,1]$) and
$d=2$ (when $\alpha\in\left(0,\frac{1}{2}\right)$). As such, like any kernel with compact
support, the S-fBM kernel is not completely monotonic and therefore
cannot be well approximated by a sum of exponentials. Nevertheless,
it can still be well approximated by a sum of cosine functions whenever
the condition for positive definiteness is satisfied ($0\leq\alpha\leq \frac{3-d}{2}$, see Theorem~\ref{thm:positive_definiteness}).

\section{The Log stationary fractional Brownian motion model}

\subsection{Overview}
\label{sec:logsfbm}

Here, we recall the definition of the Log S-fBM model. To that end, we start by briefly recalling its main building block, namely the stationary fractional Brownian motion (S-fBM) model introduced in~\citet{wu2022rough}, highlighting its main construction and properties.

The S-fBM process is built using the so-called time-scale domain:
\begin{equation}
C_{l,T}(t^{*}) := \left\{ (t,h) \,:\, h>l,\; |t-t^{*}| < \tfrac{1}{2}\min(h,T) \right\},
\end{equation}
and defined as:
\begin{equation}
\omega_{H,T}(t) = \mu_H + \int_{C_T(t)} \mathrm{d}G_H,
\end{equation}
where $C_T(t) := C_{0,T}(t)$ and $\mathrm{d}G_H$ is a centered non-homogeneous Gaussian white noise with variance measure
\begin{equation}
\mathbb{E}\!\left( \mathrm{d}G_H(t,h)^2 \right)
= \lambda^2 \frac{h^{2H-2}}{T^{2H}} \, \mathrm{d}t \, \mathrm{d}h .
\end{equation}
The constant $\mu_H$ is chosen such that $\mathbb{E}(e^{\omega_{H,T}(t)})=1$, yielding:
\begin{equation}
\mu_H = -\frac{\nu^2}{4}.
\end{equation}
The resulting process $(\omega_{H,T}(t))_t$ is a stationary Gaussian process with autocovariance:
\begin{eqnarray}
\label{eq:sfbmcovkernel}
C_{\omega}(\tau) =
\begin{cases}
\dfrac{\nu^2}{2} \left( 1 - \left( \dfrac{|\tau|}{T} \right)^{2H} \right), & |\tau| < T, \\[1ex]
0, & \text{otherwise},
\end{cases}
\end{eqnarray}
so that $\Var(\omega_{H,T}(t))=\nu^2/2$. Here, $T$
 plays the role of a correlation limit. The so-called intermittency parameter $\lambda$ is related to $\nu$ and $H$ as follows:
\begin{eqnarray}
\label{eq:nu2intermformula}
\nu^2 = \frac{\lambda^2}{H(1-2H)}, \qquad H \in \left(0,\frac{1}{2}\right).
\end{eqnarray}
Following \citet{bacry2001modelling}, \citet{wu2022rough} define the associated multifractal random measure:
\begin{equation}
M_{H,T}(\mathrm{d}t) = \exp(\omega_{H,T}(t)) \, \mathrm{d}t .
\end{equation}
A key feature of the framework is its continuity at $H=0$. As $H \to 0$, the measure $M_{H,T}$ converges weakly to the multifractal random measure of a Multifractal Random Walk (MRW) denoted $\tilde{M}_T(.)$ as introduced in \citet{bacry2001multifractal, bacry2001modelling, muzy2000modelling, muzy2002multifractal}. In fact, \citet{wu2022rough} showed that:
\begin{equation}
M_{H,T}(\mathrm{d}t) \xrightarrow[H \to 0]{\mathrm{w}} \tilde{M}_T(\mathrm{d}t),
\end{equation}
so that the S-fBM framework naturally unifies fractional ($H>0$) and log-correlated ($H=0$) regimes.\\
Finally, the Log S-fBM model is defined by subordinating a Brownian motion to the multifractal measure:
\begin{equation}
\mathrm{d}X_t = \exp(\omega_{H,T}(t)) \, \mathrm{d}B_t
= \frac{M_{H,T}(\mathrm{d}t)}{\mathrm{d}t} \, \mathrm{d}B_t,
\end{equation}
where $(B_t)_t$ is a standard Brownian motion. This construction provides a flexible stochastic volatility model encompassing both rough and multifractal behaviors observed in empirical data.\\\\
In the upcoming sections, we derive the main ingredients for the RFF representation, namely the spectral density of the S-fBM covariance kernel, the sampling algorithm of the Random Features, and the associated convergence guarantees. Before that, we choose to dedicate the following section to fix the notations for the sake of clarity.

\subsection{Spectral analysis of the S-fBM kernel\label{sec:spectralsfbm}}
This section is devoted to a spectral analysis of the S-fBM kernel, which plays a central role both in characterizing its admissibility as a covariance function and in tailoring the Euler scheme with RFF representation. Since the kernel is radial, compactly supported and shift invariant, its spectral density is obtained as the Fourier transform of the covariance function in $\mathbb{R}^d$ and admits explicit representations in terms of special functions. In order to derive these representations and establish quantitative approximation results, we first introduce the necessary notations. We recall the Fourier transform convention adopted throughout the paper, together with the analytical tools required to handle radial Fourier transforms, including Bessel functions of the first kind and generalized hypergeometric functions. We also introduce suitable truncations of these series expansions, which will be instrumental in the construction of computable approximations of the spectral density and in the derivation of explicit error bounds. Working in arbitrary dimension $d \geq 1$ and focusing on the rough regime $H\in \left(0,\frac{1}{2}\right)$ relevant for applications, we express the S-fBM kernel in a normalized form that highlights its scaling properties and facilitates its spectral decomposition. The section then proceeds by establishing explicit spectral density formulas in both Bessel-series and hypergeometric forms, together with sharp uniform and $\mathcal{L}^1$ tail error estimates for their truncated counterparts, thereby quantifying the trade-off between dimensionality, truncation order and approximation accuracy. These results reveal a pronounced curse of dimensionality in the spectral domain, which is nevertheless mitigated by the exponential rate of convergence of the proposed truncation, a feature that is particularly relevant in high-dimensional and large-correlation-length regimes encountered in practice.\\

For any function $f\in \mathcal{L}^1(\R^d)$, we define its Fourier transform $\Hat{f}$ as:
\begin{equation}\label{eq:fourier}
    \forall \boldsymbol{\omega} \in \R^d,\ \ \hat{f}(\boldsymbol{\omega}):= \frac{1}{\left(2\pi\right)^{d}}\int_{\R^d}f(\mathbf{x})e^{-i\boldsymbol{\omega}^T\mathbf{x}}d\mathbf{x}.
\end{equation}
The Fourier transform $\hat{f}$ is continuous and bounded. If $\hat{f}$ also belongs to $\mathcal{L}^1(\R^d)$, then its inverse Fourier transform is given by
\begin{eqnarray}\label{eq:inverse_fourier}
    \forall \mathbf{x} \in \R^d, \ \  f(\mathbf{x}):= \int_{\R^d}\hat{f}(\boldsymbol{\omega})e^{i\boldsymbol{\omega}^T\mathbf{x}}d\boldsymbol{\omega}
\end{eqnarray}
Let $J_\nu$ denote the Bessel function of the first kind of order $\nu>0$ \citep[10.2.2]{dlmf}:
\begin{eqnarray}
\label{eq:bessel}
J_{\nu}(z)=\sum_{n=0}^{\infty}\frac{(-1)^{n}(z/2)^{\nu+2n}}{n!\Gamma(n+\nu+1)}
\end{eqnarray}
where  $\Gamma(\cdot)$ is the Gamma function \citep[5.2]{dlmf} defined for $\Re(z)>0$ by
\[
\Gamma(z)=\int_{0}^{\infty} t^{z-1} e^{-t}\,dt, \tab z\in \C
\]
For any $(p,q)\in \N^2$, $\left(a_1,\dots,a_p \right),\left(b_1,\dots,b_q\right) \in \R^p\times \R^q$, we introduce the generalized hypergeometric function ${}_pF_q$ \citep[16.2]{dlmf}, defined by:
\begin{eqnarray}
    {}_pF_q\big(a_1,\dots,a_p; b_1,\dots,b_q; z\big)
= \sum_{n=0}^{\infty}
\frac{(a_1)_n (a_2)_n \cdots (a_p)_n}
{(b_1)_n (b_2)_n \cdots (b_q)_n}
\frac{z^n}{n!},
\end{eqnarray}
% Definition of the Pochhammer symbol
where $(a)_n$ is the Pochhammer symbol (also called rising factorial):
\[
(a)_n = a(a+1)\cdots(a+n-1)
= \frac{\Gamma(a+n)}{\Gamma(a)}.
\]
For any sequence $\left(a_n\right)_{n\in \N}\in \R^{\N}$ such that the following series $h$ is absolutely convergent on $D\subset \R$,
\begin{equation}
    \forall x\in D,\ \  h(x):=\sum_{n=0}^{+\infty}a_nx^n,
\end{equation}
we define its truncated version $\hat{h}^N$ as:
\begin{eqnarray}
    \forall x\in D,\ \  \hat{h}^N(x):=\sum_{n=0}^{N}a_nx^n.
\end{eqnarray}
Similarly, the truncated version of the Bessel function $J_{\alpha}$ and the generalized hypergeometric function are respectively:
\begin{eqnarray}
\label{eq:truncatedbesselandhyper}
\forall N\in \N,\tab
    \begin{cases}
        \widehat{J}^N_{\alpha}(x):=\sum_{m=0}^{N} 
\frac{(-1)^m}{m!\,\Gamma(m+\alpha+1)}
\left(\frac{x}{2}\right)^{2m+\alpha} \\
 \widehat{{}_pF_q}^N\!\big(a_1,\dots,a_p; b_1,\dots,b_q; z\big)
:= \sum_{k=0}^{N}
\frac{(a_1)_k (a_2)_k \cdots (a_p)_k}
{(b_1)_k (b_2)_k \cdots (b_q)_k}
\frac{z^k}{k!},
    \end{cases}
\end{eqnarray}
For the sake of generality, we choose to consider the S-fBM kernel defined in $\R^d$, $d\geq 1$ as follows:
\begin{eqnarray}
\label{eq:S-fBM_kernel}
C_{\omega}(\mathbf{x}):=\frac{\nu^{2}}{2}\left(1-\left(\frac{\left\Vert \mathbf{x}\right\Vert }{T}\right)^{2H}\right)\mathbbm{1}_{\{\left\Vert \mathbf{x}\right\Vert \leq T\}}\ ,\ \mathbf{x}\in\mathbb{R}^{d}
\end{eqnarray}
For any continuous functional $f:\R^d \longrightarrow \R$ and any compact $\mathcal{K}\subset \R^d$, we denote:
\begin{eqnarray}
    \left|\left| f\right|\right|_{\infty}^\mathcal{K}:=\underset{x\in \mathcal{K}}\sup \left|f(x) \right|
\end{eqnarray}
Similarly, we denote the truncated $\mathcal{L}^1$ norm for any arbitrary function $f\in \mathcal{L}^{\infty}\left(\R^d,\R\right)$ and $r>0$:
\begin{eqnarray}
    \|f\|_{1}^{>r}
:= 
\int_{\R^d \setminus B(0, r)} |f(\mathbf{x})| \, d\mathbf{x}.
\end{eqnarray}
where $B(0,r):=\{\mathbf{x}\in \R^d,\|\mathbf{x}\|\leq r\}$ is the centred ball of radius $r$. In the sequel, we consider $H\in \left(0,\frac{1}{2}\right)$ and $T>0$.

\subsection{Fourier transform of the S-fBM kernel}
\label{subsec:fourier_transform}

In this section, we explicitly derive the Fourier transform of the S-fBM kernel. When the parameters of the kernel are such that it is positive definite (see Subsection~\ref{subsec:positive_definite}), then this Fourier transform is a probability density function, namely the spectral density of the S-fBM kernel.
The first straightforward attempt is the following result, proved in Appendix~\ref{subsec:proof_spectral_density_sfbm}.
\begin{prop}
\label{prop:spectraldensityS-fBM}
     If $d\in 2\N^{*}$, the following claims hold:
     \begin{enumerate}
         \item  The spectral density of the S-fBM kernel admits the following representation:
   \begin{eqnarray}
   \label{eq:spectraldensityS-fBM}
    \forall c\in \R^d,\tab  g_{\omega}(\mathbf{x})
=\frac{\lambda^2 T^d}{H(1-2H)}
\left[
\frac{J_{\frac{d}{2}}(T\|\mathbf{x}\|)}{(T\|\mathbf{x}\|)^{\frac{d}{2}}}
-T^{1-\frac{d}{2}}\sum_{m=0}^\infty
\frac{(-1)^m}{m!\,\Gamma(m+\tfrac d2)}
\frac{(T\|\mathbf{x}\|/2)^{2m}}{2H+d+2m}
\right]
\end{eqnarray}
\item For any $N\in \N$ and any compact $\mathcal{K}\subset \R^d$ , there exist a positive constant $C$ such that:
\begin{eqnarray}
    \left\| \Hat{g}^N_{\omega}- g_{\omega} \right\|_{\infty}^\mathcal{K} \leq \frac{(TC/2)^{2N}}{N!\,\Gamma(N+\frac{d}{2})(N+\frac{d}{2})} \left(
\left(\frac{TC}{2}\right)^{\frac{d}{2}}+1\right)
\end{eqnarray}
where:
\begin{eqnarray}
    \forall \mathbf{x}\in \mathbb{R}^d,\quad 
    \Hat{g}^N_{\omega}(\mathbf{x})=\frac{\lambda^2 T^d}{H(1-2H)}
\left[
\frac{\widehat{J}^N_{\frac{d}{2}}(T\|\mathbf{x}\|)}{(T\|\mathbf{x}\|)^{\frac{d}{2}}}
-T^{1-\frac{d}{2}}\sum_{m=0}^N
\frac{(-1)^m}{m!\,\Gamma(m+\frac{d}{2})}
\frac{(T\|\mathbf{x}\|/2)^{2m}}{2H+d+2m}
\right] \nonumber \\
\end{eqnarray}
where the truncated version of the Bessel function $\widehat{J}^N_{\frac{d}{2}}$ is defined as in Eq.~\eqref{eq:truncatedbesselandhyper}.\\

     \end{enumerate}
\end{prop}
A particular derivation of Eq.~\eqref{eq:spectraldensityS-fBM} where $d=1$ is developed in Appendix~\ref{subsec:proof_spectral_density_sfbmd1}.\\\\
Besides, the spectral density admits a representation denoted $f_{\omega}$ using the Hypergeometric function, more suitable for simulation purposes, which is the object of the following theorem proved in Appendix~\ref{subsec:proof_fourier_sfbm_hypergeometric}.
\begin{thm}
\label{thm:fourier_S-fBMhypergeometric}
The spectral density $f_{\omega}$ satisfies the following:
\begin{enumerate}
    \item 
    \begin{eqnarray}
    \label{eq:spectraldensityhypergeom}
f_{\omega}(\mathbf{x})=\frac{\nu^{2}T^{d}}{2^{d+1}\pi^{\frac{d}{2}}\left(\frac{d}{2H}+1\right)\Gamma\!\left(\frac{d}{2}+1\right)}\,_{1}F_{2}\left(\frac{d}{2}+H;\frac{d}{2}+H+1,\frac{d}{2}+1;-\frac{\left\Vert \mathbf{x}\right\Vert ^{2}T^{2}}{4}\right),\ \mathbf{x}\in\mathbb{R}^{d},\label{eq:fourier_S-fBM_kernel}
\end{eqnarray}

\item For any compact $\mathcal{K}\subset \R^d$, the following error bound holds:
\begin{eqnarray}
   \forall N\in \N^{*},\tab \left\|\Hat{f}^N_{\omega}-f_{\omega} \right\|_{\infty}^\mathcal{K}\leq\frac{\nu^{2}T^{d}}{2^{d+1}\pi^{\frac{d}{2}}\left(\frac{d}{2H}+1\right)\Gamma\!\left(\frac{d}{2}+1\right)}\frac{C}{e\left(2C r\left(\mathcal{K}\right) e(N-1)-1\right)\left(N-1\right)^{2C r\left(\mathcal{K}\right) e(N-1)-1}}\nonumber\\
\end{eqnarray}
where $C$ is a positive constant and $r\left(\mathcal{K}\right)=\underset{\mathbf{x}\in \mathcal{K}}\sup\| \mathbf{x}\|$ the radius of $\mathcal{K}$.
\end{enumerate}

\end{thm}
These results reveal a curse of dimensionality. A fortiori when $T$ is significant which is the case in financial applications see for instance the  experiments reported in \citet[Figures 8 and 9]{muzy2013random}  in the context of MRW model. However, the exponential speed of convergence compensates this effect.\\
The exponential speed also appears in the error related to the truncated norm as showcased in the following result.
\begin{prop}
\label{prop:queuelonenormerror}
    The following claim holds:
    \begin{eqnarray}
       \forall r>0,\tab  \left\|  \hat f_\omega^N -f_\omega \right\|_1^{>r}\le 
 \frac{\Omega_d \, C(T,d,\nu,H)}{T^{d}} \, \frac{(Tr + \kappa(d))^d  - Tr^d}{d\kappa(d)}e^{-Tr}
    \end{eqnarray}
where:
\begin{itemize}
\item $\Omega_d = \frac{2 \pi^{\frac{d}{2}}}{\Gamma(\frac{d}{2})}$
    \item $\kappa(d) =
\begin{cases}
\gamma(d+1)^{\frac{1}{d-1}}, & \text{if } d > 1, \\[1.2em]
e^{1-\gamma} , & \text{if } d = 1
\end{cases}
$
 \item $\gamma \approx 0.577$ is the Euler constant

 \item $C(T,d,\nu,H) = A\left(\frac{T^2}{2}\right)^{\frac{d}{2}} \frac{\nu^2}{\Gamma(\frac{d}{2})} \frac{d+H}{d(d+H)}$, $A$ is a positive constant.
\end{itemize}
\end{prop}
The proof is in Appendix~\ref{subsec:proof_prop_lone_norm_error}.\\
This means that the truncated $L^1$ norm of the truncation error decays exponentially in $r$, with an algebraic prefactor.

\subsection{Positive definiteness of the S-fBM kernel\label{subsec:positive_definite}}

Here, we investigate the positive definiteness of the kernel underlying the S-fBM construction, a property that is essential for its interpretation as a valid covariance function in Euclidean space. Since positive definiteness is tightly linked to the non-negativity of the associated spectral density, our analysis relies on Fourier-analytic arguments and classical results on radial kernels. We provide a complete characterization of the admissible values of the smoothness parameter as a function of the ambient dimension $d$, thereby identifying the precise regimes in which the S-fBM kernel is well-posed.
\begin{thm}
\label{thm:positive_definiteness}The isotropic S-fBM kernel $K(\Vert\mathbf{u}\Vert)=\dfrac{\nu^2}{2} \left( 1 - \left( \dfrac{\Vert\mathbf{u}\Vert}{T} \right)^{\!2H} \right)\mathbbm{1}_{\{\Vert\mathbf{u}\Vert\leq T\}}$, $\mathbf{u} \in \R^d$, $\nu>0$, $T>0$, $H>0$, is positive definite in $\mathbb{R}^{d}$ if and only if $0<H\leq\frac{3-d}{4}$.
\end{thm}
\begin{proof}
Since a linear rescaling does not affect positive definiteness, it is sufficient to prove that the standardized kernel $K_{\alpha}(\Vert\mathbf{u}\Vert):=(1-\left\Vert \mathbf{u}\right\Vert ^{\alpha})\mathbbm{1}_{\{\left\Vert \mathbf{u}\right\Vert \leq1\}}$ is positive definite if and only if $0<\alpha\leq\frac{3-d}{2}$.
The Fourier transform $\hat{K}_{\alpha}$
of $K_{\alpha}$ is given by equation~\eqref{eq:fourier_S-fBM_kernel} with $\alpha=2H$ and $T=1$. According to Bochner's theorem, $K_{\alpha}$ is positive definite
in $\mathbb{R}^{d}$ if and only if $\hat{K}_{\alpha}(\mathbf{x})\geq0$
for all $\mathbf{x}\in\mathbb{R}^{d}$. 
Remark that $K_{\alpha}$ coincides with the function $I_{\delta,\mu,\nu,\alpha}(t)$ in
\citet[equation~(9) page~1187]{zastavnyi2006buhmann} with the change
of notation $(\delta,\mu,\nu,\alpha,t)=(\alpha,2,\frac{d-1}{2},d-1,\left\Vert \mathbf{x}\right\Vert )$.
Now, \citet[Theorem~5]{zastavnyi2006buhmann} gives conditions under
which $\hat{K}_{\alpha}$ takes negative values. In our situation,
\citet[Theorem~5.3(i)]{zastavnyi2006buhmann} applies and yields
the fact that $K_{\alpha}$ is not positive definite in $\mathbb{R}^{d}$
if $\alpha>(3-d)/2$. Since $\alpha>0$, this means that the S-fBM
kernel can only be positive definite in the cases $d=1$ or $d=2$.
When $d=1$, $K_{\alpha}$ is positive definite if and only if $\alpha\in(0,1]$
\citep{wu2022rough}. When $d=2$ and $\alpha=1/2$, $\hat{K}_{\alpha}$
is given by $\frac{1}{20\pi}\,_{1}F_{2}\left(\frac{5}{4};\frac{9}{4},2;-\frac{\left\Vert \mathbf{x}\right\Vert ^{2}}{4}\right)$,
which is non-negative according to \citet[Theorem~5.1]{cho2018newton},
see also the diagram \citet[Figure~2]{cho2018newton}. From Bochner's
theorem, this shows that $K_{1/2}$ is positive definite in $\mathbb{R}^{2}$.
Since \citet[Theorem~5.3(i)]{zastavnyi2006buhmann} shows that $K_{\alpha}$
is not positive definite in $\mathbb{R}^{2}$ when $\alpha>1/2$,
we can use the characterization of the parameter space in \citet{golubov1981abel} to conclude that $K_{\alpha}$
is positive definite in $\mathbb{R}^{2}$ for every $\alpha\in(0,1/2]$.
This concludes the proof. 
\end{proof}

Theorem~\ref{thm:positive_definiteness} shows that the isotropic S-fBM kernel can only be positive definite when $d=1$ or $d=2$. For $d=1$, the positive definiteness condition is  $0<H\leq\frac{1}{2}$, corresponding to the rough volatility regime already well documented in the literature \citep{gatheral2018volatility,bayer2016pricing,wu2022rough}. To ensure the non-explosiveness of the S-fBM variance (see Eq.~\eqref{eq:nu2intermformula} and Eq.~\eqref{eq:sfbmcovkernel}), the admissible range of positive definiteness becomes $0<H<\frac{1}{2}$.

\subsection{Error bounds of the RFF scheme from Algorithm \ref{algo:fast_euler}}
\label{subsec:error_bounds_rff_sfbm}

\subsubsection{Monte Carlo convergence rate}
The Markovian structure of HMC described in Section \ref{sec:HMC} is useful to derive a convergence bound of the RFF representation toward the true value. In the following result, we consider the problem of estimating the empirical kernel representation of  Eq.~\eqref{eq:MCrandom_fourier_features_expression}. \\
We denote here $\hat{X}$ the Hamiltonian Markov process whose ergodic distribution is $\hat{\pi}$ defined up to a multiplicative constant times $\hat{f}^N_{\omega}$ the spectral density, for an arbitrary truncation order $N>0$,  obtained via the truncation of the hypergeometric function (see Section~\ref{subsec:fourier_transform}).\\
Without loss of generality, we introduce for any function $h:\mathcal{X}\to[a,b]$ where $\mathcal{X}\subset\R^d$ and reals $a\leq b$  the empirical and theoretical means:
\begin{eqnarray}
    \begin{cases}
        \hat{\mu}_M = \frac{1}{M}\sum_{i=1}^M h(\hat{X}_i) \\
        \mu_M = \frac{1}{M}\sum_{i=1}^M h(Y_i) \\
        \mu= \mathbb{E}\left(h(\eta)\right)
    \end{cases}
\end{eqnarray}
where $\hat{X}$ and $Y$ are the Markov processes with ergodic distribution, defined up to a multiplicative constant, $\hat{f}^N_{\omega}$ and $f_{\omega}$ respectively.\\\\
We present the following convergence guarantees demonstrated in Appendix~\ref{subsec:proof_convergence_hmc_spectral_sfbm} with explicit bounds. In particular, we start from the geometric ergodicity (a power decay of the associated spectral gap) of the underlying Markov processes to end up with the convergence in probability of the sample mean toward the target expectation with exponential rate.  
\begin{thm}
\label{thm:convergenceHMCspectralS-fBM}
There exist positive constants $C$ and $C_{\Psi}$ such that for any arbitrary positive real number $R$ the following guarantee holds:
\begin{eqnarray}
    \mathbb{P}(|\hat{\mu}_M - \mu| \ge t)
\le  \exp\!\Big( - \frac{M t^2}{8 C v} + C \Psi_N(T,\nu,d,H,C_{\Psi}) \Big),
\end{eqnarray}
where:
\begin{itemize}
    \item \begin{align*}
\Psi_N(T,\nu,d,H,C_{\Psi}):=\Bigg[
&\left( \frac{T^2}{2} \right)^{\!\frac{d}{2}}
\frac{\nu^2}{\Gamma\!\left( \tfrac{d}{2} \right)}
\Bigg(
    \frac{1}{d}
    \left( \frac{e T^2}{4} \right)^{\!N}
    \frac{1}{e \, \bigl( \tfrac{d}{2} + 1 \bigr)_N}
\\[0.5em]
&
    +\;
    \frac{1}{d + 2H}
    \frac{C_{\Psi}}{
        e \, \bigl( 2 C_{\Psi} \sqrt{N} e (N-1) - 1 \bigr)
        (N-1)^{2 C_{\Psi} \sqrt{N} e (N-1) - 1}
    }
\Bigg)
\\[0.5em]
&
+\;
\frac{\Omega_d \, C_\omega}{T^{d}}
\frac{(T \sqrt{N} + \kappa(d))^{d} - T N^{\frac{d}{2}}}{d \, \kappa(d)} \,
e^{-T \sqrt{N}}
\Bigg].
\end{align*}
\item \[
v := \frac{(b-a)^2}{4} \left(1
+\frac{8\rho}{1-\rho} \right),
\]
\item $\rho\in [0,1]$ is the ergodicity parameter of the underlying Markov process introduced in Lemma~\ref{lem:covstatiocovHMC}.
\end{itemize}
\end{thm}
This theorem provides a non-asymptotic concentration bound for the Monte Carlo estimator $\hat{\mu}_M$ obtained from the HMC sampler tailored to the case of the S-fBM kernel. It states that the estimator converges in probability to the true quantity $\mu$ exponentially fast with the number of samples $M$. The leading term in the exponent is of Gaussian type and reflects the effective variance of the estimator, which depends on the range of the test function and on the geometric ergodicity of the underlying Markov chain through the parameter $\rho$. The additional correction term $\Psi_N(T,\nu,d,H,C_{\Psi})$ captures finite-dimensional and approximation effects arising from truncation, discretization and spectral properties of the underlying S-fBM kernel; importantly, this term does not depend on $M$ but on $N$, the order of truncation of the S-fBM spectral density.\\
Overall, the result shows that HMC achieves reliable statistical accuracy at an exponential rate in the sample size, with explicit constants that quantify both the mixing properties of the chain and the complexity of the target model.\\\\
Similarly, one can derive an $\mathcal{L}^p$ error bound whose proof is in Appendix~\ref{subsec:proof_cor_lp_error_hmc}.
\begin{cor}
\label{cor:lperrorHMC}
    The $p^{th}$ order moment of the discrepancy error satisfies:
\begin{eqnarray}
\label{eq:RFFapproxlperror}
    \|\hat{\mu}_M-\mu \|_{\mathcal{L}^p}
\le e^{C \frac{\Psi_N(T,\nu,d,H,C_{\Psi})}{p}}\,\left(\Gamma\Big(1+\frac{p}{2}\Big)\,\right)^{\frac{1}{p}}
\Big(\frac{1}{8 C v}\,M\Big)^{-\frac{1}{2}},\tab p>0
\end{eqnarray}
where $C$ is a positive constant.
\end{cor}
This inequality shows that the $p^{\mathrm{th}}$ moment of the estimation error decays at the standard Monte Carlo rate $M^{-\frac{1}{2}}$ as the number of samples $M$ increases. Importantly, this correction grows at most exponentially with $1/p$ and does not affect the asymptotic rate in $M$. Overall, the bound confirms that HMC achieves the optimal root-$M$ convergence rate in $\mathcal{L}^p$, with explicit constants that quantify the influence of the chain’s dependence structure.\\\\
In the upcoming section, we tackle the accelerated numerical simulation scheme of Volterra processes with RFF. First, we go through the main accelerated simulation method focusing on its advantages. Then, we derive quantitative estimates of the developed numerical scheme with explicit bounds for an arbitrary kernel.

\subsubsection{Error bounds}

The following results quantify the propagation of errors using the scheme of Eq.~\eqref{eq:XMN} in the particular case of the S-fBM kernel leveraging the previous theoretical guarantees. As already mentioned, Theorem~\ref{thm:eulerXNXNMRFFerrorbound} and Proposition~\ref{prop:eulerXXNMRFFerrorbound} proved in Appendix~\ref{subsec:proof_euler_xnxnm_rff_error_bound}, provides an $\mathcal{L}^p$ bound on the discrepancy between respectively the Euler scheme $X_N$, its random features approximation $X_N^M$, and the true solution $X$. Now we conduct a similar analysis to derive the corresponding bounds in the case of the S-fBM kernel.

\begin{thm}
\label{thm:eulerXNXNMRFFerrorboundSfBM}
If $\sigma$ satisfies \hyperref[eq:hypothesis1]{\(\bm{\mathcal{H}_{\mathrm{1}}}\)}, \hyperref[eq:hypothesis2]{\(\bm{\mathcal{H}_{\mathrm{2}}}\)} and \hyperref[eq:hypothesis3]{\(\bm{\mathcal{H}_{\mathrm{3}}}\)}, then for any arbitrary compact $\mathcal{K}\subset \R_+$, $p>0$, and $\varepsilon \in \left(\frac{1}{p}, \alpha \wedge 1\right)$ there exist positive constants $K^1_p$ and $C_\mathcal{K}^1$ such that for any $M\in \N$ and $t\in \mathcal{K}$:
\begin{eqnarray}
\sup_{s\in[0,t]} \mathbb{E}\big[|X_s-X_N^M(s)|^p\big]
\le  K^1_p\Bigg(
 \Big(1 + \mathbb{E}\left[|X_0|^p\right] \Big)\, \delta_N^{\,p(\alpha \wedge 1) - \varepsilon}
+
e^{C_\mathcal{K}^1 \left(\Psi_N(T,\nu,d,H,C_{\Psi})+t\right)}\,\Gamma\Big(1+\frac{p}{2}\Big)
\left(\frac{M}{v}\right)^{\!-p/2}
\Bigg)
\end{eqnarray}
\end{thm}
The proof is in Appendix~\ref{subsec:proof_euler_xnxnm_rff_error_bound_sfbm}.\\
In the same line as Proposition~\ref{prop:eulerXXNMRFFerrorbound}, Theorem~\ref{thm:eulerXNXNMRFFerrorboundSfBM} establishes a strong $\mathcal{L}^p$-error estimate for the approximation $X_N^M$ of the process $X$. Again, the error naturally decomposes into two contributions: the first term, of order $\delta_N^{\,p(\alpha \wedge 1)-\varepsilon}$, corresponds to the time discretization approximation, while the second term reflects the approximation error induced by the leapfrog discretization within the HMC sampler, in particular via the exponential factor involving $\Psi_N(T,\nu,d,H,C_{\Psi})$ and $t$ as well as the characteristic rate $M^{-p/2}$ of Euler-type schemes driven by Gaussian increments. Overall, the theorem provides explicit convergence rates in both approximation parameters and quantifies the interplay between time discretization, the leapfrog iteration and the proper HMC sampling.

\section{Numerical experiments\label{sec:numerical}}

\subsection{RFF approximation of the S-fBM kernel}

\subsubsection{Goodness of fit}

The first numerical step consists in validating the accuracy of the spectral density representation given in Eq.~\eqref{eq:spectraldensityhypergeom}. In Figure~\ref{fig:spectraldensityhypergeo} of Appendix~\ref{subsec:numerics_sfbm_spectral_density}, we benchmark the one-dimensional case (\(d = 1\)) by comparing the analytical formula in Eq.~\eqref{eq:spectraldensityhypergeom} with a reference obtained through numerical Fourier integration using the Simpson integration scheme. The results exhibit an excellent agreement between the two approaches, with discrepancies remaining at the level of machine precision. This demonstrates that the observed error is not induced by the modelling assumptions or the numerical integration procedure, but is instead attributable to the intrinsic limitations of floating point arithmetic.  

We further extend this validation to the two-dimensional setting by examining the accuracy of the corresponding spectral density representation against numerical Fourier integration. The surface plot displays the pointwise difference over the frequency domain, offering a global view of the approximation error, while the associated color map representations emphasize its spatial structure and magnitude. In all cases, the discrepancies are again confined to machine precision, confirming that the proposed representation accurately reproduces the spectral density in both one and two dimensions. Consequently, the representation can be regarded as numerically exact for practical purposes, with residual errors entirely dominated by floating point rounding effects.

\begin{figure}[H]
    \centering
    % --- Top row ---
    \begin{subfigure}[t]{0.49\textwidth}
        \centering
        \includegraphics[width=\textwidth]{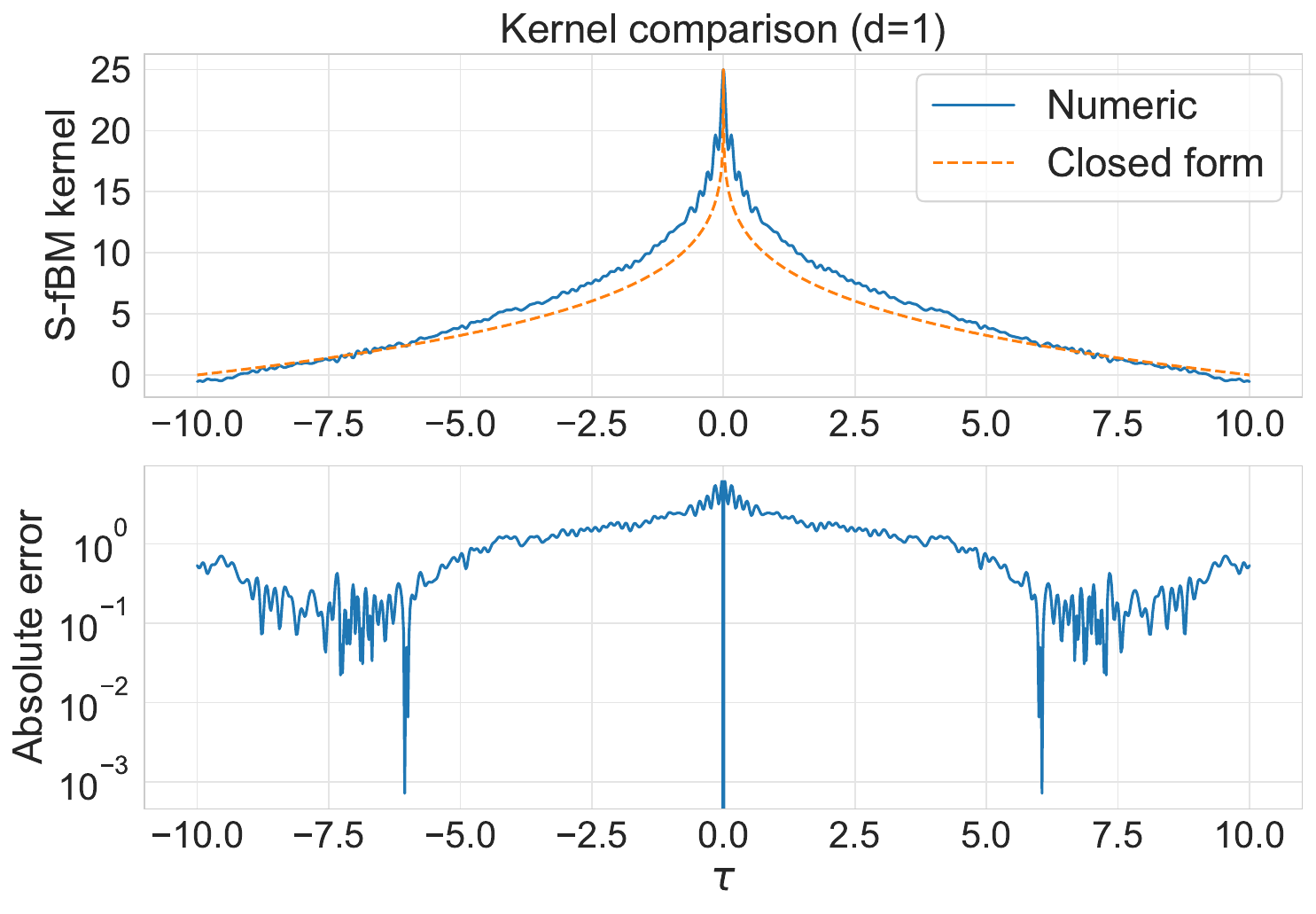}
        \caption{\(T = 10\)}
    \end{subfigure}
    \hfill
    \begin{subfigure}[t]{0.49\textwidth}
        \centering
        \includegraphics[width=\textwidth]{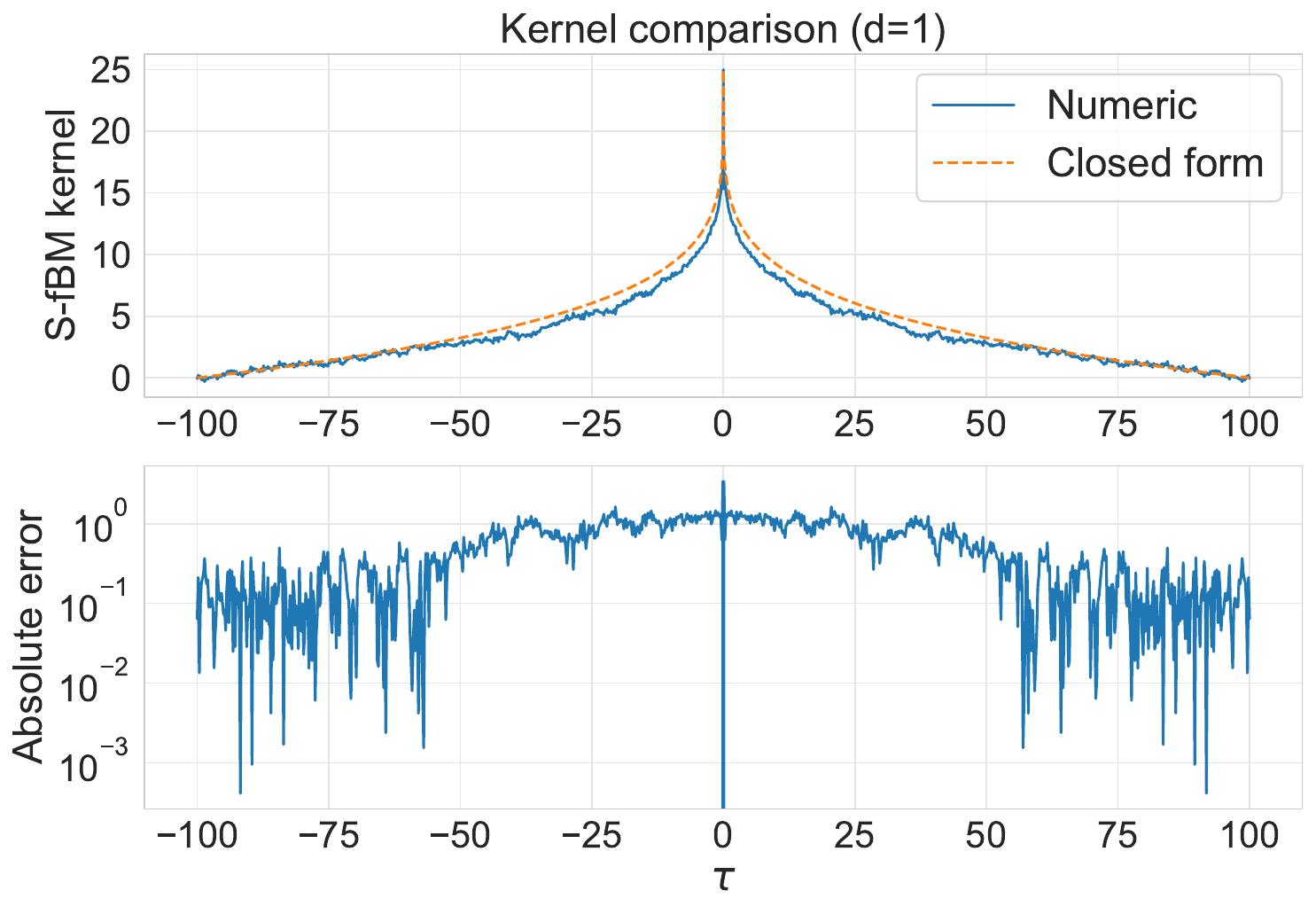}
        \caption{\(T = 100\)}
    \end{subfigure}

    \vspace{0.4cm}

    % --- Bottom row ---
    \begin{subfigure}[t]{0.49\textwidth}
        \centering
        \includegraphics[width=\textwidth]{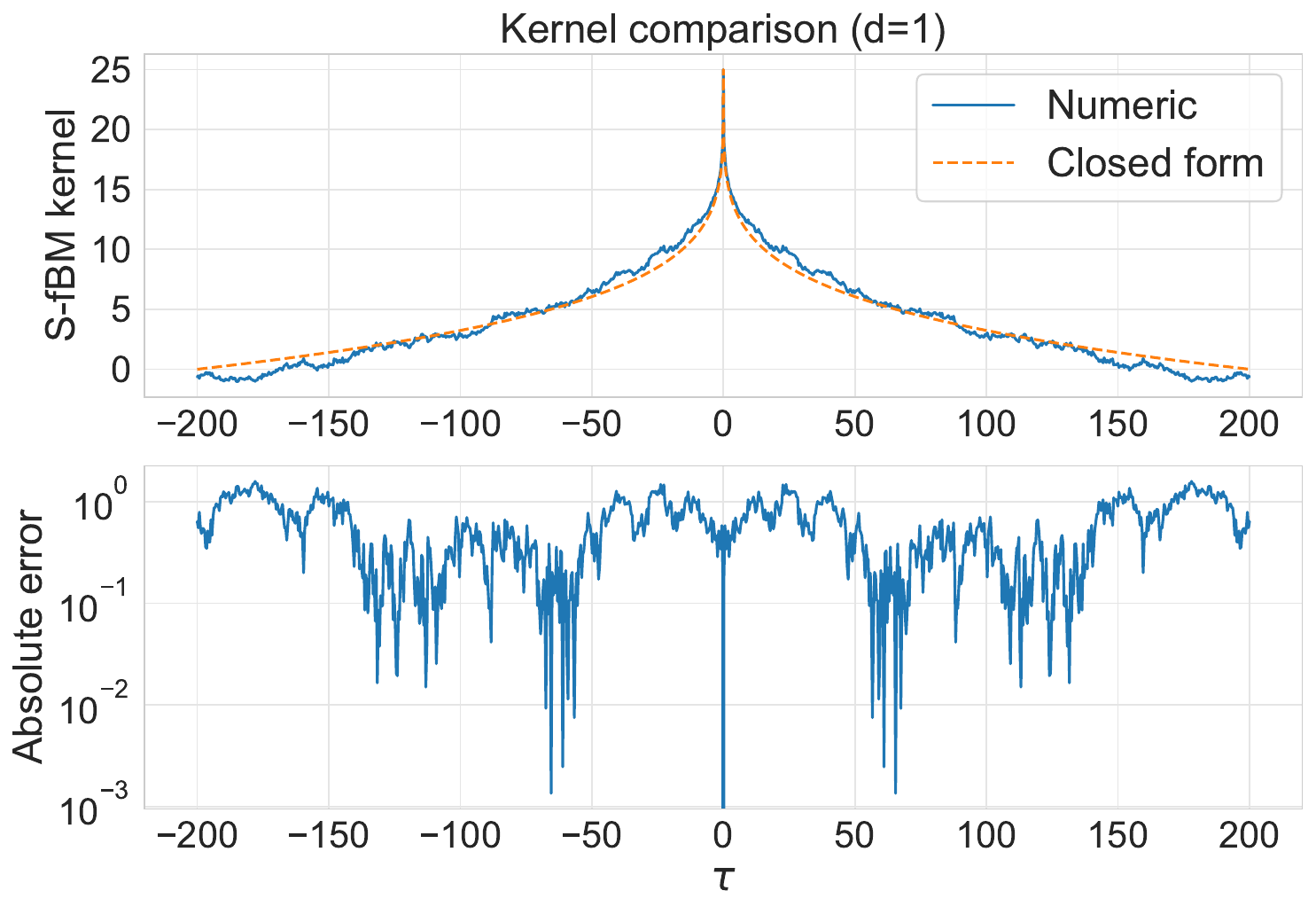}
        \caption{\(T = 200\)}
    \end{subfigure}

    \caption{S-fBM kernel approximation using the RFF representation for different correlation limits \(T\) together with the absolute error between the approximated and theoretical kernel.
    The number of spectral Monte Carlo simulations is \(M = 8000\), with \(\nu^2 = 50\) and \(H = 0.1\).
    The top panels correspond to \(T = 10\) and \(T = 100\), while the bottom panel corresponds to \(T = 200\).}
    \label{fig:RFF_kernel_error_T}
\end{figure}

Figure~\ref{fig:RFF_kernel_error_T} illustrates the accuracy of the kernel approximation obtained via the Random Fourier Features (RFF) representation for different correlation limits \(T\). The results are based on Monte Carlo experiments in which the underlying latent variables were sampled using the HMC procedure described in the Section~\ref{subsec:sampling_algo_hmc}.  The figure compares the behavior of the kernel error across increasing values of \(T\), highlighting the impact of the sample size on the quality of the RFF-based approximation.

\begin{figure}[H]
    \centering
    \begin{subfigure}[t]{\textwidth}
        \centering
        \includegraphics[width=\textwidth]{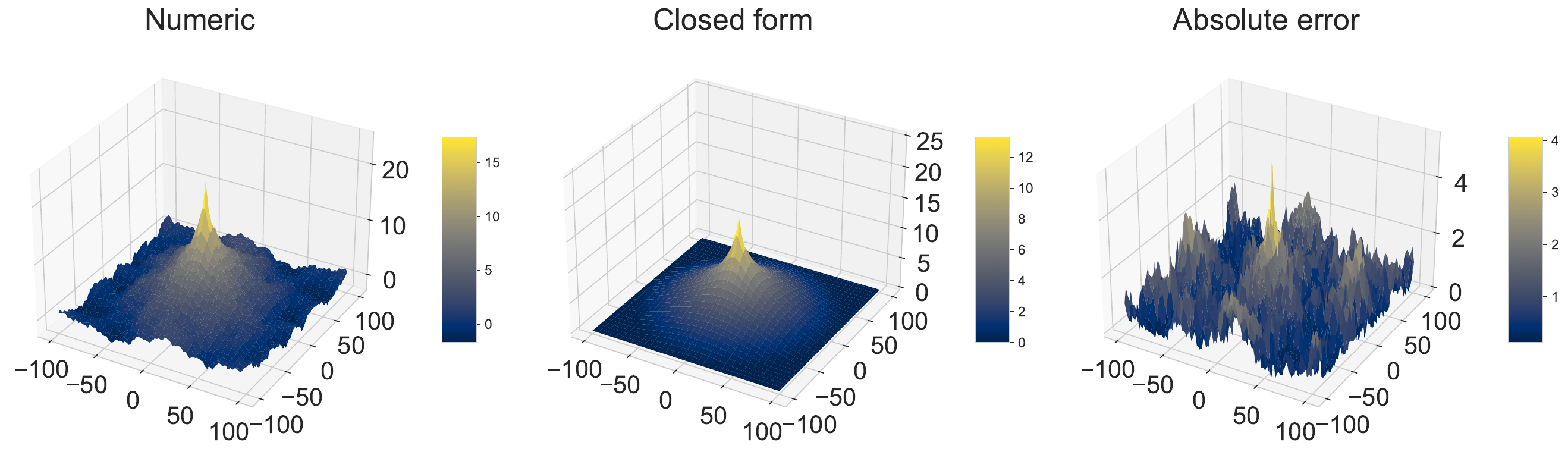}
        \caption{Kernel approximation error surface}
    \end{subfigure}

    \vspace{0.4cm}

    \begin{subfigure}[t]{\textwidth}
        \centering
        \includegraphics[width=\textwidth]{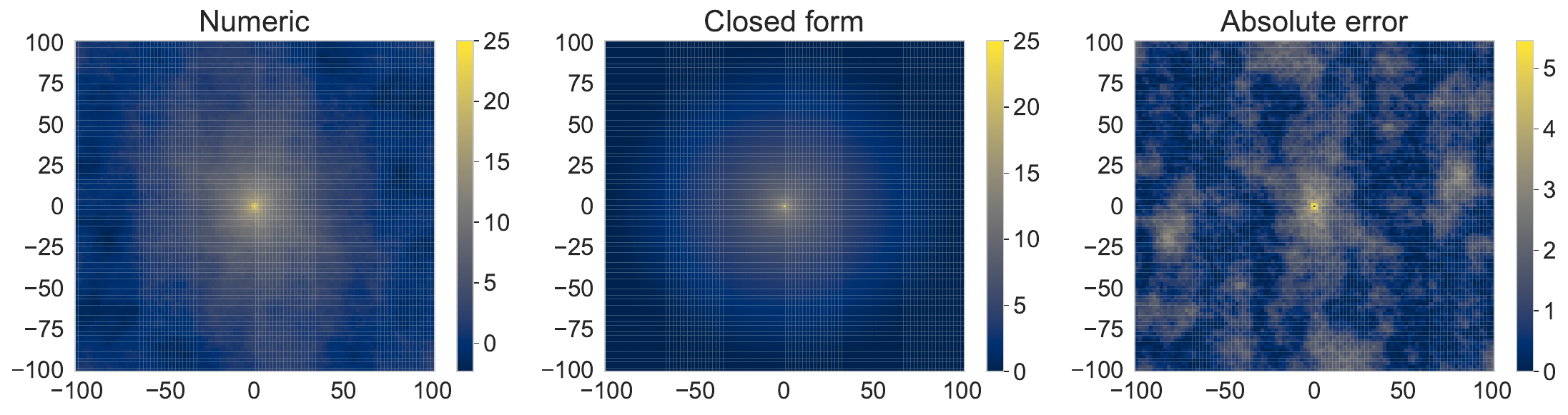}
        \caption{Corresponding colormap representation}
    \end{subfigure}

    \caption{Two dimensional visualization of the S-fBM kernel approximation error obtained using the RFF representation.
    Results are based on \(M = 8000\) Monte Carlo simulations with parameters \(\nu^2 = 50\) and \(H = 0.1\).
    The top panel shows the error surface, while the bottom panel displays the corresponding colormap both for \(T = 100\).}
    \label{fig:RFF_kernel_error_surface}
\end{figure}
Figure~\ref{fig:RFF_kernel_error_surface} presents a two dimensional visualization of the kernel approximation error obtained again using the RFF representation. The results are computed from \(M = 8000\) spectral Monte Carlo simulations with parameters \(\nu^2 = 50\) and \(H = 0.1\), where the latent variables are sampled using the Hamiltonian Monte Carlo (HMC) scheme introduced in the previous section. The top panel displays the error surface as a function of the time lags, providing a global view of the approximation accuracy, while the bottom panel shows the corresponding colormap representation, which emphasizes the spatial variability and magnitude of the error. Together, these visualizations offer complementary insights into the structure and intensity of the approximation error for a fixed sample size \(T = 100\).

%\newpage

\subsubsection{Parameters retrieval}
\noindent
The aim of this section is to verify the good retrieval of the S-fBM parameters from the RFF representation of its kernel. We recall the so-called Generalized Method of Moments (GMM) framework for the estimation of the parameters of the S-fBM model based on autocovariance information evaluated at a finite collection of time lags. For a fixed integer $Q\in\mathbb{N}$, we consider a prescribed sequence of lags $\Tau_Q=(\tau_l)_{l\in\llbracket 1,Q\rrbracket}$ and formulate the GMM estimator as the minimizer of a quadratic discrepancy between empirical and theoretical moment conditions. The latter are defined through the autocovariance function associated with RFF representation of the S-fBM model, leading to a tractable contrast function on the feasible parameter domain $\Xi$. This yields an approximated GMM estimator $\Hat{\theta}_M$ constructed from $M$ independent samples of the random features. We finally state a non-asymptotic error bound for this estimator, which quantifies the impact of the random features approximation.\\

In this subsection we focus on the case $d=1$. We denote, for a fixed $Q\in \N$,  $\Tau_Q:=\left(\tau_l\right)_{l\in \llbracket 1,Q \rrbracket}$ a sequence of time lags. The Generalized Method of Moments can be formulated as follows:
\begin{eqnarray}
  \label{eq:gmm_solution}
       \theta=\argmin_{\bm{x}\in\Xi}\, \bm{h}(\bm{x})^{\!\top} W \bm{h}(\bm{x})
\end{eqnarray}
where $\Xi=[0,1]\times\left(0,\frac{1}{2}\right)\times\R_{+}$ stands for the feasible domain of the vector of true S-fBM parameters $\theta^{*}:=\left(\lambda^2,H,T\right)$, 
 the map $\bm{h}:\Xi \mapsto \R^{Q}$ represents the discrepancy between the empirical and theoretical kernel and $W\in \R^{Q\times Q}$ is a symmetric positive definite matrix.\\
Herein, we consider the discrepancy function being parameterised by the spectral density representation considered:
\begin{eqnarray}
    \bm{h}(\bm{x})= \bm{C}-\bm{C}(\bm{x})
\end{eqnarray}
where :
\begin{itemize}
    \item $\bm{C}=\left(  C_{\omega}(\tau),\tau \in\Tau_Q\right)$ is the observed autocovariance function across the time lags.
    \item $\bm{C}(\bm{x})=\left(  C_{\omega}(\tau)(\bm{x}),\tau \in\Tau_Q\right)$ is the theoretical autocovariance function across the same time lags.
\end{itemize}
We introduce the approximated version of the GMM solution:
\begin{eqnarray}
  \label{eq:gmm_solution_hat}
       \Hat{\theta}_M=\argmin_{\bm{x}\in\Xi} \, \Hat{\bm{h}}_M(\bm{x})^{\!\top} W \Hat{\bm{h}}_M(\bm{x})
\end{eqnarray}
where:
\begin{eqnarray}
    \begin{cases}
     \Hat{\bm{h}}_M(\bm{x})=\Hat{\bm{C}}_M-\bm{C}(\bm{x}) \\
     \Hat{\bm{C}}_M=\left(  \Hat{C}_{\omega,}^{M}(\tau),\tau \in\Tau_Q\right)\\
      \Hat{C}_{\omega}^{M}(\tau)=\frac{K(0)}{M}\sum_{m=1}^{M}\cos(\eta_{m}\tau)
\end{cases}
\end{eqnarray}
$ \Hat{C}_{\omega}^{M}$ represents the empirical counterpart of $K(0)\E\left[\cos\left(\eta \tau \right)\right]$ with $\eta_{1}$, $\eta_{2}$, $\ldots$, $\eta_{M}$
being $M$ independent and identically distributed copies of $\eta$ drawn from the spectral density of the autocovariance function $\boldsymbol{C}$.\\\\
For any matrix $M\in \R^{Q\times Q}$, we consider the matrix norm:
\begin{eqnarray}
    \left|\left|\left|M \right|\right|\right|:=\sup_{x\in \R^Q}\frac{\left|\left|Mx\right|\right|_{2}}{\left|\left|x\right|\right|_2} \nonumber
\end{eqnarray}
where for any $x\in \R^Q$, $\left\|x\right\|_2 = \sqrt{\displaystyle\sum_{i=1}^{Q} x_i^2}$ is the Euclidean norm of $x$. We also introduce the $l^{\infty}$ matrix norm for any $M\in \R^{Q\times Q}$:
\begin{eqnarray}
    \left|\left|M \right|\right|_{\infty}:=\sup_{1\leq i,j\leq Q}\left|M_{i,j}\right| \nonumber
\end{eqnarray}
\begin{thm}
\label{thm:GMMRFF}
    There exists a positive constant $K_{W}$ such that:
    \begin{eqnarray}
\left\|\Hat{\theta}_M - \theta^{*}  \right\|_{\mathcal{L}^2} \leq K_{W} \left|\left|\left|W \right|\right|\right|\sqrt{\frac{Q}{M}}
\end{eqnarray}
\end{thm}
The proof is in Appendix~\ref{subsec:proof_gmm_rff}.\\

Theorem~\ref{thm:GMMRFF} shows that the estimation error decays at rate $M^{-1/2}$, up to explicit constants depending on the weighting matrix and the number of moments. This corresponds to the speed of convergence of the central limit theorem coming from the Monte Carlo RFF approximation.\\

Now, we leverage the GMM method to re-estimate the S-fBM parameter considering as observed data the RFF representation of the kernel obtained via the sampling procedure described in Section~\ref{subsec:rff_representation_sfbm}. For each fixed parameters $H,\lambda^2,T$, we construct a sample of 50 estimates of respectively the Hurst exponent and intermittency parameter denoted $\Hat{H}$ and $\Hat{\lambda}^2$. We do not focus on estimating the correlation limit $T$ which is intractable in practice as illustrated in the numerical experiments of \citet{muzy2013random}. In the Table~\ref{tab:estimatorsfromRFFkernel}, we report the mean value and the standard deviation of each estimate with respect to each configuration of chosen parameters.

\begin{table}[H]
\centering
\begin{tabular}{lccc}
\toprule
\boldmath$\lambda^2 = 0.02$ & \boldmath$H=0.04$ & \boldmath$H=0.1$ & \boldmath$H=0.3$ \\
\midrule
\multicolumn{4}{c}{\boldmath$T = 100$} \\
\midrule
$\hat{H}$            &0.0488(0.0583)  &0.0962(0.0256)  & 0.2982(0.0320) \\
$\hat{\lambda}^2$    & 0.0350(0.0097) &0.0229(0.0027)  &0.0199(0.0013)  \\
\midrule
\boldmath$\lambda^2 = 0.04$ & &  &  \\
\midrule
\multicolumn{4}{c}{\boldmath$T = 200$} \\
\midrule
$\hat{H}$            &0.0417(0.0186)  &0.1028(0.0180)  &0.3033(0.0523)  \\
$\hat{\lambda}^2$    &0.0686(0.0146)  &0.0458(0.0038)  &0.0383(0.0051)  \\
\midrule
\boldmath$\lambda^2 = 0.06$ & &  &  \\
\midrule
\multicolumn{4}{c}{\boldmath$T = 400$} \\
\midrule
$\hat{H}$            &0.0417(0.0232)  &0.1020(0.0245)  &0.3530(0.1095)  \\
$\hat{\lambda}^2$    &0.0999(0.0145)  &0.0676(0.0077)  & 0.0406(0.0239) \\
\midrule
\boldmath$\lambda^2 = 0.08$ & &  &  \\
\midrule
\multicolumn{4}{c}{\boldmath$T = 600$} \\
\midrule
$\hat{H}$            &0.0409(0.0195)  &0.1030(0.0379)  &0.4210(0.2206)  \\
$\hat{\lambda}^2$    &0.1317(0.0245)  &0.0893(0.0120)  &0.1085(0.2312)  \\
\bottomrule
\end{tabular}
\caption{Estimated mean and standard deviation obtained from 50 independent samples of estimated Hurst and intermittency parameters where $M=50,\!000$ using the time lag sequence $\{\tau_k=\lfloor 2^{\frac{k}{2}} \rfloor,k=0...,Q\}$ and $Q=19$}
\label{tab:estimatorsfromRFFkernel}
\end{table}
We observe that the retrieval of the true parameters works well in most situations. For small $H$ values, a small bias in $\hat{\lambda}^2$ is noticed particularly when $T$ is large as it introduces long range dependence, making the kernel highly sensitive to low frequency behaviour and boundary effects. Second, cosine components are inherently oscillatory which can lead to strong multicollinearity among parameters particularly when $H$ and $\lambda^2$ are of the same order of magnitude as they are both involved in the variance formula of the S-fBM process.
Apart from this configuration, the parameters are well recovered with a decent accuracy.

\begin{figure}[H]
    \centering
    % ------- First Row -------
   \begin{minipage}[b]{0.49\linewidth}
        \centering
        \includegraphics[width=\linewidth, height=4.6cm]{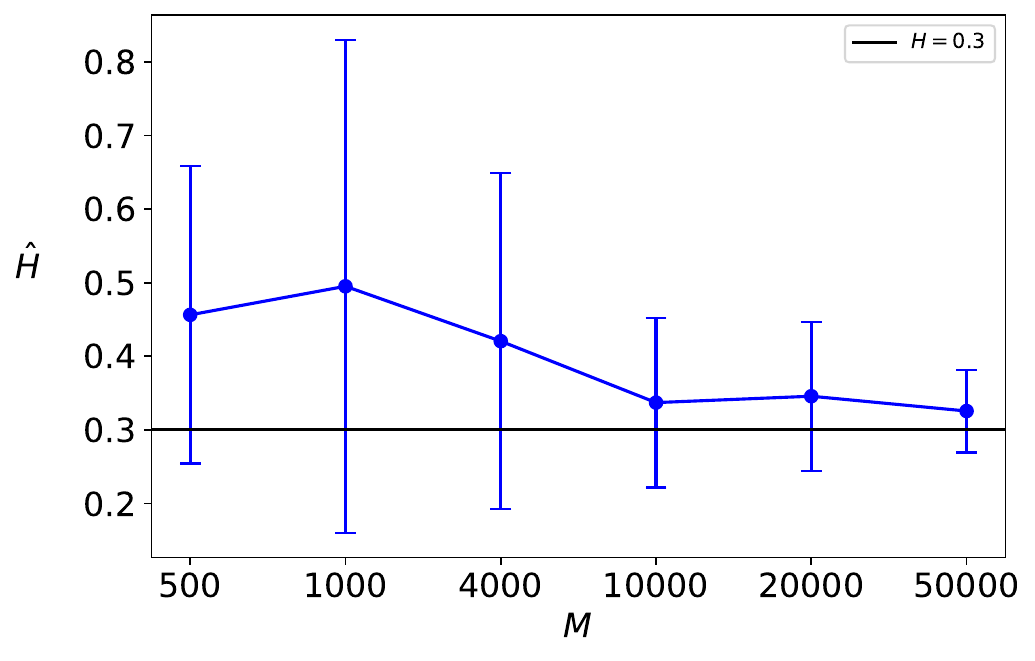}
    \end{minipage}
    \hfill
    \begin{minipage}[b]{0.49\linewidth}
        \centering
        \includegraphics[width=\linewidth, height=4.6cm]{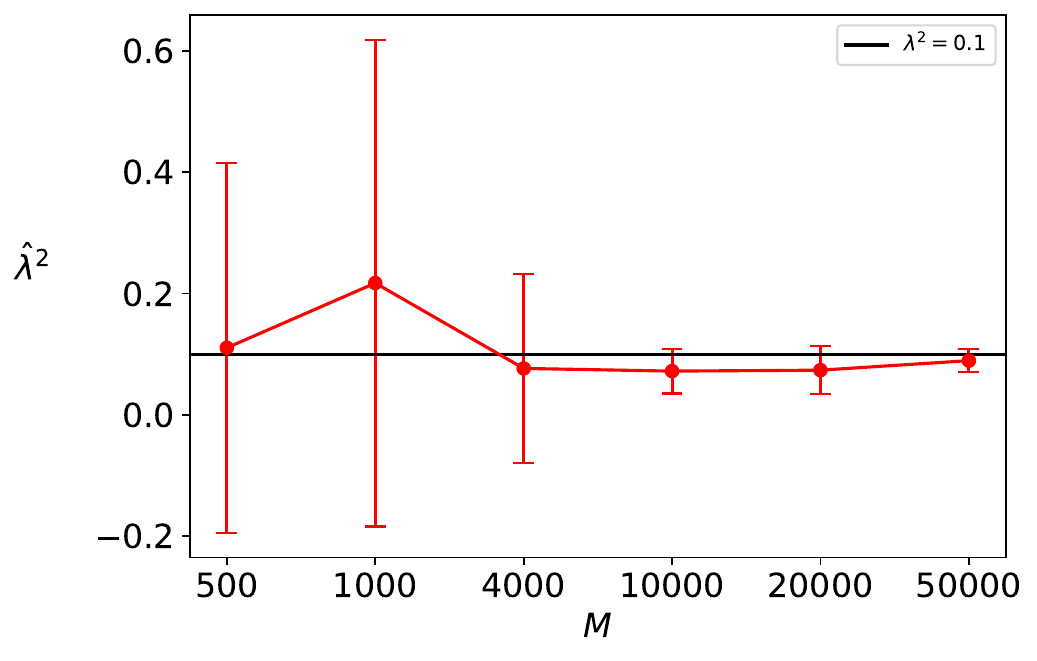}
    \end{minipage}

    \vspace{0.1cm}

    % ------- Second Row -------

     \begin{minipage}[b]{0.49\linewidth}
        \centering
        \includegraphics[width=\linewidth, height=4.6cm]{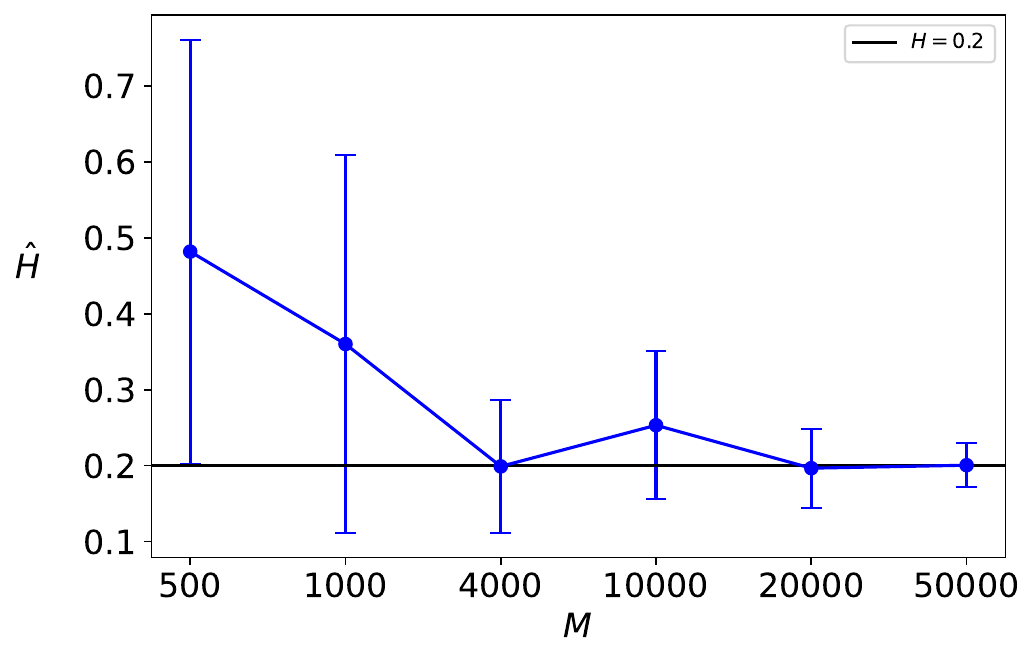}
    \end{minipage}
    \hfill
    \begin{minipage}[b]{0.49\linewidth}
        \centering
        \includegraphics[width=\linewidth, height=4.6cm]{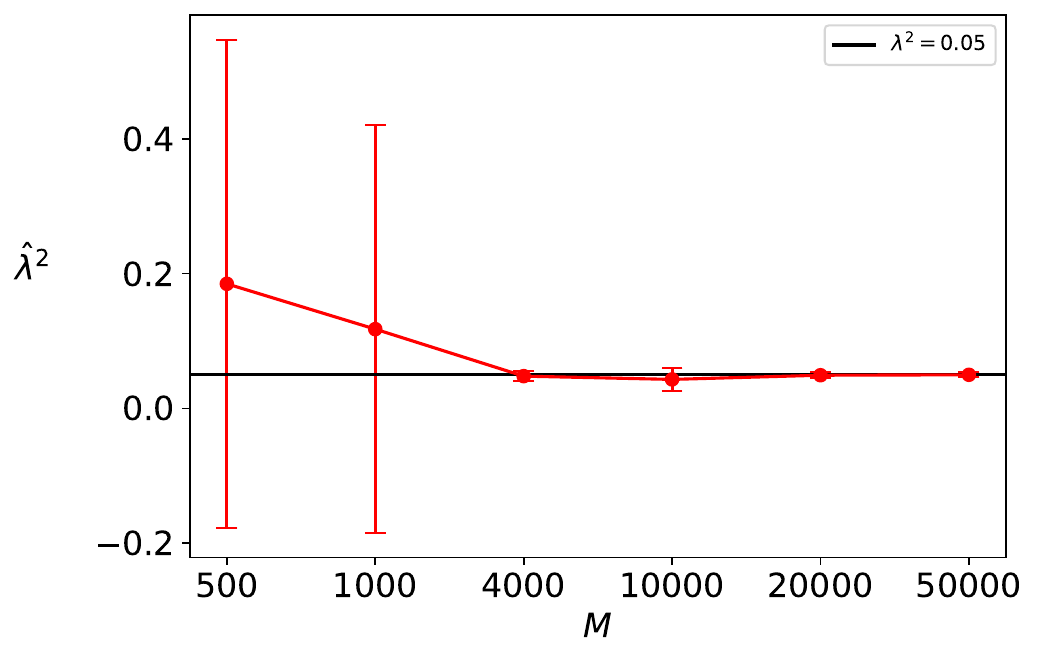}
    \end{minipage}

     % ------- Third Row -------
    \begin{minipage}[b]{0.49\linewidth}
        \centering
        \includegraphics[width=\linewidth, height=4.6cm]{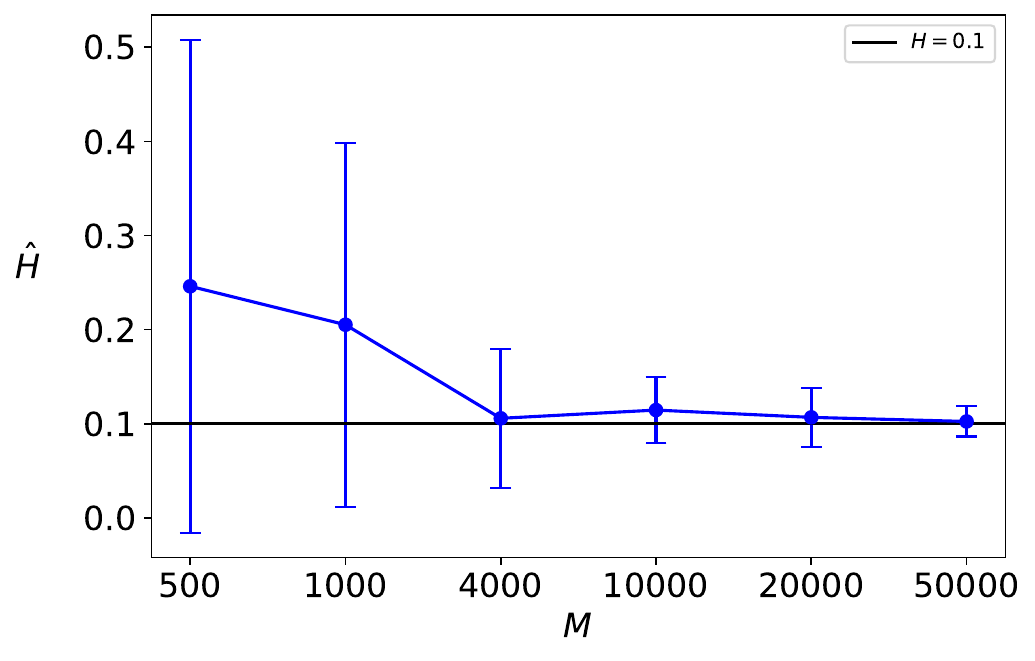}
    \end{minipage}
    \hfill
    \begin{minipage}[b]{0.49\linewidth}
        \centering
        \includegraphics[width=\linewidth, height=4.6cm]{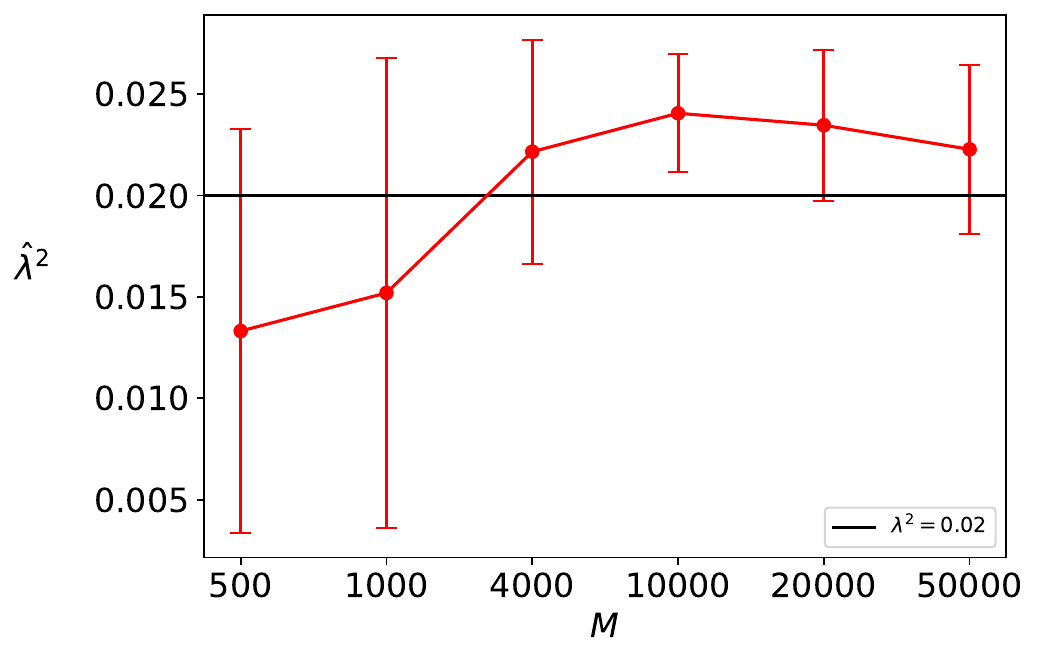}
    \end{minipage}

    \caption{Using the RFF representation of the S-fBM kernel for different sampling realizations $M$, we plot the mean value and the associated $95\%$ confidence interval from a sample of $20$ copies of the Hurst exponent (\textbf{blue}) and the estimated intermittency parameter (\textbf{red})  together with error bars compared with the respective theoretical values: $H = 0.1,\lambda^2=0.02$ (\textbf{bottom}) $H = 0.2,\lambda^2=0.05$ (\textbf{middle}) and $H = 0.3,\lambda^2=0.1$ (\textbf{top}). Here $T=200$ and the time lag sequence used is $\{\tau_k=\lfloor 2^{\frac{k}{2}} \rfloor,k=0...,Q\}$ where $Q=19$.}
    \label{fig:barerrorRFFkernelestiHlambda}
\end{figure}

Figure~\ref{fig:barerrorRFFkernelestiHlambda} illustrates the estimation of the Hurst exponent \(H\) and the intermittency parameter \(\lambda^2\) from the RFF representation of the S-fBM kernel for different spectral sampling sizes \(M\). For each value of \(M\), the mean estimates and the associated \(95\%\) confidence intervals are computed from a sample of \(20\) independent copies and compared with the corresponding theoretical values. The three panels correspond to increasing levels of regularity and intermittency. As \(M\) increases, both estimators converge toward their theoretical values and the confidence intervals shrink, indicating that the estimation procedure becomes increasingly unbiased and more precise.

\subsection{RFF acceleration of SVE simulation}

This section is devoted to the numerical investigation of RFF techniques as a means of accelerating the simulation of SVEs. We begin by examining the qualitative impact of the RFF kernel approximation on the simulated trajectories by comparing sample paths generated via an Euler discretization with and without RFF. We then turn to a quantitative assessment of computational efficiency, where average execution times are reported as a function of the number of time steps, allowing for a direct comparison between the classical Euler scheme and its RFF-based counterpart, both with and without internal feature sampling. Finally, we investigate the accuracy of the proposed method through a weak and strong error analysis. Weak errors are evaluated for smooth test functions under different coupling scenarios of the driving Brownian motions, while strong errors are analyzed in $\mathcal{L}^p$ norms for several values of $p$, with particular attention paid to disentangling the contribution of the RFF approximation from that of stochastic regularization. Together, these experiments provide a comprehensive assessment of the trade-off between computational speed-up and numerical accuracy induced by the RFF approximation.\\\\

Figure~\ref{fig:sample_paths_euler_rff} compares examples of sample paths simulated by the Euler discretization scheme (Algorithm~\ref{alg:slow_euler}) with the paths simulated by the RFF-Euler scheme (Algorithm~\ref{algo:fast_euler}). To isolate the effect of the kernel approximation, both trajectories are generated using the same underlying Brownian path and identical numerical parameters, with only the value of \(H\) varying between the two panels. This controlled setting allows for a direct comparison of the resulting paths and highlights the influence the RFF approximation. 
\begin{figure}[H]
\begin{minipage}[t]{0.48\columnwidth}%
\includegraphics[width=0.38\paperwidth]{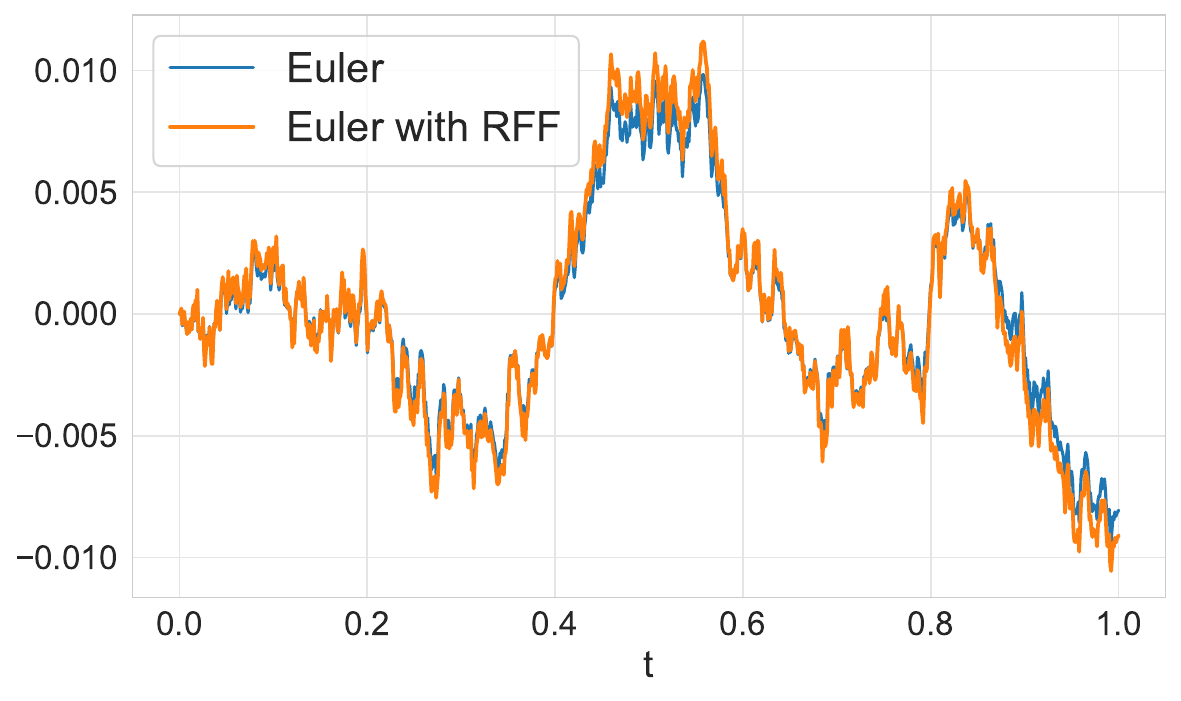}%
\end{minipage}\hfill{}%
\begin{minipage}[t]{0.48\columnwidth}%
\includegraphics[width=0.38\paperwidth]{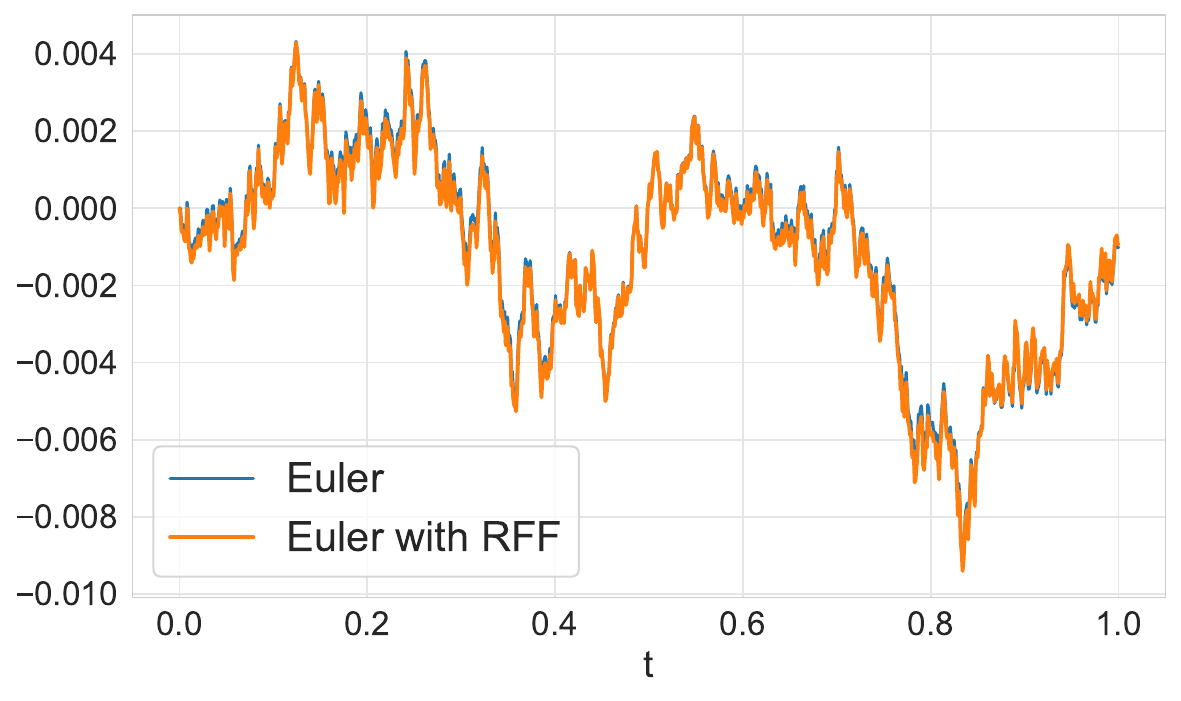}%
\end{minipage}
\caption{Sample paths generated using the Euler scheme with versus without the RFF approximation, driven by the same underlying Brownian path to facilitate comparison, with $H=0.1$ (\textbf{left}) and $H=0.2$ (\textbf{right}). The parameters considered are \(\lambda^2 = 0.01\), \(T = 100\) and \(M = 10{,}000\). The volatility function is given by \(\sigma(t,x) = 0.3\,(1 + 0.1x)\). The simulation uses \(1000\) time steps over the interval \([0,1]\).}
\label{fig:sample_paths_euler_rff}
\end{figure}
One can observe that the Euler and RFF-Euler schemes are very close to each others and that the error between the two is lower for $H=0.2$ compared to $H=0.1$. The fact that the error increases when $H$ decreases is what is expected theoretically, since in the upper bound of Eq.~\eqref{eq:RFFapproxlperror}, the function $\Psi_N$ increases as $H$ decreases. Intuitively, the fact that the $\mathcal{L}^p$ error between the S-fBM kernel and its RFF approximation increases when $H$ decreases is tied to the difficulty to approximate the kernel near zero for small $H$, which affects the decay of the spectral density.\\
\begin{figure}[H]
\includegraphics[width=0.9\textwidth]{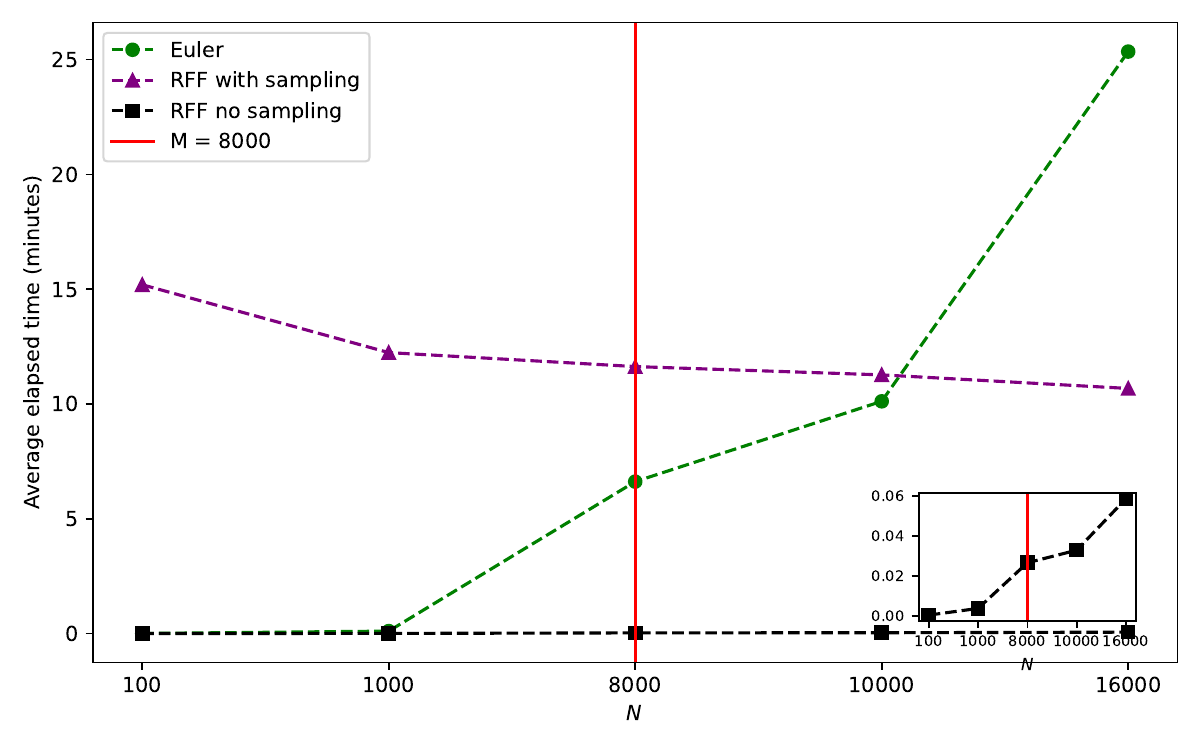}

    \caption{Average elapsed time over 20 runs as a function of the number of time steps $N$ for the classical Euler scheme (\textbf{green}) and the RFF-Euler scheme, respectively, both including (\textbf{purple}) and excluding (\textbf{black}) the random features sampling. The parameters considered are \(H = 0.1\), \(\lambda^2 = 0.01\), \(T = 100\) and \(M = 8000\).
    The volatility function is given by \(\sigma(t,x) = 0.3\,(1 + 0.1x)\). This is performed within a CPU-only job on a single node with 10 cores.}
    \label{fig:elapsedtime_EulerVSRFF}
\end{figure}
Figure~\ref{fig:elapsedtime_EulerVSRFF} compares the computational cost of the Euler simulation scheme implemented with and without RFF sampling performed inside the algorithm. Both implementations are evaluated under identical numerical conditions, using the same underlying Brownian path. The experiment is designed to assess the impact of internal RFF sampling on the overall execution time and to quantify the computational complexity of the scheme described in Algorithm~\ref{algo:fast_euler}. The figure shows a significant advantage of incorporating the RFF with respect to the pure Euler scheme of Algorithm~\ref{alg:slow_euler} a fortiori without incorporating the random feature sampling step in the elapsed execution time. One can see that external feature sampling has a clear advantage over the classical Euler scheme (it is about $200$ times faster). Even internal feature sampling is competitive when simulating relatively long paths of the underlying Volterra process, being $2.5$ times faster. It is important to mention that in practice the feature sampling is performed once and can be stored and reused for multiple tasks.\\

Now, we are interested in the weak error of the RFF-Euler scheme with respect to the classical Euler scheme. In Figure~\ref{fig:weak_error_phi}, we illustrate a weak error analysis of the proposed simulation scheme by evaluating the convergence of expectations of smooth test functions. More precisely, for a given test function \(\phi\), we consider the weak error defined by:
\[
\mathcal{E}_{\phi}(M,T_f) \;=\; \big| \mathbb{E}[\phi(X_n^{M}(T_f))] - \mathbb{E}[\phi(X_n(T_f))] \big|,
\]
where \(X_n^{M}(T_f)\) denotes the numerical approximation of the process at time \(T_f\) using the RFF-Euler scheme and \(X_n(T_f)\) is its classical Euler counterpart. The expectations are estimated using Monte Carlo sampling and the weak errors are reported for two representative test functions, \(\phi(x) = x^2\) and \(\phi(x) = \exp(-x^2)\). The figure compares scenarios in which the numerical and reference solutions are driven either by distinct Brownian paths or by the same Brownian path, for two values of the Hurst exponent \(H\). This setup allows us to assess the sole impact of the RFF kernel approximation on the one hand and both the impact of Brownian regularization and the RFF approximation on the other hand.

\begin{figure}[H]
    \centering
    % -------- First row --------
    \begin{subfigure}[t]{0.48\textwidth}
        \centering
        \includegraphics[width=\textwidth]{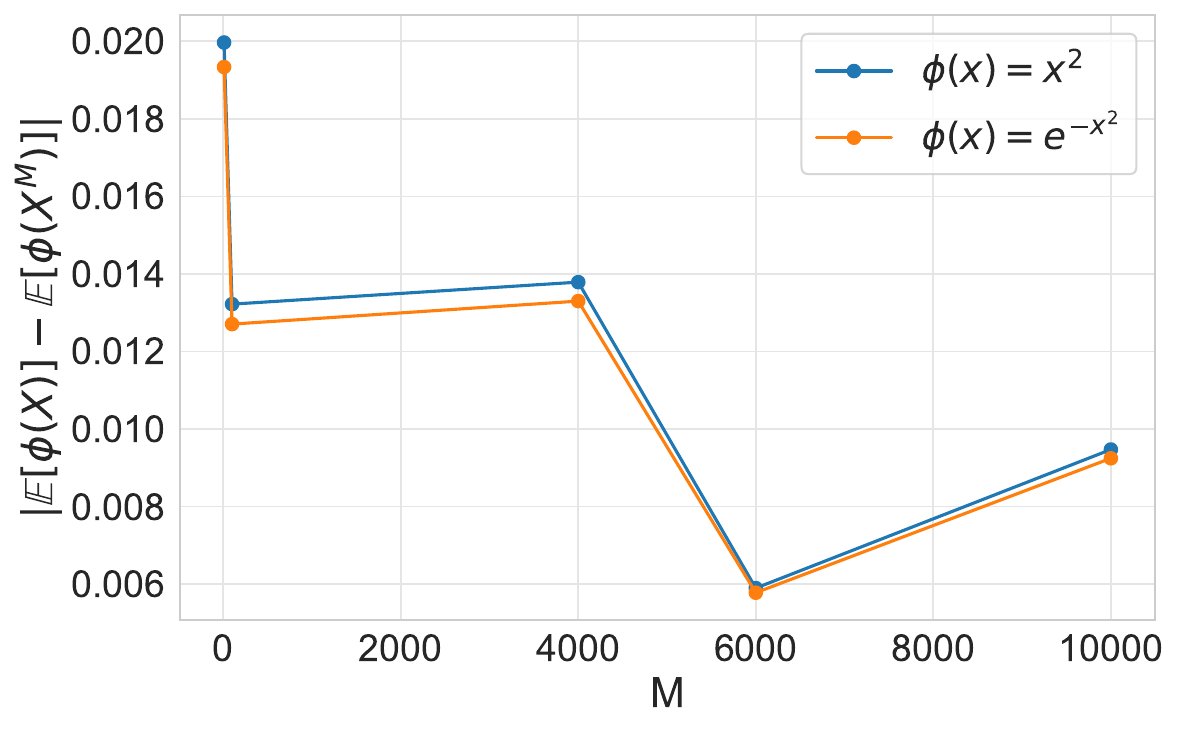}
        \caption{Different Brownian path, \(H = 0.05\)}
    \end{subfigure}
    \hfill
    \begin{subfigure}[t]{0.48\textwidth}
        \centering
        \includegraphics[width=\textwidth]{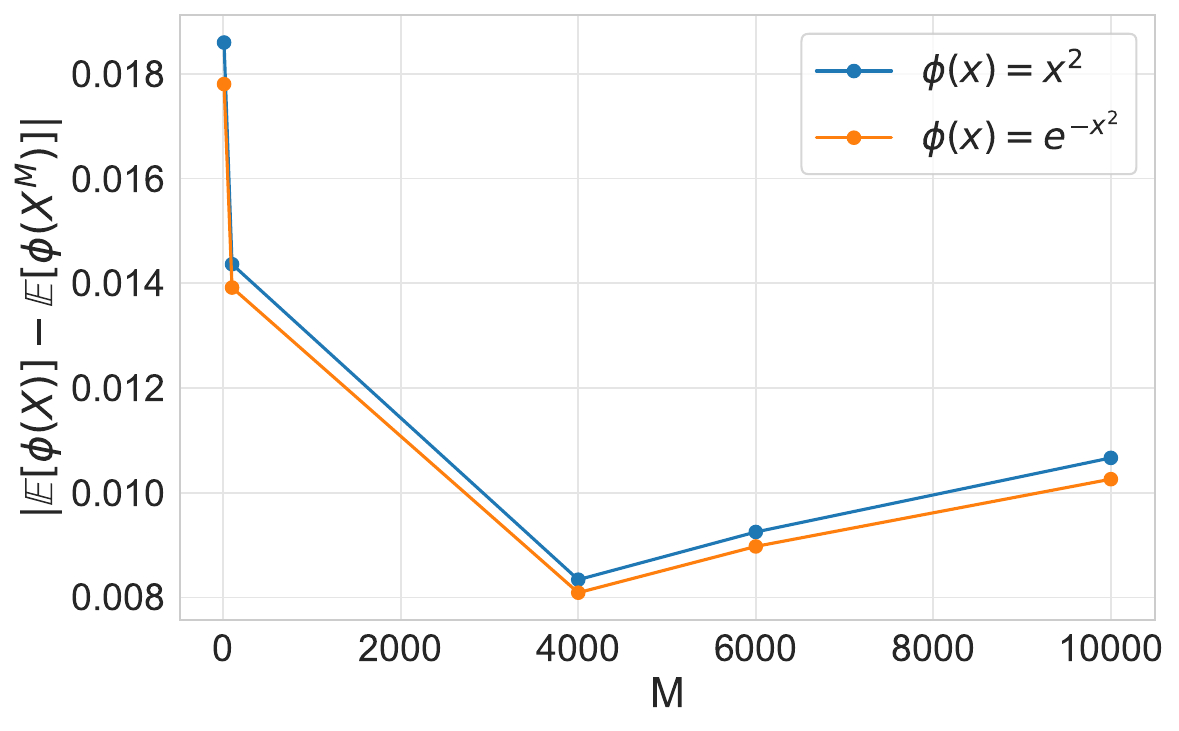}
        \caption{Same Brownian path, \(H = 0.05\)}
    \end{subfigure}

    \vspace{0.4cm}

    % -------- Second row --------
    \begin{subfigure}[t]{0.48\textwidth}
        \centering
        \includegraphics[width=\textwidth]{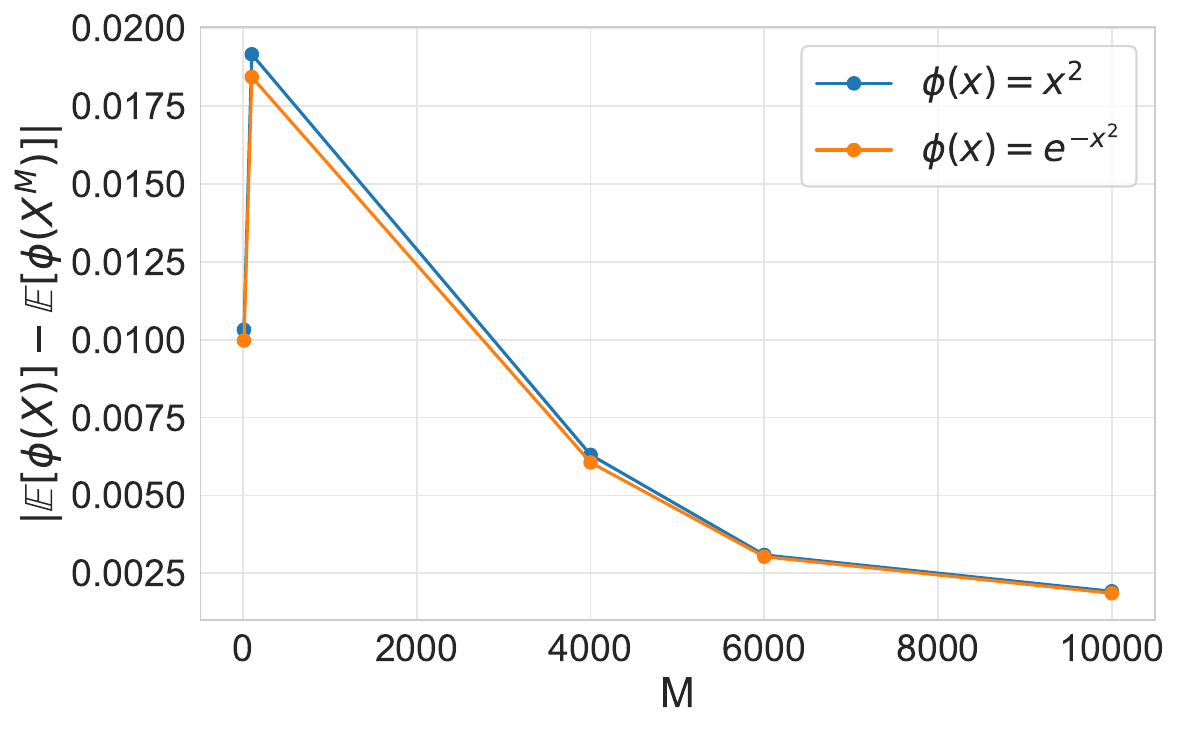}
        \caption{Different Brownian path, \(H = 0.1\)}
    \end{subfigure}
    \hfill
    \begin{subfigure}[t]{0.48\textwidth}
        \centering
        \includegraphics[width=\textwidth]{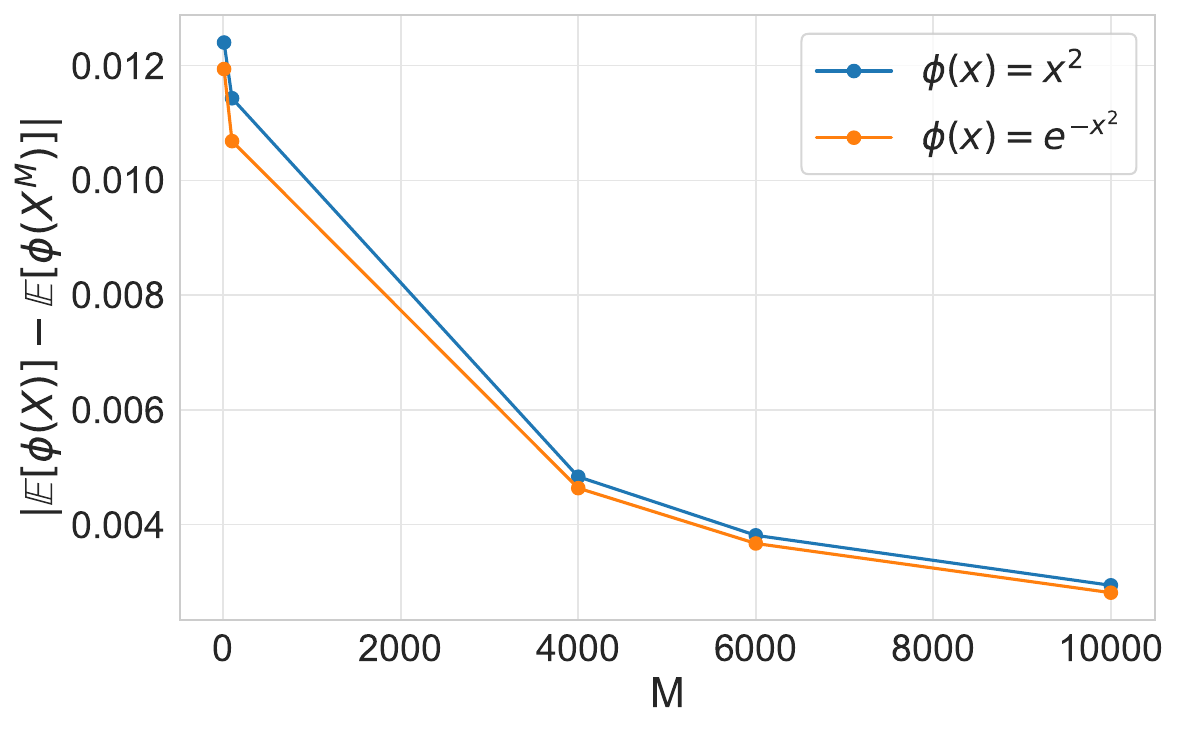}
        \caption{Same Brownian path, \(H = 0.1\)}
    \end{subfigure}

    \caption{Weak error analysis for two test functions \(\phi(x) = x^2\) (blue) and \(\phi(x) = \exp(-x^2)\) (orange).
    Results are obtained with \(T_f = 1\), \(T = 200\), \(\nu = 1\) and \(100\) simulated Brownian paths.
    The left column corresponds to simulations driven by independent Brownian motion paths (between the blue and the orange curves), while the right column uses a common Brownian path.
    The top row shows the case \(H = 0.05\) and the bottom row corresponds to \(H = 0.1\).}
    \label{fig:weak_error_phi}
\end{figure}

The results displayed in Figure~\ref{fig:weak_error_phi} indicate a clear and consistent decay of the weak error as the RFF approximation is refined, for both test functions under consideration. Using a common Brownian path systematically reduces the weak error by a factor of about $10$. Moreover, the qualitative behaviour of the weak error remains similar as a function of $M$ for both \(\phi(x)=x^2\) and \(\phi(x)=\exp(-x^2)\), suggesting that the numerical scheme exhibits robust weak convergence.\\

We now investigate the strong approximation error induced by the use of a finite number \(M\) of simulated features in the construction of the numerical scheme. For any fixed final time horizon \(T_f > 0\) and any \(p \ge 1\), we define the strong error at time \(T_f\) as:
\begin{eqnarray}
\mathcal{E}_{\mathrm{strong}}^{p}(M,T_f)
\;:=\;
\big\| X_n^{M}(T_f) - X_n(T_f) \big\|_{\mathcal{L}^p}^p,
\end{eqnarray}
 The dependence of the strong error on the number of features is then characterized by the mapping:
\[
M \;\longmapsto\; \mathcal{E}_{\mathrm{strong}}^{p}(M,T_f),
\]
which we analyze for arbitrary but fixed time horizon \(T_f\). In order to disentangle the different sources of approximation error, we consider two distinct coupling scenarios. In the first scenario, both processes \(X_n\) and \(X_n^{M}\) are driven by the same underlying Brownian path. This synchronous coupling isolates the contribution of the random Fourier features approximation of the S-fBM kernel and allows us to assess its sole impact on the strong error. In the second scenario, the two processes are driven by independent Brownian motions. This setting captures not only the error induced by the random features approximation, but also the additional discrepancy arising from the stochastic regularization introduced by the Brownian perturbations themselves. Together, these two regimes provide a comprehensive characterization of the strong convergence behavior of the scheme as a function of \(M\).

\begin{figure}[H]
    \centering
    % -------- First row --------
    \begin{subfigure}[t]{0.48\textwidth}
        \centering
        \includegraphics[width=\textwidth]{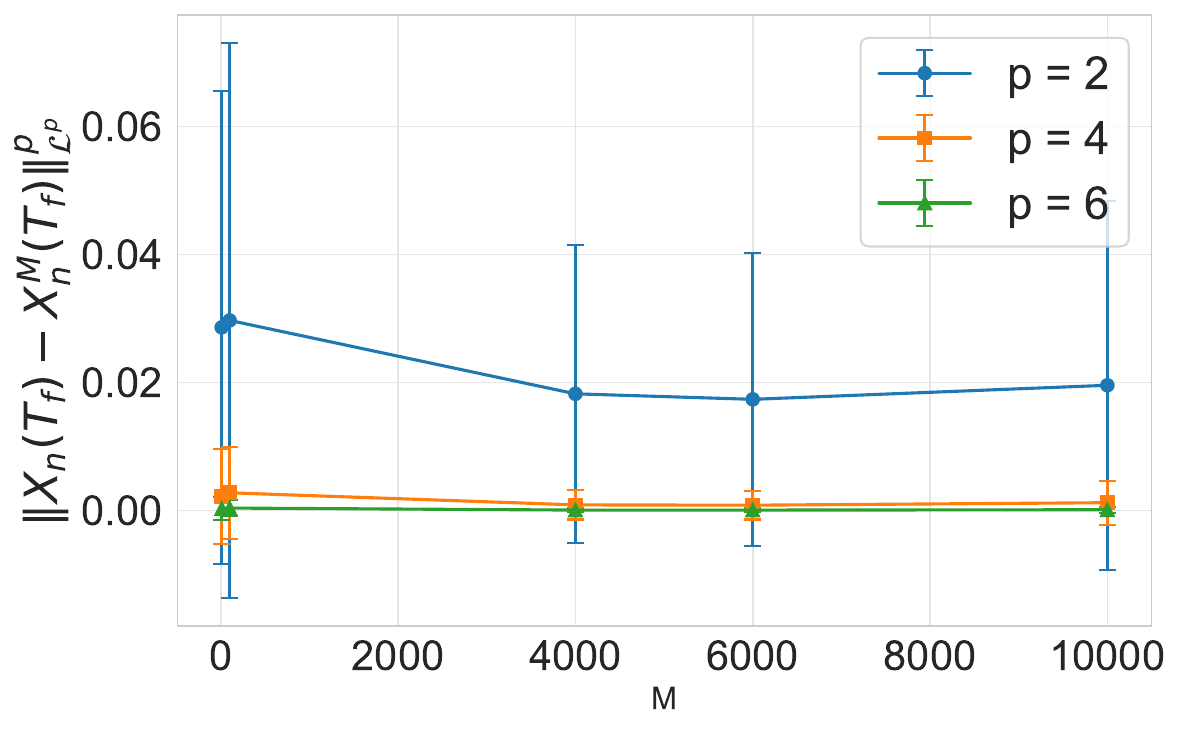}
        \caption{Different Brownian path, \(H = 0.05\)}
    \end{subfigure}
    \hfill
    \begin{subfigure}[t]{0.48\textwidth}
        \centering
        \includegraphics[width=\textwidth]{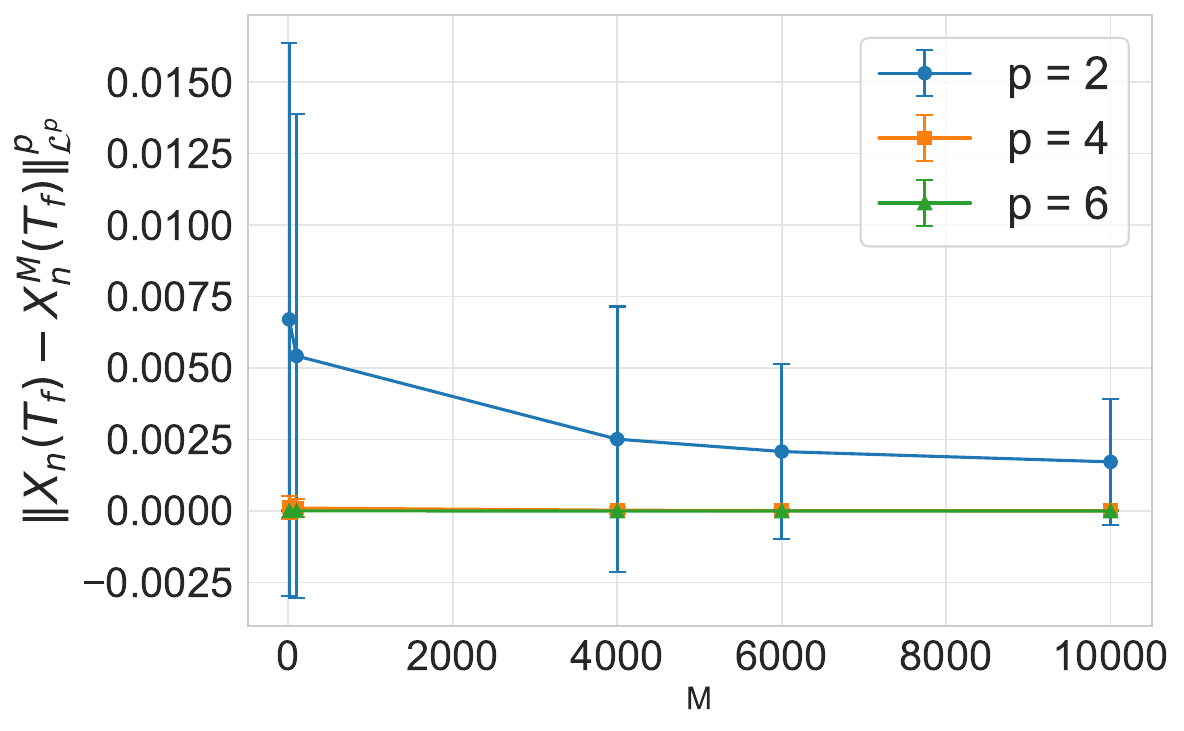}
        \caption{Same Brownian path, \(H = 0.05\)}
    \end{subfigure}

    \vspace{0.4cm}

    % -------- Second row --------
    \begin{subfigure}[t]{0.48\textwidth}
        \centering
        \includegraphics[width=\textwidth]{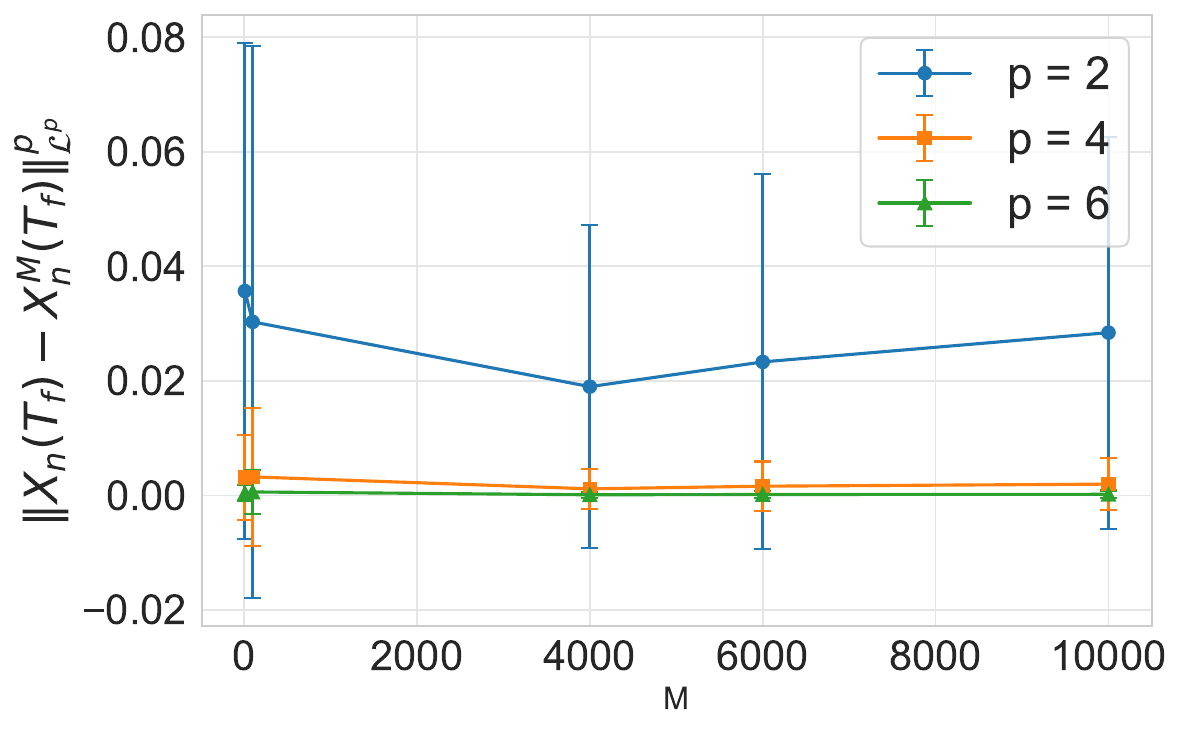}
        \caption{Different Brownian path, \(H = 0.1\)}
    \end{subfigure}
    \hfill
    \begin{subfigure}[t]{0.48\textwidth}
        \centering
        \includegraphics[width=\textwidth]{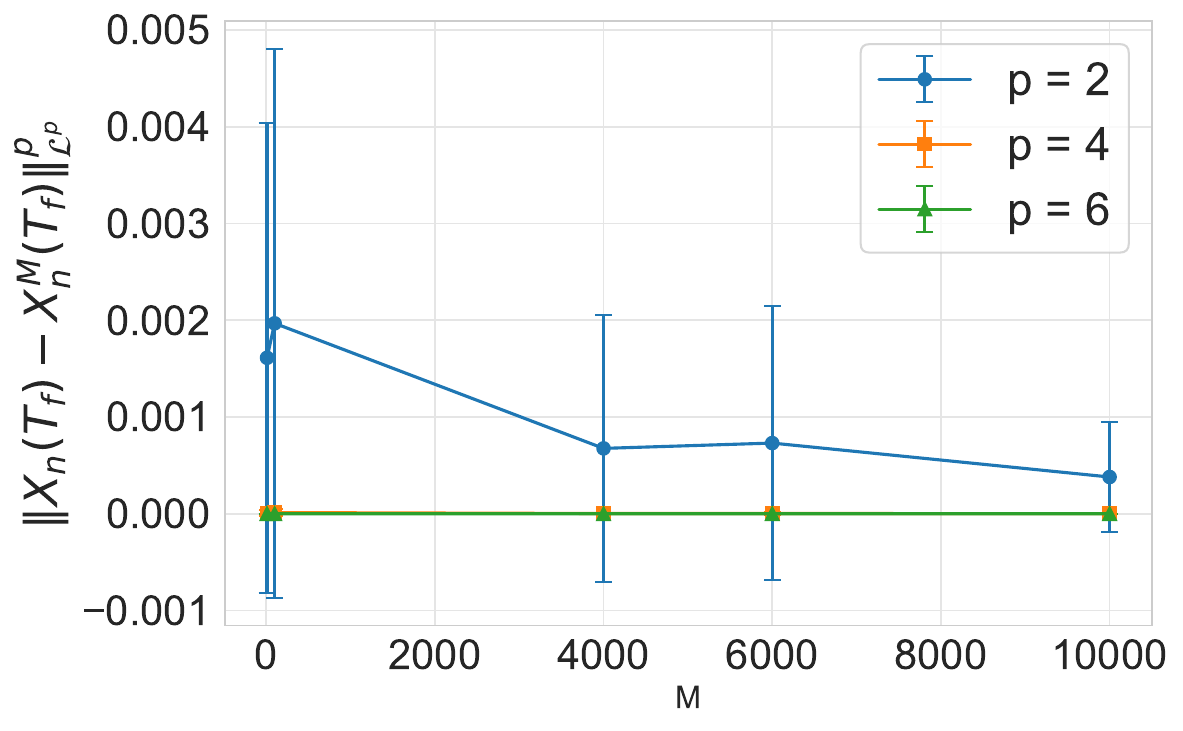}
        \caption{Same Brownian path, \(H = 0.1\)}
    \end{subfigure}

    \caption{Strong error analysis with the corresponding $95\%$ confidence interval considering $p=2$, $p=4$ and $p=6$.
    Results are obtained with \(T_f = 1\), \(T = 200\), \(\nu = 1\) and \(100\) Brownian simulated paths.
    The left column corresponds to simulations driven by independent Brownian motion paths (between the blue, the orange and the green curves), while the right column uses a common Brownian path.
    The top row shows the case \(H = 0.05\) and the bottom row corresponds to \(H = 0.1\).}
    \label{fig:strong_error}
\end{figure}

In Figure \ref{fig:strong_error}, the numerical results indicate an improvement in approximation accuracy as the number of simulated features \(M\) increases. In particular, the strong approximation error is observed to decay monotonically as a function of \(M\), reflecting the enhanced fidelity of the random feature representation of the underlying kernel. This decay is consistent with the expected averaging effects induced by Monte Carlo-type approximations. In parallel, the associated confidence intervals around the empirical error estimates become narrower as \(M\) grows, indicating a reduction in the variance of the error. Together, these effects demonstrate that increasing the number of features improves the average accuracy of the scheme.

\section{Conclusion and prospects\label{sec:conclusion}}
In this work, we have developed a comprehensive theoretical and numerical framework of an accelerated numerical scheme for stochastic Volterra processes based on the Random Fourier Features approximation, with a particular focus on stochastic Volterra equations involving the stationary fractional Brownian motion (S-fBM) covariance kernel. 

Using the theoretical connection between the complete monotonicity and positive definiteness, we proved that the proposed RFF approximation is applicable to a larger class of kernels than the classical multifactor approximation by a sum of exponential terms, including kernels with compact support such as the S-fBM kernel.

The RFF-based spectral Monte Carlo technique, implemented using Hamiltonian Monte Carlo (HMC) sampling, substantially accelerates the simulation of stochastic Volterra equations, with controlled accuracy. Rigorous quantitative results, including strong $\mathcal{L}^p$ error bounds (Proposition~\ref{prop:eulerXXNMRFFerrorbound} and Theorem \ref{thm:eulerXNXNMRFFerrorbound}), establish uniform-in-time control of the schemes and explicit convergence rates in terms of the time discretization truncation parameter $N$ and the number of random features $M$. 

Focusing on the S-fBM kernel, for which we determined explicitly its spectral density and parameter condition for positive definiteness, our numerical experiments confirm that the proposed method reproduces the kernel and its statistical properties with high fidelity. Non-asymptotic error bounds decompose the overall approximation error into contributions from time discretization, spectral truncation and stochastic sampling, with each component explicitly controlled.

To sum up, the proposed RFF-Euler scheme for Volterra equations allows the efficient approximation of a broad class of memory-driven kernels, while maintaining strong guarantees on accuracy, convergence and computational tractability. Future work may include practical applications of the scheme to quantitative finance problems, such as option pricing under rough volatility dynamics and stochastic control problems.

\section*{Acknowledgements}

Othmane Zarhali acknowledges the support from the French government, managed by the National Research Agency (ANR), under the ``France 2030'' program with reference ``ANR-23-IACL-0008''. Nicolas Langren\'e acknowledges the support of the Guangdong Provincial/Zhuhai Key Laboratory of IRADS (2022B1212010006).

%\newpage

\bibliographystyle{apalike}
\bibliography{biblio}
%\newpage

\appendix
In the proofs below, the multiplicative constants in the upper bound errors may change from one line to another.
\section{S-fBM spectral density} 

\subsection{Proof of Proposition~\ref{prop:spectraldensityS-fBMd1}}
\label{subsec:proof_spectral_density_sfbmd1}

\begin{prop}
\label{prop:spectraldensityS-fBMd1}
   The Fourier transform of the S-fBM kernel for $d=1$ can be expressed as follows:
   \begin{eqnarray}
    \forall u\in \R,\tab  g_{\omega}(u)=\sqrt{\frac{2}{\pi}}\frac{\lambda^2}{ H(1-2H)}\left( \frac{\sin(uT)}{u}-T\underset{\substack{k\geq 0 \\k\hspace{0.1cm} even}}\sum\frac{\left(Tu\right)^k\left(-1\right)^{\lfloor k/2 \rfloor}}{\left(2H+k+1\right)k!}\right)
\end{eqnarray}
\end{prop}
\begin{proof}
    $\omega_{H,T}$ is stationary and its autocovariance function is in $\mathcal{L}^1([-T,T])$. Thus, it admits a spectral density denoted $f_{\omega}$.\\\\
    We have:
    \begin{eqnarray}
        g_{\omega}(u)=\frac{1}{\sqrt{2\pi}}\int_{\R}C_{\omega_{H,T}}(h)e^{-ihu}dh
        \nonumber
    \end{eqnarray}
Using the expression of $C_{\omega_{H,T}}(.)$, we end up with:
\begin{eqnarray}
        g_{\omega}(u)=\frac{\lambda^2}{\sqrt{2\pi} H(1-2H)}\left( 2\frac{\sin(uT)}{u}-\int_{-T}^T\left(\frac{|h|}{T}\right)^{2H} e^{-ihu}dh\right)
        \nonumber
\end{eqnarray}
Furthermore, we have by straightforward computations:
$$\overline{\int_{-T}^T|h|^{2H} e^{-ihu}dh}=-\int_{-T}^T|h|^{2H} e^{-ihu}(-dh)=\int_{-T}^T|h|^{2H} e^{-ihu}dh,$$
which means that the previous function is real-valued. Besides, 
\begin{eqnarray}
    \int_{-T}^T|h|^{2H} e^{-ihu}dh=2\int_{0}^T h^{2H} \cos\left(hu\right)dh=2\mathcal{R}\left(\int_{0}^T h^{2H} e^{ihu}dh \right)\nonumber
\end{eqnarray}
Using the series expansion, one has:
\begin{eqnarray}
    \int_{0}^T h^{2H} e^{ihu}dh =\sum_{k=0}^{\infty}\frac{\left(iu\right)^k}{k!}\int_0^Th^{2H+k} dh= T^{2H+1}\sum_{k=0}^{\infty}\frac{\left(iTu\right)^k}{\left(2H+k+1\right)k!}\nonumber
\end{eqnarray}
After some algebra, one has:
\begin{eqnarray}
    \int_{-T}^T|h|^{2H} e^{-ihu}dh=2T^{2H+1}\sum_{k=0}^{\infty}\frac{\left(Tu\right)^k\cos\left(\frac{k\pi}{2} \right)}{\left(2H+k+1\right)k!}\nonumber
\end{eqnarray}
Thus, 
\begin{eqnarray}
    \int_{-T}^T|h|^{2H} e^{-ihu}dh=2T^{2H+1}\underset{\substack{k\geq 0 \\k\hspace{0.1cm} even}}\sum\frac{\left(Tu\right)^k\left(-1\right)^{\lfloor k/2 \rfloor}}{\left(2H+k+1\right)k!},\nonumber
\end{eqnarray}
which gives the claimed result.
\end{proof}

\subsection{Proof of Proposition~\ref{prop:spectraldensityS-fBM}}
\label{subsec:proof_spectral_density_sfbm}

\begin{proof}
\begin{enumerate}
    \item 
Since $C_\omega\in L^1(\mathbb{R}^d)$, one has by definition:
\[
g_{\omega}(\mathbf{x})
=\frac{1}{(2\pi)^{\frac{d}{2}}}\int_{\mathbb{R}^d} C_\omega(\mathbf{h})\,e^{-i\langle \mathbf{h},\mathbf{x}\rangle}\,dh,
\qquad \mathbf{x}\in\mathbb{R}^d.
\]

Substituting $C_\omega$ yields
\begin{align*}
g_{\omega}(\mathbf{x})
&=\frac{\lambda^2}{(2\pi)^{\frac{d}{2}} H(1-2H)}\Biggl[
\int_{\|h\|\le T} e^{-i\langle \mathbf{h},\mathbf{x}\rangle}\,dh
-\frac{1}{T^{2H}}\int_{\|\mathbf{h}\|\le T} \|\mathbf{h}\|^{2H} e^{-i\langle \mathbf{h},\mathbf{x}\rangle}\,dh
\Biggr].
\end{align*}

Both integrals are radial. Using the Fourier--Bessel transform formula for radial functions in $\mathbb{R}^d$, we obtain
\begin{align*}
\int_{\|\mathbf{h}\|\le T} e^{-i\langle \mathbf{h},\mathbf{x}\rangle}\,dh
&=(2\pi)^{\frac{d}{2}}T^d \frac{J_{\frac{d}{2}}(T\|\mathbf{x}\|)}{(T\|\mathbf{x}\|)^{\frac{d}{2}}},\\[0.5em]
\int_{\|\mathbf{h}\|\le T}\|\mathbf{h}\|^{2H}e^{-i\langle \mathbf{h},\mathbf{x}\rangle}\,dh
&=(2\pi)^{\frac{d}{2}}\|\mathbf{x}\|^{1-\frac{d}{2}}\int_0^T r^{2H+\frac{d}{2}} J_{\frac{d}{2}-1}(\|\mathbf{x}\| r)\,dr.
\end{align*}
Thus,  
\[
g_{\omega}(\mathbf{x})
=\frac{\lambda^2}{H(1-2H)}
\left[
T^d\,\frac{J_{\frac{d}{2}}(T\|\mathbf{x}\|)}{(T\|\mathbf{x}\|)^{\frac{d}{2}}}
-\frac{1}{T^{2H}}\,
\|\mathbf{x}\|^{1-\frac{d}{2}}\int_0^T r^{2H+\frac{d}{2}} J_{\frac{d}{2}-1}(\|\mathbf{x}\| r)\,dr
\right].
\]

This shows that \(f_{\omega}\) is a real-valued, radial function of \(\mathbf{x}\).  
Integrating term by term gives
\[
\int_0^T r^{2H+\frac{d}{2}} J_{\frac{d}{2}-1}(\|\mathbf{x}\| r)\,dr
=\sum_{m=0}^\infty 
\frac{(-1)^m}{m!\,\Gamma(m+\frac{d}{2})}
\left(\frac{\|\mathbf{x}\|}{2}\right)^{2m+\frac{d}{2}-1}
\frac{T^{2H+2m+\frac{d}{2}+1}}{2H+2m+\frac{d}{2}+1}.
\]

Plugging this into the expression of \(f_{\omega}(\mathbf{x})\) and simplifying yields
\[
g_{\omega}(\mathbf{x})
=\frac{\lambda^2 T^d}{H(1-2H)}
\left[
\frac{J_{\frac{d}{2}}(T\|\mathbf{x}\|)}{(T\|\mathbf{x}\|)^{\frac{d}{2}}}
-T^{1-\frac{d}{2}}\sum_{m=0}^\infty
\frac{(-1)^m}{m!\,\Gamma(m+\frac{d}{2})}
\frac{(T\|\mathbf{x}\|/2)^{2m}}{2H+d+2m}
\right].
\]
\item 
The Bessel function representation is an alternating series. For the sake of compactness, we consider for any given $x>0$:
\begin{eqnarray}
    a_m(x) = \frac{(-1)^m}{m!\,\Gamma(m+\alpha+1)}
\left(\frac{x}{2}\right)^{2m+\alpha},\tab m\in \N \nonumber
\end{eqnarray}
For a integer order $\alpha>0$, the Bessel function of the first kind is given by the series
\[
J_{\alpha}(x)=\sum_{m=0}^{\infty} a_m(x).
\]
The associated remainder denoted $R_N(x)$ is known as:
\[
R_N(x) 
= \sum_{m=N}^{\infty} 
\frac{(-1)^m}{m!\,\Gamma(m+\alpha+1)}
\left(\frac{x}{2}\right)^{2m+\alpha}.
\]

Let
\[
b_m = |a_m| = 
\frac{1}{m!\,\Gamma(m+\alpha+1)}
\left(\frac{x}{2}\right)^{2m+\alpha}, m\in \N
\]
Using the identity:
\begin{eqnarray}
\label{eq:identitygammafunction}
    \forall x\in \R_+^*,\tab \Gamma(x+1)=x\Gamma(x) \nonumber
\end{eqnarray}
leads to claim that $\left(b_m\right)_{m\in \N}$ is monotonically decreasing (this means that the radius of convergence is $+\infty$, ensuring absolute convergence in $\R^*_+$). Thus, 
the alternating-series remainder satisfies
\begin{eqnarray}
   \forall x>0,\tab  |R_N(x)| \le b_N 
= \frac{1}{N!\,\Gamma(N+\alpha+1)}
\left(\frac{x}{2}\right)^{2N+\alpha}. \;
\end{eqnarray} 
\begin{eqnarray}
\label{eq:reminderbessel}
    \left|J_{\frac{d}{2}}(T\|\mathbf{x}\|)-\widehat{J}^N_{\frac{d}{2}}(T\|\mathbf{x}\|)\right| 
    \leq \frac{1}{N!\,\Gamma(N+\frac{d}{2}+1)}
    \left(\frac{T\|\mathbf{x}\|}{2}\right)^{2N+\frac{d}{2}}
\end{eqnarray}

Arguing the same way leads for any $\mathbf{x}\in \R^d$ to:

\begin{eqnarray}
    \left|\sum_{m=N}^{\infty}
    \frac{(-1)^m}{m!\,\Gamma(m+\frac{d}{2})}
    \frac{(T\|\mathbf{x}\|/2)^{2m}}{2H+d+2m} \right|
    \leq 
    \frac{(T\|\mathbf{x}\|/2)^{2N}}{N!\,\Gamma(N+\frac{d}{2})\left(2H+d+2N\right)}
\end{eqnarray}

As a result, for any $\mathbf{x}\in \R^d$:

\begin{eqnarray}
    \left| \Hat{g}^N_{\omega}(\mathbf{x})- g_{\omega}(\mathbf{x}) \right| 
    \leq 
    \frac{1}{N!\,\Gamma(N+\frac{d}{2}+1)}
    \left(\frac{T\|\mathbf{x}\|}{2}\right)^{2N+\frac{d}{2}}
    + \frac{(T\|\mathbf{x}\|/2)^{2N}}{N!\,\Gamma(N+\frac{d}{2})\left(2H+d+2N\right)}
\end{eqnarray}

This means that for any compact $K\subset \R^d$, there exists a positive constant $C$ such that:
\begin{eqnarray}
    \left|\left| \Hat{g}^N_{\omega}- g_{\omega} \right|\right|_{\infty}^K \leq \frac{(TC/2)^{2N}}{N!\,\Gamma(N+\frac{d}{2})(N+\frac{d}{2})} \left(
\left(\frac{TC}{2}\right)^{\frac{d}{2}}+1\right)
\end{eqnarray}
where $C:=\underset{x\in K}\sup \left|\left|x\right|\right|$.
\end{enumerate}
\end{proof}

\subsection{Proof of Theorem~\ref{thm:fourier_S-fBMhypergeometric}}
\label{subsec:proof_fourier_sfbm_hypergeometric}

Before demonstrating the error bound, we need this technical lemma:
\begin{lem}
\label{lem:factorialbounds}
    For any $n\in \N$,
    \begin{eqnarray}
        \left( \frac{n}{e} \right)^n e \le n! \le \frac{e^2}{4} \left( \frac{n+1}{e} \right)^{\,n+1}.
    \end{eqnarray}
\end{lem}
\begin{proof}
We start from the basic inequality:
\[
\forall x\in \R, \tab x - 1 < \lfloor x \rfloor \le x,
\]
which immediately implies, for \(x > 1\),
\[
\log(x - 1) \le \log(\lfloor x \rfloor) \le \log(x).
\]
Integrating over the interval \([2, n+1]\), we obtain:
\[
\int_2^{\,n+1} \log(x - 1)\, dx \le \int_2^{\,n+1} \log(\lfloor x \rfloor)\, dx 
= \sum_{k=2}^{n} \log k \le \int_2^{\,n+1} \log x \, dx.
\]
Evaluating the integrals gives:
\[
(n-1)\log n - (n-2) \le \log(n!) \le n \log(n+1) - (n-1) - (2 \log 2 - 2),
\]
which can be simplified to:
\[
n (\log n - 1) + 1 \le \log(n!) \le (n+1)(\log(n+1) - 1) - (2 \log 2 - 2).
\]
Exponentiating both sides yields the bounds:
\[
e^{\,n(\log n - 1) + 1} \le n! \le e^{\, (n+1)(\log(n+1)-1) - (2 \log 2 - 2)},
\]
or equivalently,
\[
\left( \frac{n}{e} \right)^n e \le n! \le \frac{e^2}{4} \left( \frac{n+1}{e} \right)^{\,n+1}.
\]
\end{proof}

We can now proceed with the proof of Theorem~\ref{thm:fourier_S-fBMhypergeometric}.

\begin{proof}
\begin{enumerate}
    \item 
We reason on the kernel $K_{\alpha}(\mathbf{u})=(1-\left\Vert \mathbf{u}\right\Vert ^{\alpha})\mathbbm{1}_{\{\left\Vert \mathbf{u}\right\Vert \leq1\}}$ as the S-fBM kernel is simply a rescaling of it. Since the kernel~\eqref{eq:S-fBM_kernel} is continuous and absolutely
integrable, its Fourier transform is given by \citep[Theorem~B.1]{fasshauer2007meshfree}

\[
\mathcal{F}_{\alpha}(\mathbf{x})=\frac{1}{(2\pi)^{\frac{d}{2}}\left\Vert \mathbf{x}\right\Vert ^{\frac{d}{2}-1}}\int_{0}^{1}t^{\frac{d}{2}}(1-t^{\alpha})J_{\frac{d}{2}-1}(\left\Vert \mathbf{x}\right\Vert t)dt\ .
\]
Then, use equation~\eqref{eq:bessel} to obtain
\begin{align*}
\mathcal{F}_{\alpha}(\mathbf{x}) & =\frac{1}{(2\pi)^{\frac{d}{2}}\left\Vert \mathbf{x}\right\Vert ^{\frac{d}{2}-1}}\int_{0}^{1}t^{\frac{d}{2}}(1-t^{\alpha})\sum_{n=0}^{\infty}\frac{(-1)^{n}(\left\Vert \mathbf{x}\right\Vert t/2)^{\frac{d}{2}-1+2n}}{n!\Gamma\!\left(n+\frac{d}{2}\right)}dt\\
 & =\frac{1}{2^{\frac{d}{2}-1}(2\pi)^{\frac{d}{2}}}\sum_{n=0}^{\infty}\frac{(-\left\Vert \mathbf{x}\right\Vert ^{2}/4)^{n}}{n!\Gamma\!\left(n+\frac{d}{2}\right)}\int_{0}^{1}t^{2n+d-1}(1-t^{\alpha})dt\ .
\end{align*}
Then, using a change of variable and the definition of the beta function
$B(a,b)=\frac{\Gamma(a)\Gamma(b)}{\Gamma(a+b)}=\int_{0}^{1}t^{a-1}(1-t)^{b-1}du$
with $a=(2n+d)/\alpha$ and $b=2$,
\begin{align*}
 & \int_{0}^{1}t^{2n+d-1}(1-t^{\alpha})dt=\frac{1}{\alpha}\int_{0}^{1}t^{\frac{2n+d}{\alpha}-1}(1-t)dt=\frac{1}{\alpha}\frac{\Gamma\!\left(\frac{2n+d}{\alpha}\right)}{\Gamma\!\left(\frac{2n+d}{\alpha}+2\right)}\\
 & =\frac{1}{\alpha}\frac{1}{\left(\frac{2n+d}{\alpha}+1\right)\left(\frac{2n+d}{\alpha}\right)}=\frac{\alpha}{4}\frac{1}{\left(\frac{d+\alpha}{2}+n\right)\left(\frac{d}{2}+n\right)}=\frac{\alpha}{4}\frac{\Gamma\!\left(\frac{d+\alpha}{2}+n\right)}{\left(\frac{d}{2}+n\right)\Gamma\!\left(\frac{d+\alpha}{2}+1+n\right)}\ ,
\end{align*}
which gives:
\begin{align*}
\mathcal{F}_{\alpha}(\mathbf{x}) & =\frac{\alpha}{2^{\frac{d}{2}+1}(2\pi)^{\frac{d}{2}}}\sum_{n=0}^{\infty}\frac{(-\left\Vert \mathbf{x}\right\Vert ^{2}/4)^{n}}{n!\Gamma\!\left(\frac{d}{2}+n\right)}\frac{\Gamma\!\left(\frac{d+\alpha}{2}+n\right)}{\left(\frac{d}{2}+n\right)\Gamma\!\left(\frac{d+\alpha}{2}+1+n\right)}\\
 & =\frac{\alpha}{2^{\frac{d}{2}+1}(2\pi)^{\frac{d}{2}}}\sum_{n=0}^{\infty}\frac{\Gamma\!\left(\frac{d+\alpha}{2}+n\right)}{\Gamma\!\left(\frac{d+\alpha}{2}+1+n\right)\Gamma\!\left(\frac{d}{2}+1+n\right)}\frac{\left(-\left\Vert \mathbf{x}\right\Vert ^{2}/4\right)^{n}}{n!}\\
 & =\frac{\alpha}{2^{\frac{d}{2}+1}(2\pi)^{\frac{d}{2}}}\frac{\Gamma\!\left(\frac{d+\alpha}{2}\right)}{\Gamma\!\left(\frac{d+\alpha}{2}+1\right)\Gamma\!\left(\frac{d}{2}+1\right)}\sum_{n=0}^{\infty}\frac{\left(\frac{d+\alpha}{2}\right)_{n}}{\left(\frac{d+\alpha}{2}+1\right)_{n}\left(\frac{d}{2}+1\right)_{n}}\frac{\left(-\left\Vert \mathbf{x}\right\Vert ^{2}/4\right)^{n}}{n!}\\
 & =\frac{1}{2^{\frac{d}{2}}(2\pi)^{\frac{d}{2}}\left(\frac{d}{\alpha}+1\right)\Gamma\!\left(\frac{d}{2}+1\right)}\,_{1}F_{2}\left(\frac{d+\alpha}{2};\frac{d+\alpha}{2}+1,\frac{d}{2}+1;-\frac{\left\Vert \mathbf{x}\right\Vert ^{2}}{4}\right)\ .
\end{align*}
Using a change of variable, we conclude the claimed formula.

\item 
We recall: 
\begin{equation}
f_{\omega}(\mathbf{x})=\frac{\nu^{2}T^{d}}{2^{d+1}\pi^{\frac{d}{2}}\left(\frac{d}{2H}+1\right)\Gamma\!\left(\frac{d}{2}+1\right)}\,_{1}F_{2}\left(\frac{d}{2}+H;\frac{d}{2}+H+1,\frac{d}{2}+1;-\frac{\left\Vert \mathbf{x}\right\Vert ^{2}T^{2}}{4}\right),\ \mathbf{x}\in\mathbb{R}^{d}
\end{equation}

We fix an arbitrary compact $\mathcal{K}\subset{ \R^d}$.\\

By definition, \[
_{1}F_{2}\left(\frac{d}{2}+H;\frac{d}{2}+H+1,\frac{d}{2}+1;-\frac{\left\Vert \mathbf{x}\right\Vert ^{2}T^{2}}{4}\right)
= \sum_{n=0}^{\infty} 
\frac{\left(\frac{d}{2}+H\right)_n}{\left(\frac{d}{2}+1\right)_n\left(\frac{d}{2}+H+1\right)_n}
\frac{\left(-\frac{T^{2}\|\mathbf{x}\|^{2}}{4}\right)^{n}}{n!},
\]

We denote for simplicity:
\begin{eqnarray}
    t_n\left(\mathbf{x}\right):=\frac{\left(\frac{d}{2}+H\right)_n}{\left(\frac{d}{2}+1\right)_n\left(\frac{d}{2}+H+1\right)_n}
\frac{\left(-\frac{T^{2}\|\mathbf{x}\|^{2}}{4}\right)^{n}}{n!}
\end{eqnarray}
After some algebra, the ratio is equal to:

\begin{eqnarray}
    \left|\frac{t_{n+1}\left(\mathbf{x}\right)}{t_n\left(\mathbf{x}\right)}\right|
= \frac{\frac{d}{2}+H+n}{n+1}
\cdot
\frac{\frac{T^{2}\|\mathbf{x}\|^{2}}{4}}
{\left(\frac{d}{2}+1+n\right)\left(\frac{d}{2}+H+1+n\right)}.
\end{eqnarray}
It can be seen that since $\mathbf{x}\in \mathcal{K}$:
\begin{eqnarray}
    \forall \mathbf{x}\in \mathcal{K},\tab  \left|\frac{t_{n+1}\left(\mathbf{x}\right)}{t_n\left(\mathbf{\mathbf{x}}\right)}\right|
\leq \frac{Br(\mathcal{K})^2}{n^2}
\end{eqnarray}
where $r\left(\mathcal{K}\right)=\underset{\mathbf{x}\in \mathcal{K}}\sup \|\mathbf{x}\|$. 
This means that:
\begin{eqnarray}
    \forall \mathbf{x}\in \mathcal{K}, \exists C>0, \forall n\in \N,\tab t_n\left(\mathbf{\mathbf{x}}\right) \leq \frac{\left(Cr\left(\mathcal{K}\right)^2\right)^n}{\left(n!\right)^2}
\end{eqnarray}
Using Lemma~\ref{lem:factorialbounds}, there exists a positive constant $C>0$ (depends only on $d$ and $H$)  such that for any $N\in \N$:
\begin{eqnarray}
    \forall x\in \R^d,\tab t_n\left(\mathbf{\mathbf{x}}\right) \leq \frac{1}{e^2} \left(\frac{C r\left(\mathcal{K}\right) e}{n}\right)^{2n}
\end{eqnarray}
the function $t\rightarrow \left(\frac{e}{t}\right)^{2t}$ in $[1,+\infty[$ is decreasing which means (thanks to the Integral Test) that:
\begin{eqnarray}
    \forall N\in \N^{*},\tab \left|\widehat{{}_1F_2}^N\!\left(\frac{d}{2}+H;\frac{d}{2}+H+1,\frac{d}{2}+1;-\frac{\left\Vert \mathbf{x}\right\Vert ^{2}T^{2}}{4}\right)-{}_1F_2\!\left(\frac{d}{2}+H;\frac{d}{2}+H+1,\frac{d}{2}+1;-\frac{\left\Vert \mathbf{x}\right\Vert ^{2}T^{2}}{4}\right) \right| \nonumber \\
    \leq \frac{1}{e^2}\int_{N-1}^{+\infty}\left(\frac{C r\left(\mathcal{K}\right) e}{t}\right)^{2t}dt\nonumber
\end{eqnarray}
After some algebra, one has:
\begin{eqnarray}
   \forall N\in \N^{*},\tab \left|\widehat{{}_1F_2}^N\!\left(\frac{d}{2}+H;\frac{d}{2}+H+1,\frac{d}{2}+1;-\frac{\left\Vert \mathbf{x}\right\Vert ^{2}T^{2}}{4}\right)-{}_1F_2\!\left(\frac{d}{2}+H;\frac{d}{2}+H+1,\frac{d}{2}+1;-\frac{\left\Vert \mathbf{x}\right\Vert ^{2}T^{2}}{4}\right) \right| \nonumber \\ \leq \frac{C}{e\left(2C r\left(\mathcal{K}\right) e(N-1)-1\right)\left(N-1\right)^{2C r\left(\mathcal{K}\right) e(N-1)-1}}\nonumber
\end{eqnarray}
Putting all the pieces together leads to:
\begin{eqnarray}
   \forall N\in \N^{*}, \left|\Hat{f}^N_{\omega}(\mathbf{x})-f_{\omega}(\mathbf{x})\right|\leq\frac{\nu^{2}T^{d}}{2^{d+1}\pi^{\frac{d}{2}}\left(\frac{d}{2H}+1\right)\Gamma\!\left(\frac{d}{2}+1\right)} \frac{C}{e\left(2C r\left(\mathcal{K}\right) e(N-1)-1\right)\left(N-1\right)^{2C r\left(\mathcal{K}\right) e(N-1)-1}},\nonumber\\
\end{eqnarray}
which yields the claimed result.
\end{enumerate}
\end{proof}

\subsection{Proof of Proposition~\ref{prop:queuelonenormerror}}
\label{subsec:proof_prop_lone_norm_error}

\begin{proof}
We recall: 
\begin{eqnarray}
f_{\omega}(\mathbf{x})=\frac{\nu^{2}T^{d}}{2^{d+1}\pi^{\frac{d}{2}}\left(\frac{d}{2H}+1\right)\Gamma\!\left(\frac{d}{2}+1\right)}\,_{1}F_{2}\left(\frac{d}{2}+H;\frac{d}{2}+H+1,\frac{d}{2}+1;-\frac{\left\Vert \mathbf{x}\right\Vert ^{2}T^{2}}{4}\right),\ \mathbf{x}\in\mathbb{R}^{d} \nonumber
\end{eqnarray}
Let $\hat f_\omega^N(\mathbf{x})$ denote the truncated series up to $N-1$ terms and define the remainder:
\[
R_N(\mathbf{x}) := f_\omega(\mathbf{x}) - \hat f_\omega^N(\mathbf{x}).
\]
The remainder $R_{1,N}$ of the hypergeometric series can be expressed as:
\begin{align*}
{}_1F_2(a;b,c;-z) &= \sum_{n=0}^{\infty} \frac{(a)_n}{(b)_n (c)_n} \frac{(-z)^n}{n!} 
\implies R_{1,N}(z) = \sum_{n=N}^{\infty} \frac{(a)_n}{(b)_n (c)_n} \frac{(-z)^n}{n!}, \quad a=\frac{d}{2}+H, b=\frac{d}{2}+H+1, c=\frac{d}{2}+1.
\end{align*}
Then the total remainder is
\[
R_N(\mathbf{x}) = \frac{\nu^{2}T^d/2}{2^{d}\pi^{\frac{d}{2}}\left(\frac{d}{\alpha}+1\right)\Gamma\!\left(\frac{d}{2}+1\right)} R_{1,N}(z) .
\]
For large $n$, the Pochhammer symbols satisfy:
\[
(c)_n = c (c+1)\cdots (c+n-1) \ge n!  \quad \text{for } c>0,
\]
so that for $n \ge N$ and $z>0$,
\[
\frac{z^n}{(c)_n n!} \le \frac{z^n}{n!}.
\]
A similar argument is valid for ${}_1F_2$,
\[
\frac{(a)_n}{(b)_n (c)_n n!} \le \frac{(C)^n}{n!^2},
\]
for some constant $C$ depending on $a,b,c$.
Thus,
\[
\sum_{n=N}^{\infty} \frac{z^n}{(n!)^2} \le C \, e^{Cz}, \quad z>r,
\]
Let 
\[
C(T,d,\nu,H) = A\left(\frac{T^2}{2}\right)^{\frac{d}{2}} \frac{\nu^2}{\Gamma(\frac{d}{2})} \frac{d+H}{d(d+H)},
\]
where $A$ is a positive constant depending on $a,b,c$. Consequently, for any $\mathbf{x}\in \R^d-B(0,r) $, there exists a positive constant $C$ such that
\[
|R_N(\mathbf{x})| \le  C(T,d,\nu,H)  e^{-CT\|\mathbf{x}\|}
\]
Given $r>0$, the $L^1$ norm over $\|\mathbf{x}\|>r$ in $d$ dimensions is
\[
\| R_N \|_1^{>r} := \int_{\|\mathbf{x}\|>r} |R_N(\mathbf{x})| \, d\mathbf{x} 
= \Omega_d \int_r^\infty |R_N(y)| \, y^{d-1} dy,
\]
with $\Omega_d = \frac{2 \pi^{\frac{d}{2}}}{\Gamma(\frac{d}{2})}$. Plugging in the pointwise bound,
\[
\| R_N \|_1^{>r}\le \Omega_d C(T,d,\nu,H) \int_r^\infty  y^{d-1} e^{-T y} dy
\]
After some algebra, one gets:
\[ 
\int_r^\infty y^{d-1 + 2\beta} e^{-T y} dy = \frac{1}{T^{d }} \int_{T r}^\infty u^{d-1 } e^{-u} du.
\]
Then the integral is the upper incomplete gamma function:
\[
\int_r^\infty y^{d-1} e^{-T y} dy = \frac{1}{T^{d}} \Gamma(d, T r).
\]
given by:
\begin{equation}
\Gamma(s, x) = \int_x^{\infty} t^{s-1} e^{-t} \, dt, \quad s>0
\end{equation}
Thus we obtain the explicit $L^1$ bound:
\[
\| f_\omega - \hat f_\omega^N \|_1^{>r} \le 
 \frac{\Omega_d \, C(T,d,\nu,H)}{T^{d}} \, \Gamma(d, T r)
\]
Thanks to the upper bound of Eq.~(1.5) in \citet{pinelis2020exact}, one has:
\[
\left\| f_\omega - \hat f_\omega^N \right\|_1^{>r}\le 
 \frac{\Omega_d \, C(T,d,\nu,H)}{T^{d}} \, \frac{(Tr + \kappa(d))^d  - Tr^d}{d\kappa(d)}e^{-Tr}
\]
which ends the proof.
\end{proof}

\section{HMC sampling: theoretical background and proofs}
\label{sec:hmc}
This section is devoted to some ergodicity results of the HMC algorithm. Let the state of the HMC chain be denoted by 
\((y,p) \in \mathbb{R}^d \times \mathbb{R}^d\),
where \(y\) represents the position and \(p\) the momentum. 
The Hamiltonian is defined as
\[
H(y,p) = U(y) + K(p),
\quad 
U(y) = -\log \pi(y),
\quad 
K(p) = \frac{1}{2} p^\top M^{-1} p.
\]
Denote by \(\Phi_\epsilon^{\mathrm{LF}}(y,p)\) the discrete Hamiltonian flow obtained after 
\(L\) leapfrog steps of size \(\epsilon\):
\[
(y', p') = \Phi_\epsilon^{\mathrm{LF}}(y,p).
\]
This flow can be decomposed as:
\[
\Phi_\varepsilon^{\mathrm{LF}} =  
\varphi_{\varepsilon/2}^{V}
\circ 
\varphi_{\varepsilon}^{T}
\circ 
\varphi_{\varepsilon/2}^{V},
\]
where \(\varphi_s^T\) denotes the exact flow for the kinetic part (advance \(y\) by \(sM^{-1}p\) holding \(p\) fixed) and \(\varphi_s^V\) denotes the exact flow for the potential part (advance \(p\) by \(-s\nabla U(y)\) holding \(y\) fixed). The Metropolis-Hastings acceptance probability is
\[
\alpha((y,p),(y',p')) 
= \min\bigg\{ 1, \exp\big(-H(y',p') + H(y,p)\big) \bigg\}.
\]
The transition kernel \(P\big((y,p), d(y',p')\big)\) of one HMC iteration is given by
\[
P\big((y,p), d(y',p')\big) =
\alpha((y,p),(y',p'))\, 
\delta_{\Phi_\epsilon^{\mathrm{LF}}(y,p)}\big(d(y',p')\big)
+ 
\big(1 - \alpha((y,p),(y',p'))\big)\,
\delta_{(y,-p)}\big(d(y',p')\big),
\]
where \(\delta_z\) denotes the Dirac measure at point \(z\). This kernel describes the probability of moving from the current state 
\((y,p)\) to a new state \((y',p')\):
with probability \(\alpha((y,p),(y',p'))\), the proposal \(\Phi_\epsilon^{\mathrm{LF}}(y,p)\) is accepted;
and with probability \(1 - \alpha((y,p),(y',p'))\), the proposal is rejected and the momentum is flipped for reversibility. In the continuous-time limit (exact Hamiltonian flow denoted \(\Phi^{\mathrm{LF}}(\cdot,\cdot)\)), the transition kernel is expressed as:
\[
\Prob\big((y,p), A\big) = \mathbf{1}_A\big(\Phi^{\mathrm{LF}}(y,p)\big),
\]
\\
The jump rate is given by:
\[
Q((y,p), d(y',p')) = \alpha(y,p) \, \delta_{\Phi^{\mathrm{LF}}(y,p)}(d(y',p')),
\]
We denote: 
\[
r(y,p) = \int Q((y,p), d(y',p')) = \alpha(y,p),
\]
with:
\[
\alpha(y,p) = \min\big(1, \exp(-H(\Phi^{\mathrm{LF}}(y,p)) + H(y,p))\big).
\]
\begin{prop}
    The infinitesimal generator reads:
\[
\mathcal{L} f(y,p)
= 
\{f,H\} 
+ 
\int_{\mathbb{R}^d} 
\big[f(y,p') - f(y,p)\big]\, Q((y,p), dp'),
\]
where
\(\{f,H\} = \nabla_y f \cdot \nabla_p H - \nabla_p f \cdot \nabla_y H\)
is the Poisson bracket and \(Q((y,p), dp')\) describes the momentum refreshment kernel.
\end{prop}
\begin{proof}
Let \(f\in C^1(\mathbb{R}^{2d})\). Assume the Hamiltonian flow \((y,p)\mapsto\Phi^{\mathrm{LF}}(y,p)\) is \(C^1\).
Assume the jump kernel \(Q((y,p),\cdot)\) has finite total rate
\[
r(y,p) := \int_{\mathbb{R}^d} Q((y,p),dp') < \infty.
\]
The semigroup is defined by
\[
P_t f(y,p) := \mathbb{E}[f(Y_t,P_t) \mid (Y_0,P_0)=(y,p)].
\]
Let \(N_t\) be the number of jumps in \([0,t]\). Then, one has by conditioning:
\[
P_t f(y,p)
= \mathbb{E}[f(Y_t,P_t)\mathbf{1}_{\{N_t=0\}}]
+ \mathbb{E}[f(Y_t,P_t)\mathbf{1}_{\{N_t=1\}}]
+ \mathbb{E}[f(Y_t,P_t)\mathbf{1}_{\{N_t\ge2\}}].
\]
As $N$ is an inhomogeneous Poisson process with rate:
\[
r_s := r(\phi_s(y,p)) = \int_{\mathbb{R}^d} Q(\phi_s(y,p), dp').
\]
One has for small \(t>0\):
\begin{eqnarray*}
\begin{cases}
  \Prob(N_t=0) &= 1 - r(y,p)t + o(t),\\
\Prob(N_t=1) &= r(y,p)t + o(t),\\
\Prob(N_t\ge2) &= o(t).  
\end{cases}
\end{eqnarray*}
Thus:
\begin{itemize}
\item If \(N_t=0\), the process follows the Hamiltonian flow:
\[
\mathbb{E}[f(Y_t,P_t)\mid N_t=0] = f(\Phi^{\mathrm{LF}}(y,p)).
\]
\item If \(N_t=1\), the momentum jumps once to \(p'\) distributed according to \(Q((y,p),dp')/r(y,p)\) and the position \(y\) remains essentially unchanged:
\[
\mathbb{E}[f(Y_t,P_t)\mid N_t=1] = \int_{\mathbb{R}^d} f(y,p')\,\frac{Q((y,p),dp')}{r(y,p)}.
\]
\end{itemize}
Consequently,
\[
\begin{aligned}
P_t f(y,p) 
&= \Prob(N_t=0)\, f(\Phi^{\mathrm{LF}}(y,p)) + \Prob(N_t=1)\, \mathbb{E}[f(Y_t,P_t)\mid N_t=1] + o(t)\\
&= (1-r(y,p)t + o(t))\, f(\Phi^{\mathrm{LF}}(y,p)) + t \int_{\mathbb{R}^d} f(y,p')\, Q((y,p),dp') + o(t).
\end{aligned}
\]
A Taylor expansion along the Hamiltonian flow yields:
\[
f(\Phi^{\mathrm{LF}}(y,p)) = f(y,p) + t \left. \frac{d}{dt} f(\Phi^{\mathrm{LF}}(y,p)) \right|_{t=0} + o(t)
= f(y,p) + t \{f,H\}(y,p) + o(t),
\]
where \(\{f,H\} = \nabla_y f \cdot \nabla_p H - \nabla_p f \cdot \nabla_y H\) is the Poisson bracket.
As a result,
\[
\begin{aligned}
P_t f(y,p) - f(y,p) 
&= t \{f,H\}(y,p) - t r(y,p) f(y,p) + t \int_{\mathbb{R}^d} f(y,p')\, Q((y,p),dp') + o(t) \\
&= t \{f,H\}(y,p) + t \int_{\mathbb{R}^d} [f(y,p') - f(y,p)]\, Q((y,p),dp') + o(t).
\end{aligned}
\]
This leads to:
\[
\mathcal{L} f(y,p) = \{f,H\}(y,p) + \int_{\mathbb{R}^d} [f(y,p') - f(y,p)]\, Q((y,p),dp').
\]
\end{proof}

\begin{prop}
    There exists a Lyapunov function for $\mathcal{L}$, i.e.:
  \begin{eqnarray}
      \exists( \alpha,\beta)\in [0,1] \times \R,\tab 
      \mathcal{L} V(y,p) \le -\alpha V(y,p) + \beta
  \end{eqnarray}  
\end{prop}

\begin{proof}
Consider the Markov process with generator
\[
\mathcal{L} f(y,p) = \{f,H\}(y,p) + \int_{\mathbb{R}^d} [f(y',p') - f(y,p)] \, Q((y,p), d(y',p')),
\]
where
\[
Q((y,p), d(y',p')) = \alpha(y,p) \, \delta_{\Phi^{\mathrm{LF}}(y,p)}(d(y',p')),
\quad
\alpha(y,p) = \min\big(1, \exp(-H(\Phi^{\mathrm{LF}}(y,p)) + H(y,p))\big).
\]
Let the Lyapunov function be
\[
V(y,p) = 1 + H(y,p).
\]
Since $V$ depends only on $H$, the Poisson bracket term vanishes:
\[
\{V,H\} = 0.
\]
Thus, the generator applied to $V$ reduces to:
\[
\begin{aligned}
\mathcal{L} V(y,p) 
&= \int_{\mathbb{R}^d} [V(y',p') - V(y,p)] \, Q((y,p), d(y',p')) \\
&= \alpha(y,p) \left[ V(\Phi^{\mathrm{LF}}(y,p)) - V(y,p) \right] \\
&= \alpha(y,p) \left[ H(\Phi^{\mathrm{LF}}(y,p)) - H(y,p) \right].
\end{aligned}
\]
Define the energy difference
\[
\Delta H(y,p) = H(\Phi^{\mathrm{LF}}(y,p)) - H(y,p).
\]
Then
\[
\mathcal{L} V(y,p) = \alpha(y,p) \, \Delta H(y,p), \quad \alpha(y,p) = \min(1, e^{-\Delta H(y,p)}).
\]
Then
\[
\mathcal{L} V(y,p) = \alpha(y,p) \big[ (\Delta H(y,p))_+ - (\Delta H(y,p))_- \big] = \alpha(y,p) \Delta H(y,p)
\]
Observing that $-(\Delta H)_- \le 0$ and $\alpha(y,p) (\Delta H)_+ \le \sup \alpha(y,p) (\Delta H)_+ =: \beta$, we obtain the drift inequality
\[
\mathcal{L} V(y,p) \le -\alpha V(y,p) + \beta,
\]
with
\[
\alpha = \inf_{(y,p)} \frac{(\Delta H)_-}{V(y,p)} 
= \inf_{(y,p)} \frac{\max(0, H(y,p) - H(\Phi^{\mathrm{LF}}(y,p)))}{1 + H(y,p)},
\]
\[
\beta = \sup_{(y,p)} \alpha(y,p) (\Delta H)_+ 
= \sup_{(y,p)} \alpha(y,p) \max(0, H(\Phi^{\mathrm{LF}}(y,p)) - H(y,p)).
\]
One can easily notice that $\alpha \in [0,1]$.\\
Hence, $V(y,p) = 1 + H(y,p)$ is a Lyapunov function for the generator $\mathcal{L}$, with drift constants $\alpha$ and $\beta$ as above.
\end{proof}
\begin{lem}
\label{lem:covstatiocovHMC}
For any function $h:\mathcal{X}\to[a,b]$ where $\mathcal{X}\subset\R^d$ and reals $a\leq b$, there exist a positive constant $C>0$ and $\rho\in[0,1]$ such that:
\begin{eqnarray}
\forall k \geq 0,\tab     \big|\Cov(h(Y_0),h(Y_k))\big|
\le C\,(b-a)^2\,\rho^k
\end{eqnarray}
\end{lem}

\begin{proof}
   According to the previous results, $Y$ is geometrically ergodic. The Meyn–Tweedie inequality leads to the existence for any $y\in \R^n$ of 
$R(y)<\infty$ and $\rho\in[0,1]$ such that
\begin{eqnarray}
    \forall t\geq 0,\tab \| P_t(y,\cdot) - \pi \|_{TV} \le R(y)\, \rho^{ t}.\nonumber
\end{eqnarray}
Besides,
as \((Y_t)_{t\ge0}\) is a stationary Markov process and homogeneous with transition kernel \(P\) and invariant distribution \(\pi\).\\
Fix \(k\ge0\). By stationarity
\begin{eqnarray}
    \E[h(Y_k)]=\E[h(Y_0)]=\pi(h):=\int_{\R^d}h(y)\pi(dy) \nonumber
\end{eqnarray}
 so
\begin{eqnarray}
    \Cov\big(h(Y_0),h(Y_k)\big)
=\E\big[h(Y_0)h(Y_k)\big]-\pi(h)^2
=\E\big[ h(Y_0)\,P^k h(Y_0)\big]-\pi(h)^2,\nonumber
\end{eqnarray}
where \(P^k h(y)=\int h(y')\,P^k(y,dy')=\E[h(Y_k)\mid Y_0=y]\).
Introduce \(\pi(h)\) inside the expectation to obtain
\[
\Cov\big(h(Y_0),h(Y_k)\big)
=\E\big[ h(Y_0)\,(P^k h(Y_0)-\pi(h))\big].
\]
Taking absolute values and using \(|h|\le b-a\) (after centering if desired) gives
\[
\big|\Cov(h(Y_0),h(Y_k))\big|
\le (b-a)\,\E\big[\,|P^k h(Y_0)-\pi(h)|\,\big].
\]
For any $y\in\mathcal X$,
\[
|P^k h(y)-\pi(h)|
=\Big|\int h(y')\,(P^k(y,dy')-\pi(dy'))\Big|
\le (b-a)\,\|P^k(y,\cdot)-\pi\|_{\mathrm{TV}}.
\]
Taking expectation over $Y_0$ and the supremum over $y$ yields
\[
\E\big[\,|P^k h(Y_0)-\pi(h)|\,\big]\le (b-a)\,\sup_y\|P^k(y,\cdot)-\pi\|_{\mathrm{TV}}.
\]
Combining the last two displays we obtain the first bound
\[
\big|\Cov(h(Y_0),h(Y_k))\big|\le (b-a)^2\sup_y\|P^k(y,\cdot)-\pi\|_{\mathrm{TV}}.
\]
If, in addition, the chain is geometrically ergodic in total variation, then the previous inequality implies
\[
\big|\Cov(h(Y_0),h(Y_k))\big|
\le (b-a)^2\,C\,\rho^k,
\]
as claimed.
\end{proof}

\section{Error of the RFF scheme}

\subsection{Supporting results}

\begin{lem}[BDG-type inequality, \citealt{pannier2019uniqueness}]\label{lem:bdg_volterra}
Let $p>2$ and $\mathbb{T}\subset \R_+$. Let $K:\mathbb{T}^2\to\mathbb{R}$ be a kernel and let $u$ be progressively measurable such that the process $M^K(u)$ defined as:
\[
M_t^K(u) := \int_0^t K(t,s)u(s)\,dW_s.
\]
is in $L^1$. Then,
\begin{equation}\label{eq:bdg_volterra_1}
\mathbb{E}\!\left[\sup_{t\in\mathbb{T}} \left| M_t^K(u) \right|^p \right]
\;\lesssim_p\;
\mathbb{E}\!\left[\left( \int_0^t |K(t,s)u(s)|^2\,ds \right)^{p/2}\right],
\end{equation}
If moreover $K(t,\cdot)\in L^{p-2}([0,t])$ for all $t\in\mathbb{T}$, then
\begin{equation}\label{eq:bdg_volterra_2}
\mathbb{E}\!\left[\left( \int_0^t |K(t,s)u(s)|^2\,ds \right)^{p/2}\right]
\;\le\;
C_{T,p}\,
\mathbb{E}\!\left[\int_0^t |u(s)|^p\,ds\right],
\end{equation}
where
\[
C_{T,p}
:=
\sup_{t\in\mathbb{T}}
\|K(t,\cdot)\|_{L^{p-2}([0,t])}^{p/2}
<\infty.
\]
\end{lem}

\begin{prop}[\citet{vonbahr1965inequalities} inequality]\label{prop:vBE}
Let $1 \le r \le 2$ and let $(X_i)_{i\ge1}$ be independent random variables with
$\mathbb{E}[X_i]=0$ and $\mathbb{E}|X_i|^r<\infty$. Then
\begin{eqnarray}\label{eq:vBE}
\mathbb{E}\left|\sum_{i=1}^n X_i\right|^r
\;\le\;
C_r \sum_{i=1}^n \mathbb{E}|X_i|^r,
\end{eqnarray}
where the constant $C_r>0$ depends only on $r$.
\end{prop}

\begin{prop}[\citet{rosenthal1970subspaces} inequality]\label{prop:rosenthal}
Let $p \ge 2$ and let $(X_i)_{i=1}^n$ be independent, mean-zero random variables with
$\mathbb{E}|X_i|^p < \infty$. Then there exists a constant $C_p>0$, depending only on $p$, such that
\begin{eqnarray}\label{eq:rosenthal}
\mathbb{E}\left|\sum_{i=1}^n X_i\right|^p
\;\le\;
C_p\left(
\sum_{i=1}^n \mathbb{E}|X_i|^p
+
\left(\sum_{i=1}^n \mathbb{E}|X_i|^2\right)^{p/2}
\right).
\end{eqnarray}
\end{prop}

\subsection{Proof of Theorem~\ref{thm:quantRFFestimatemoment}}
\label{subsec:proof_rff_estimate_moment}

\begin{enumerate}
    \item 
Let us reason on $X_n$ as the proof is similar. It follows from H\"older's inequality that for any $t\in \mathcal{K}$ :
\[
\mathbb{E}\big[|X_n(t)|^p\big]
\le  \mathbb{E}\big[|X_0|^p\big]
+  \mathbb{E}\!\left[
  \left| \int_0^t K\big(t-\zeta_n(s)\big)\,
  \sigma\big(\zeta_n(s), X_n(\zeta_n(s))\big)\, dW_s
  \right|^p
\right].
\]
The Burkholder--Davis--Gundy inequality (Lemma \ref{lem:bdg_volterra}) together with Theorem~2.3 in \cite{pannier2019uniqueness} yield  that there exists a positive constant $C>0$ such that
\[
\mathbb{E}\big[|X_n(t)|^p\big]
\le C \mathbb{E}\big[|X_0|^p\big]
+ C \mathbb{E}\!\left[
  \left( \int_0^t
    \big|K\big(t-\zeta_n(s)\big)\,
    \sigma\big(\zeta_n(s), X_n(\zeta_n(s))\big)\big|^2 ds
  \right)^{\frac{p}{2}}
\right].
\]
Applying H\"older's inequality for any $\beta>0$, 
\begin{align*}
 & \Bigg(\int_{0}^{t}\big|K(t-\zeta_{n}(s))\,\sigma(\zeta_{n}(s),X_{n}(\zeta_{n}(s)))\big|^{2}\,ds\Bigg)^{\frac{p}{2}}\\
 & \le\Bigg(\int_{0}^{t}|K(t-\zeta_{n}(s))|^{2\beta}\,ds\Bigg)^{\frac{p}{2\beta}}\Bigg(\int_{0}^{t}|\sigma(\zeta_{n}(s),X_{n}(\zeta_{n}(s)))|^{\frac{2\beta}{\beta-1}}\,ds\Bigg)^{\frac{p(\beta-1)}{2\beta}}.
\end{align*}
By the linear growth condition on $\sigma$, there exists an arbitrary constant $C>0$ such that:
\begin{align*}
 & \Bigg(\int_{0}^{t}\big|K(t-\zeta_{n}(s))\,\sigma(\zeta_{n}(s),X_{n}(\zeta_{n}(s)))\big|^{2}\,ds\Bigg)^{\frac{p}{2}}\\
 & \le C\Bigg(\int_{0}^{t}|K(t-\zeta_{n}(s))|^{2\beta}\,ds\Bigg)^{\frac{p}{2\beta}}\Bigg(\int_{0}^{t}\big(1+|X_{n}(\zeta_{n}(s))|^{\frac{2\beta}{\beta-1}}\big)\,ds\Bigg)^{\frac{p(\beta-1)}{2\beta}},
\end{align*}
In particular,
\begin{align*}
 & \Bigg(\int_{0}^{t}\big|K(t-\zeta_{n}(s))\,\sigma(\zeta_{n}(s),X_{n}(\zeta_{n}(s)))\big|^{2}\,ds\Bigg)^{\frac{p}{2}}\\
 & \le C\Bigg(\int_{0}^{t}|K(t-\zeta_{n}(s))|^{2\beta}\,ds\Bigg)^{\frac{p}{2\beta}}\Bigg(\int_{0}^{t}\big(1+|X_{n}(\zeta_{n}(s))|^{\frac{2\beta}{\beta-1}}\big)\,ds\Bigg)^{\frac{p(\beta-1)}{2\beta}},
\end{align*}
Again, using H\"older's inequality, 
\[
\begin{aligned}
\begin{cases}
    \Bigg(\int_0^t \big(1+|X_n(\zeta_n(s))|^{\frac{2\beta}{\beta-1}}\big)\,ds\Bigg)^{\frac{p(\beta-1)}{2\beta}}
\le t^{\frac{p(\beta-1)}{2\beta}-1}\int_0^t \big(1+|X_n(\zeta_n(s))|^{p}\big)\,ds.\\

\Bigg(\int_0^t |K(t-\zeta_n(s))|^{2\beta}\,ds\Bigg)^{\frac{p}{2\beta}}
\le  t^{\frac{p}{2\beta}-1}\int_0^t |K(t-\zeta_n(s))|^{p}\,ds.
\end{cases}
\end{aligned}
\]
meaning that:
\[
\begin{aligned}
\mathbb{E}\Bigg[
\Big( \int_0^t &\big|K(t-\zeta_n(s))\,\sigma(\zeta_n(s), X_n(\zeta_n(s)))\big|^2 \, ds \Big)^{\frac{p}{2}}
\Bigg]\le C\, t^{\frac{p}{2}-1}
\mathbb{E}\Bigg[\Bigg(\int_0^t |K(t-\zeta_n(s))|^{p}\,ds.\Bigg)
\int_0^t \Big(1+|X_n(\zeta_n(s))|^{p}\Big)ds\Bigg]\,
\end{aligned}
\]
This leads to:
\[
\mathbb{E}\big[|X_n(t)|^p\big]
\le C \left(\left(\mathbb{E}\big[|X_0|^p\big]
+ t^{\frac{p}{2}}K(0)^p\right)+\int_0^t
\mathbb{E}\Bigg[|X_n(\zeta_n(s))|^{p}\Bigg]\,ds\right).
\]
Applying Gr\"onwall's lemma yields
\[
\mathbb{E}\big[|X_n(t)|^p\big]
\le C\left(\mathbb{E}\big[|X_0|^p\big] + t^{\frac{p}{2}}K(0)^p\right)e^{Ct}.
\]
which leads to the claimed result.

\item 
For any arbitrary $0\leq s \leq t$, we write
\begin{align*}
X^M_n(t) - X^M_n(s) 
&= \int_s^t \hat{K}_M(t-\zeta_n(u)) \, \sigma(\zeta_n(u), X^M_n(\zeta_n(u))) \, dW_u \\
&\quad + \int_0^s \Big[ \hat{K}_M(t-\zeta_n(u)) - \hat{K}_M(s-\zeta_n(u)) \Big] 
      \sigma(\zeta_n(u), X^M_n(\zeta_n(u))) \, dW_u \\
&=: I_1 + I_2.
\end{align*}
For $p > 0$, according to the Burkholder--Davis--Gundy inequality (Lemma \ref{lem:bdg_volterra}) together with Theorem~2.3 in \cite{pannier2019uniqueness}, there exists $C_p>0$ such that
\begin{eqnarray*}
\begin{cases}
    \mathbb{E}[|I_1|^p] 
    & \le C_p \, \mathbb{E}\Bigg[\Big(\int_s^t |\hat{K}_M(t- \zeta_n(u))|^2 \, 
      |\sigma(\zeta_n(u), X^M_n(\zeta_n(u)))|^2 \, du \Big)^{p/2}\Bigg], \\
      \mathbb{E}[|I_2|^p] 
&\le C_p \, \mathbb{E}\Bigg[\Big(\int_0^s |\hat{K}_M(t-\zeta_n(u)) - \hat{K}_M(s-\zeta_n(u))|^2 \,
      |\sigma(\zeta_n(u), X^M_n(\zeta_n(u)))|^2 \, du \Big)^{p/2}\Bigg].
\end{cases}
\end{eqnarray*}
Thanks to the linear growth assumption, there exists, after some algebra, a positive constant $C_p>0$ such that:
\begin{align*}
\mathbb{E}[|I_1|^p] 
\le C_p|t-s|^{p/2}.
\end{align*}
Besides, by direct computations, one has for any $s\le t$:
\begin{eqnarray*}
    \left|\hat{K}_M(t-\zeta_n(u)) - \hat{K}_M(s-\zeta_n(u))\right|=2\mathbb{E}\left[\left| \sin\!\left(\eta \left(\frac{t + s}{2}-\zeta_n(u) - s\right)\right)
      \sin\!\left(\eta \left(\frac{t - s}{2}\right)\right)\right| \right] \le 2
\end{eqnarray*}
and similarly using the linear growth assumption, there exists a positive constant $C_p>0$ such that:
\begin{align*}
\mathbb{E}[|I_2|^p] 
\le C_p.
\end{align*}
On the other hand, using Minkowski's inequality:
\[
\mathbb{E}[|X^M_n(t) - X^M_n(s)|^p] 
\le C_p \Big( |t-s|^{p/2} +1 \Big)
\]
\end{enumerate}
which ends the proof.\\
\cqfd

\subsection{Proof of Theorem~\ref{thm:eulerXNXNMRFFerrorbound}}
\label{subsec:proof_euler_xnxnm_rff_error_bound}

For simplicity assume \(X_0^M=X_0\) (the deterministic initial difference is handled similarly).\\
By some algebra, the following decomposition holds:
\[
\begin{aligned}
X_N(t)-X_N^M(t)
&= \int_0^t K(t-\zeta_N(s))\Big(\sigma(\zeta_N(s),X_N(\zeta_N(s)))-\sigma(\zeta_N(s),X_N^M(\zeta_N(s)))\Big)\,dW_s\\[4pt]
&\qquad
+ \int_0^t \big(K(t-\zeta_N(s))-\widehat K_M(t-\zeta_N(s))\big)\,
\sigma(\zeta_N(s),X_N^M(\zeta_N(s)))\,dW_s\\[2pt]
&=: I_t^{(1)}+I_t^{(2)}.
\end{aligned}
\]
Let \(p\ge2\). By the Burkholder--Davis--Gundy inequality (Lemma \ref{lem:bdg_volterra}) together with Theorem 2.3 in \cite{pannier2019uniqueness}, there exists \(C>0\) such that
\[
\mathbb{E}[|X_N(t)-X_N^M(t)|^p]
\le C\,\mathbb{E}\big[\langle I^{(1)}\rangle_t^{p/2}\big]
+ C\,\mathbb{E}\big[\langle I^{(2)}\rangle_t^{p/2}\big].
\]
Using the Lipschitz property of \(\sigma\) in space, we have:
\[
\begin{aligned}
\mathbb{E}\big[\langle I^{(1)}\rangle_t^{p/2}\big]
\le L\,\mathbb{E}\Bigg[\Big( \int_0^t |K(t-\zeta_N(s))|^2 |X_N(\zeta_N(s))-X_N^M(\zeta_N(s))|^2 ds \Big)^{\frac{p}{2}}\Bigg].
\end{aligned}
\]
Apply H\"older's inequality in time for any exponent arbitrary \(\beta>1\) gives:
\begin{eqnarray*}
    \mathbb{E}\Bigg[\Big( \int_0^t |K(t-\zeta_N(s))|^2 |X_N(\zeta_N(s))-X_N^M(\zeta_N(s))|^2 ds \Big)^{\frac{p}{2}}\Bigg]
\le \Big(\int_0^t |K(t-\zeta_N(s))|^{2\beta}\,ds\Big)^{\frac{p}{2\beta}}\\
\mathbb{E}\Bigg[\Big(\int_0^t |X_N(\zeta_N(s))-X_N^M(\zeta_N(s))|^{\frac{2\beta}{\beta-1}}\,ds\Big)^{\frac{p(\beta-1)}{2\beta}}\Bigg].
\end{eqnarray*}
As the random Fourier features representation of \(K\) is bounded:
\[
\mathbb{E}\big[\langle I^{(1)}\rangle_t^{p/2}\big]
\le \int_0^t \mathbb{E}\big[|X_N(\zeta_N(s))-X_N^M(\zeta_N(s))|^p\big]\,ds,
\]
hence
\[
\mathbb{E}\big[|I_t^{(1)}|^p\big]\le C\int_0^t \mathbb{E}\big[|X_N(\zeta_N(s))-X_N^M(\zeta_N(s))|^p\big]\,ds.
\]
in particular:
\[
\mathbb{E}\big[|I_t^{(1)}|^p\big]\le C\int_0^t \sup_{0\leq u\leq s}\mathbb{E}\big[|X_N(u)-X_N^M(u)|^p\big]\,ds.
\]
On the other hand, we estimate the quadratic variation:
\[
\begin{aligned}
\mathbb{E}\big[\langle I^{(2)}\rangle_t^{p/2}\big]
&= \mathbb{E}\Bigg[\Big( \int_0^t |K(t-\zeta_N(s))-\widehat K_M(t-\zeta_N(s))|^2
|\sigma(\zeta_N(s),X_N^M(\zeta_N(s)))|^2 ds \Big)^{\frac{p}{2}}\Bigg].
\end{aligned}
\]
Applying H\"older's inequality with exponent \(\beta>1\) gives:
\[
\begin{aligned}
\mathbb{E}\big[\langle I^{(2)}\rangle_t^{p/2}\big]
&\le
\mathbb{E}\big[\Big(\int_0^t |K-\widehat K_M|^{2\beta}(t-\zeta_N(s))\,ds\Big)^{\frac{p}{2\beta}}
\Big(\int_0^t |\sigma(\zeta_N(s),X_N^M(\zeta_N(s)))|^{\frac{2\beta}{\beta-1}}\,ds\Big)^{\frac{p(\beta-1)}{2\beta}}\big].
\end{aligned}
\]
Using the linear growth property together with the time-averaging inequality (i.e., for any continuous function $a$ and $t\geq 0$, \((\int_0^t a(x) dx)^{r}\le t^{r-1}\int_0^t a(x)^r dx\) with
\(r=\dfrac{p(\beta-1)}{2\beta}\ge1\), one has:
\[
\Big(\int_0^t |\sigma(\zeta_N(s),X_N^M(\zeta_N(s)))|^{\frac{2\beta}{\beta-1}} ds\Big)^{\frac{p(\beta-1)}{2\beta}}
\le C \int_0^t \big(1+|X_N^M(\zeta_N(s))|^p\big)\,ds.
\]
Thus as the random features are independent from the driving Brownian motion $W$, one has:
\[
\begin{aligned}
\mathbb{E}\big[\langle I^{(2)}\rangle_t^{p/2}\big]
&\le C_t\mathbb{E}\big[\Big(\int_0^t |K-\widehat K_M|^{2\beta}(t-\zeta_N(s))\,ds\Big)^{\frac{p}{2\beta}}\big]
\int_0^t \mathbb{E}\big[1+|X_N^M(\zeta_N(s))|^p\big]\,ds.
\end{aligned}
\]
Again apply the time-averaging trick to the first factor to get the convenient bound:
\[
\Big(\int_0^t |K-\widehat K_M|^{2\beta}(t-\zeta_N(s))\,ds\Big)^{\frac{p}{2\beta}}
\le C_t \int_0^t |K(t-\zeta_N(s))-\widehat K_M(t-\zeta_N(s))|^{p}\,ds.
\]
Consequently
\[
\mathbb{E}\big[|I_t^{(2)}|^p\big]
\le C\mathbb{E}\Bigg(\int_0^t |K(t-\zeta_N(s))-\widehat K_M(t-\zeta_N(s))|^{p}\,ds\Bigg)
\Bigg(1+\sup_{u\in[0,t]}\mathbb{E}\big[|X_N^M(u)|^p\big]\Bigg).
\]
Using the moment estimate, there exists a positive constant $C(\mathcal{K})>0$ for any arbitrary compact $\mathcal{K}\subset \R_+$ such that:
\[
\mathbb{E}\big[|I_t^{(2)}|^p\big]
\le C(\mathcal{K})
\,\Big(\mathbb{E}\!\left[\,|X^M_0|^p\,\right] + |K(0)|^p\Big)\mathbb{E}\Bigg(\int_0^t |K(t-\zeta_N(s))-\widehat K_M(t-\zeta_N(s))|^{p}\,ds\Bigg).
\]
Using Fubini's theorem, 
\[
\mathbb{E}\big[|I_t^{(2)}|^p\big]
\le C(\mathcal{K})
\,\Big(\mathbb{E}\!\left[\,|X^M_0|^p\,\right] + |K(0)|^p\Big)\int_0^t\mathbb{E}\Bigg( |K(t-\zeta_N(s))-\widehat K_M(t-\zeta_N(s))|^{p}\Bigg)\,ds.
\]
\begin{enumerate}[label=$\blacktriangleright$]
    \item For $1\leq p<2$, the \textit{Von Bahr-Essen inequality} (Proposition~\ref{prop:vBE})  gives for any $0\leq s \leq t$, after some algebra:
\[
\mathbb{E}\Bigg( |K(t-\zeta_N(s))-\widehat K_M(t-\zeta_N(s))|^{p}\Bigg)
\le C_p M^{-(p - 1)}.
\]
where $C_p$ is a positive constant. Consequently:
\[
\sup_{s\in[0, t]}\mathbb{E}[|X_N(s)-X_N^M(s)|^p]
\le C_\mathcal{K}^p\left(\int_0^t \sup_{0\leq u\leq s}\mathbb{E}\big[|X_N(u)-X_N^M(u)|^p\big]\,ds+\frac{1}{M^{p-1}}\right).
\]
where $C_\mathcal{K}^p=4C_pC(\mathcal{K})
\,\Big(\mathbb{E}\!\left[\,|X^M_0|^p\,\right] + |K(0)|^p\Big)\vee C_t$.
Hence, we obtain thanks to Gr\"onwall's lemma:
\[
\sup_{s \in [0, t]} \mathbb{E}\big[|X_N(s) - X_N^M(s)|^p\big]
\le \frac{C_\mathcal{K}^p}{M^{p-1}} e^{C_\mathcal{K}^p t}.
\]
\item For $p\geq 2$, \textit{Rosenthal's inequality} (Proposition~\ref{prop:rosenthal}) states that there exists a constant $C_p$ such that for any $0\leq s \leq t$:
\begin{align*}
\mathbb{E}\left(\left|
  \sum_{m=1}^M K(t-\zeta_N(s)) - \cos(\eta_{m}(t-\zeta_N(s)))\right|^{p}
\right)
&\le C_pK(0) \Bigg(
    \sum_{m=1}^M 
      \mathbb{E}\big[
        |K(t-\zeta_N(s)) -  \cos(\eta_{m}(t-\zeta_N(s)))|^p
      \big] \\
&\qquad\quad
    + \Bigg(
        \sum_{m=1}^M 
          \mathbb{E}\big[
            |K(t-\zeta_N(s)) - \cos(\eta_{m}(t-\zeta_N(s)))|^2
          \big]
      \Bigg)^{p/2}
  \Bigg).
\end{align*}
thanks to the random Fourier features representation, one can have the upper bound after some algebra:
\begin{align*}
\mathbb{E}\left(\left|
  \sum_{m=1}^M K(t-\zeta_N(s)) - \cos(\eta_{m}(t-\zeta_N(s)))\right|^{p}
\right)
&\le C_p \Bigg(
    2M+\left(
        4M
      \right)^{p/2}
  \Bigg).
\end{align*}
Consequently,
\begin{align*}
\mathbb{E}\Bigg( |K(t-\zeta_N(s))-\widehat K_M(t-\zeta_N(s))|^{p}\Bigg)
&\le C_p \Bigg(
    \frac{2}{M^{p-1}}+\left(\frac{4^p}{M}\right)^{\frac{1}{2}}
  \Bigg).
\end{align*}
This leads to:
\[
\sup_{s\in[0, t]}\mathbb{E}[|X_N(s)-X_N^M(s)|^p]
\le C_\mathcal{K}^p\left(\int_0^t \sup_{0\leq u\leq s}\mathbb{E}\big[|X_N(u)-X_N^M(u)|^p\big]\,ds+\left(\frac{2}{M^{p-1}}+\left(\frac{4^p}{M}\right)^{\frac{1}{2}}\right)\right).
\]
where $C_\mathcal{K}^p=C(\mathcal{K})
\,\Big(\mathbb{E}\!\left[\,|X^M_0|^p\,\right] + |K(0)|^p\Big)\vee C_p$.\\
Applying Gr\"onwall's lemma yields:
\begin{align*}
\sup_{s\in[0, t]}\mathbb{E}[|X_N(s)-X_N^M(s)|^p]\le  C_\mathcal{K}^p \Bigg( \frac{2}{M^{p-1}} + \Big(\frac{4^p}{M}\Big)^{1/2} \Bigg) e^{C_\mathcal{K}^p t}.
\end{align*}
\end{enumerate}
which ends the proof.
\\
\cqfd

\subsection{Proof of Proposition~\ref{prop:eulerXXNMRFFerrorbound}}

Applying Minkowski's inequality:
\[
\Bigg(\sup_{s\in[0,t]} \mathbb{E}|X_s-X_N^M(s)|^p \Bigg)^{1/p}
\le 
\Bigg(\sup_{s\in[0,t]} \mathbb{E}|X_s-X_N(s)|^p \Bigg)^{1/p}
+
\Bigg(\sup_{s\in[0,t]} \mathbb{E}|X_N(s)-X_N^M(s)|^p \Bigg)^{1/p}.
\]
Using Jensen's inequality,
\[
\sup_{s\in[0,t]} \mathbb{E}|X_s-X_N^M(s)|^p
\le 2^{p-1} \Bigg(
\sup_{s\in[0,t]} \mathbb{E}|X_s-X_N(s)|^p
+
\sup_{s\in[0,t]} \mathbb{E}|X_N(s)-X_N^M(s)|^p
\Bigg).
\]
The error bound of Theorem~\ref{thm:eulerXNXNMRFFerrorbound}
as well as Theorem~2.2 (ii) in \citet{richard2021discrete}
claim the existence of positive constants $K^1_p$ and $K^2_p$ such that
\begin{eqnarray*}
\sup_{s\in[0,t]} \mathbb{E}|X_s-X_N^M(s)|^p
\le 
    \begin{cases}
        K^1_p\Bigg(
 \,\Big(1 + \mathbb{E}\left[|X_0|^p\right] \Big)\, \delta_N^{\,p(\alpha \wedge 1) - \varepsilon}
+
\frac{e^{C_K^p t}}{M^{p-1}} 
\Bigg), & \text{if $1 \leq p<2$}, \\[6pt]
K^2_p \Bigg(
\,\Big(1 + \mathbb{E}\left[|X_0|^p\right] \Big)\, \delta_N^{\,p(\alpha \wedge 1) - \varepsilon}
+ \Bigg( \frac{2}{M^{p-1}} + \Big(\frac{4^p}{M}\Big)^{\frac{1}{2}} \Bigg) e^{C_K^p t}, 
& \text{if $p \geq 2$}.
    \end{cases}
\end{eqnarray*}
which ends the proof.\\
\cqfd

\section{Error of the RFF scheme: the S-fBM kernel}
\subsection{Proof of Theorem~\ref{thm:convergenceHMCspectralS-fBM}}
\label{subsec:proof_convergence_hmc_spectral_sfbm}

By stationarity one has from Lemma~\ref{lem:covstatiocovHMC} that:
\[
\forall (i,j)\in \mathbb{N}, \quad 
\big|\Cov(h(\hat{X}_i),h(\hat{X}_j))\big|
\le C\,(b-a)^2\,\rho^{\big|i-j\big|}.
\]
($C$ is a positive constant).\\
Let \( Z_i = h(\hat{X}_i) - \mathbb{E}[h(\hat{X}_i)] \) so that \( S_M = \sum_{i=1}^M Z_i \). Each \( Z_i \) satisfies \( |Z_i| \le  b-a \) and by assumption
\[
|\Cov(h(\hat{X}_i), h(\hat{X}_j))| \le (b-a)^2\,C\,\rho^{|i-j|}, \qquad 0<\rho<1.
\]
We first bound the variance of \(S_M\):
\[
\Var(S_M)
= \sum_{i=1}^M \Var(Z_i)
+ 2\sum_{1 \le i < j \le M} \Cov(Z_i, Z_j).
\]
By the covariance bound,
\[
\Var(S_M)
\le M \sup_i \Var(Z_i)
+ 2 (b-a)^2 C \sum_{i=1}^{M} \sum_{k=1}^{M-i} \rho^{k}
= M \sup_i \Var(Z_i)
+ 2 (b-a)^2 C \sum_{k=1}^{M-1} (M-k)\rho^{k}.
\]
Since
\[
\sum_{k=1}^{M-1} (M-k)\rho^{k}
\le M \sum_{k=1}^{\infty} \rho^{k}
= M \frac{\rho}{1-\rho},
\]
we obtain
\[
\Var(S_M)
\le M \left(
\sup_{1\leq i\leq M} \Var(Z_i)
+ 2(b-a)^2 C \frac{\rho}{1-\rho}
\right).
\]
By Popoviciu’s inequality on variances, 
\[
\Var(S_M)
\le M \left(
 \frac{(b-a)^2}{4}
+ 2(b-a)^2 C \frac{\rho}{1-\rho}
\right).
\]
We define
\[
v := \frac{(b-a)^2}{4} \left(1
+\frac{8\rho}{1-\rho} \right),
\]
in order to introduce compactly the positive constant $C$ such that \( \Var(S_M) \le C M  v \) . Next, recall a standard Bernstein-type moment generating function inequality for \(S_M\):
\[
\log \mathbb{E} e^{u S_M}
\le C u^2 \Var(S_M).
\]
Thus, we obtain
\[
\log \mathbb{E} e^{u S_M}
\le C u^2 M v.
\]
Consider now the normalized sum
\[
T_M := \frac{1}{M}\sum_{i=1}^M \big(h(Y_i) - \mathbb{E}[h(Y_i)]\big)
= \frac{S_M}{M}.
\]
One straightforwardly has following the same previous computations:
\[
\Var(T_M)
\le \frac{v}{M},
\]
Applying the same Bernstein-type moment generating function inequality as before to \(T_M\) leads to:
\[
\log \mathbb{E} e^{u T_M}
\le C u^2 \frac{v}{M}.
\]
Hence,
\[
\log\!\left [\mathbb{E} \!\left( \exp\!\left(
\,\frac{u}{M}\sum_{i=1}^M \big(h(Y_i) - \mathbb{E}\left[h(\eta)\right]\big)
\right)\right)\right]
\le
C\, u^2 \frac{v}{M}.
\]
Similarly,
\[
\log\!\left [\mathbb{E} \!\left( \exp\!\left(
\,\frac{u}{M}\sum_{i=1}^M \big(h(\hat{X}_i) - \mathbb{E}\left [ h(\hat{X}_0)\right]\big)
\right)\right)\right]
\le
C\, u^2  \frac{v}{M}.
\]
By the propagation of errors, one can write:
\[
\log\!\left [\mathbb{E} \!\left( \exp\!\left(
u\,\frac{1}{M}\sum_{i=1}^M \big(h(\hat{X}_i) - \mathbb{E}\left[h(\eta)\right]\big)
\right)\right)\right]
\le
C\, u^2  \frac{v}{M}+\left|\mathbb{E}\left [ h(\hat{X}_0)\right]-\mathbb{E}\left[h(\eta)\right]\right|.
\]
As \(h\) is bounded, by fixing an arbitrary positive real \(R\) one can write using the triangle inequality:
\[
\left|\mathbb{E}\left [ h(\hat{X}_0)\right]-\mathbb{E}\left[h(\eta)\right]\right|\leq  
C\!\left( \int_{B(0,R)} \!\!\left|f_\omega(\mathbf{x}) - \hat f_\omega^N(\mathbf{x}) \right| d\mathbf{x}
+ \| f_\omega - \hat f_\omega^N \|_1^{>r} \right).
\]
Thanks to Theorem~\ref{thm:convergenceHMCspectralS-fBM} as well as Proposition~\ref{prop:queuelonenormerror}, one has:
\[
\left| \mathbb{E}\left [ h(\hat{X}_0)\right] - \mathbb{E}[h(\eta)] \right|
\le
C \Psi^R_N(T,\nu,d,H,C_{\Psi}),
\]
where 
\begin{align*}
\Psi^R_N(T,\nu,d,H,C_{\Psi}):=\Bigg[
&\left( \frac{T^2}{2} \right)^{\!\frac{d}{2}}
\frac{\nu^2}{\Gamma\!\left( \tfrac{d}{2} \right)}
\Bigg(
    \frac{1}{d}
    \left( \frac{e T^2 R^2}{4N} \right)^{\!N}
    \frac{1}{e \, \bigl( \tfrac{d}{2} + 1 \bigr)_N}
\\[0.5em]
&
    +\;
    \frac{1}{d + 2H}
    \frac{C_{\Psi}}{
        e \, \bigl( 2 C_{\Psi} R e (N-1) - 1 \bigr)
        (N-1)^{2 C_{\Psi} R e (N-1) - 1}
    }
\Bigg)
\\[0.5em]
&
+\;
\frac{\Omega_d \, C_\omega}{T^{d}}
\frac{(T R + \kappa(d))^{d} - T R^{d}}{d \, \kappa(d)} \,
e^{-T R}
\Bigg].
\end{align*}
This is true for any $R>0$ in particular $R=\sqrt{N}$, which means that:
\[
\left| \mathbb{E}\left [ h(\hat{X}_0)\right] - \mathbb{E}\left[h(\eta)\right] \right|
\le
C \Psi_N(T,\nu,d,H,C_{\Psi}),
\]
where $\Psi_N(T,\nu,d,H,C_{\Psi})=\Psi^{\sqrt{N}}_N(T,\nu,d,H,C_{\Psi})$
Putting the terms together, one has:
\[
\log\!\left [\mathbb{E} \!\left( \exp\!\left(
\,\frac{u}{M}\sum_{i=1}^M \big(h(\hat{X}_i) - \mathbb{E}\left[h(\eta)\right]
\right)\right)\right]
\le
C \!\left(  \frac{2vu^2}{M}+\Psi_N(T,\nu,d,H,C_{\Psi})\right),
\]
Using Markov's inequality, for any $t\geq 0$:
\[
\mathbb{P}\!\left(
|\hat{\mu}_M - \mu| \geq t
\right)
\leq
\,\exp\!\left(
- u t
+ C \!\left(  \frac{vu^2}{M} + \Psi_N(T,\nu,d,H,C_{\Psi})\right)
\right),
\]
Similarly to Hoeffding's inequality (by minimizing the exponent), this holds in particular for:
\[
u^\star = \frac{M t}{4 C v}.
\]
which leads to:
\[
\mathbb{P}(|\hat{\mu}_M - \mu| \ge t)
\le  \exp\!\Big( - \frac{M t^2}{8 C v} + C \Psi_N(T,\nu,d,H,C_{\Psi}) \Big).
\]\\
\cqfd

\subsection{Proof of Corollary~\ref{cor:lperrorHMC}}
\label{subsec:proof_cor_lp_error_hmc}
  
For the sake of compactness, define:
\[
\alpha := \frac{1}{8 C v},\qquad \beta := C \Psi_N(T,\nu,d,H,C_{\Psi}),
\]
so that the tail bound is
\[
\mathbb{P}(|\hat{\mu}_M - \mu| \ge t) \le e^{\beta} e^{-\alpha M t^2},\qquad t\ge0.
\]
For any real \(p>0\) we use the identity
\[
\mathbb{E}\big[|\hat{\mu}_M-\mu|^p\big] = \int_{0}^{\infty} p\,t^{\,p-1}\,\mathbb{P}(|\hat{\mu}_M-\mu| \ge t)\,dt,
\]
Hence
\begin{align*}
\mathbb{E}\big[|\hat{\mu}_M-\mu|^p\big]
&= \int_{0}^{\infty} p\,t^{p-1}\,\mathbb{P}(|\hat{\mu}_M-\mu|\ge t)\,dt \\
&\le e^{\beta}\, p \int_{0}^{\infty} t^{p-1} e^{-\alpha M t^2}\,dt.
\end{align*}
The remaining integral is standard: for \(b>0\),

\[
\int_{0}^{\infty} t^{p-1} e^{-b t^2}\,dt
= \tfrac12 b^{-p/2}\, \Gamma\!\big(\tfrac{p}{2}\big).
\]
We obtain
\[
\mathbb{E}\big[|\hat{\mu}_M-\mu|^p\big]
\le e^{\beta}\, p \cdot \tfrac12 (\alpha M)^{-p/2} \Gamma\!\big(\tfrac{p}{2}\big).
\]
Using the identity \(\Gamma(1+s)=s\Gamma(s)\), this can be written compactly as
\[
\mathbb{E}\big[|\hat{\mu}_M-\mu|^p\big]
\le e^{\beta}\,\Gamma\!\Big(1+\frac{p}{2}\Big)\,(\alpha M)^{-p/2},
\]
meaning that:
\[
\mathbb{E}\big[|\hat{\mu}_M-\mu|^p\big]
\le e^{C \Psi_N(T,\nu,d,H,C_{\Psi})}\,\Gamma\!\Big(1+\frac{p}{2}\Big)\,
\Big(\frac{1}{8 C v}\,M\Big)^{-p/2}.
\]
which leads to the claimed result.
\\
\cqfd

\subsection{Proof of Theorem~\ref{thm:eulerXNXNMRFFerrorboundSfBM}}
\label{subsec:proof_euler_xnxnm_rff_error_bound_sfbm}

We use the same notations as in the proof of Theorem~\ref{thm:eulerXNXNMRFFerrorbound} (Appendix~\ref{subsec:proof_euler_xnxnm_rff_error_bound}).
We recall that:
\[
\begin{aligned}
X_N(t)-X_N^M(t)
&=I_t^{(1)}+I_t^{(2)}.
\end{aligned}
\]
where:
\begin{eqnarray}
    \begin{cases}
        I_t^{(1)}=\int_0^t K(t-\zeta_N(s))\Big(\sigma(\zeta_N(s),X_N(\zeta_N(s)))-\sigma(\zeta_N(s),X_N^M(\zeta_N(s)))\Big)\,dW_s\\
        I_t^{(2)}=\int_0^t \big(K(t-\zeta_N(s))-\widehat K_M(t-\zeta_N(s))\big)\,
\sigma(\zeta_N(s),X_N^M(\zeta_N(s)))\,dW_s
    \end{cases}
\end{eqnarray}
Following the same arguments as in Appendix~\ref{subsec:proof_euler_xnxnm_rff_error_bound}, one has the bound:
\begin{eqnarray}
\begin{cases}
     \mathbb{E}\big[|I_t^{(1)}|^p\big]\le C\int_0^t \sup_{0\leq u\leq s}\mathbb{E}\big[|X_N(u)-X_N^M(u)|^p\big]\,ds\\
     \mathbb{E}\big[|I_t^{(2)}|^p\big]
\le C(\mathcal{K})
\,\Big(\mathbb{E}\!\left[\,|X^M_0|^p\,\right] + |K(0)|^p\Big)\int_0^t\mathbb{E}\Bigg( |K(t-\zeta_N(s))-\widehat K_M(t-\zeta_N(s))|^{p}\Bigg)\,ds
\end{cases}\nonumber
\end{eqnarray}
Using the result of Corollary \ref{cor:lperrorHMC}, one has:
\begin{eqnarray}
    \mathbb{E}\Bigg( |K(t-\zeta_N(s))-\widehat K_M(t-\zeta_N(s))|^{p}\Bigg)
\le  e^{C \Psi_N(T,\nu,d,H,C_{\Psi})}\,\Gamma\!\Big(1+\frac{p}{2}\Big)\,
\Big(\frac{1}{8 C v}\,M\Big)^{-p/2}.\nonumber
\end{eqnarray}
As a result, 
\[
\sup_{s\in[0, t]}\mathbb{E}[|X_N(s)-X_N^M(s)|^p]
\le C_\mathcal{K}^p\left(\int_0^t \sup_{0\leq u\leq s}\mathbb{E}\big[|X_N(u)-X_N^M(u)|^p\big]\,ds+e^{C\Psi_N(T,\nu,d,H,C_{\Psi})}\,\Gamma\!\Big(1+\frac{p}{2}\Big)\,
\Big(\frac{1}{8 C v}\,M\Big)^{-p/2}\right).
\]
where $C_\mathcal{K}^p=C(\mathcal{K})
\,\Big(\mathbb{E}\!\left[\,|X^M_0|^p\,\right] + |K(0)|^p\Big)\vee C$
Hence, we obtain thanks to Gr\"onwall's lemma:
\[
\sup_{s \in [0, t]} \mathbb{E}\big[|X_N(s) - X_N^M(s)|^p\big]
\le C_\mathcal{K}^p e^{C \Psi_N(T,\nu,d,H,C_{\Psi})+C_\mathcal{K}^p t}\,\Gamma\!\Big(1+\frac{p}{2}\Big)\,
\Big(\frac{1}{8 C v}\,M\Big)^{-p/2} .
\]
Applying Minkowski's inequality:
\[
\Bigg(\sup_{s\in[0,t]} \mathbb{E}\big[|X_s-X_N^M(s)|^p\big] \Bigg)^{1/p}
\le 
\Bigg(\sup_{s\in[0,t]} \mathbb{E}\big[|X_s-X_N(s)|^p\big] \Bigg)^{1/p}
+
\Bigg(\sup_{s\in[0,t]} \mathbb{E}\big[|X_N(s)-X_N^M(s)|^p\big] \Bigg)^{1/p}.
\]
Using Jensen's inequality,
\[
\sup_{s\in[0,t]} \mathbb{E}|X_s-X_N^M(s)|^p
\le 2^{p-1} \Bigg(
\sup_{s\in[0,t]} \mathbb{E}|X_s-X_N(s)|^p
+
\sup_{s\in[0,t]} \mathbb{E}|X_N(s)-X_N^M(s)|^p
\Bigg).
\]
Using the error bound of the classical Euler scheme of  Theorem~2.2 (ii) in \citet{richard2021discrete}, 
there exist positive constants $K^1_p$ and $C_p$ such that
\begin{eqnarray*}
\sup_{s\in[0,t]} \mathbb{E}\big[|X_s-X_N^M(s)|^p\big]
\le  K^1_p\Bigg(
 \,\Big(1 + \mathbb{E}\left[|X_0|^p\right] \Big)\, \delta_N^{\,p(\alpha \wedge 1) - \varepsilon}
+
e^{C \left(\Psi_N(T,\nu,d,H,C_{\Psi})+t\right)}\,\Gamma\!\Big(1+\frac{p}{2}\Big)\,
\Big(\frac{1}{8 C v}\,M\Big)^{-p/2}
\Bigg)
\end{eqnarray*}
without loss of generality, we denote $C^1_p$ as $C_p$ which ends the proof.\\
\cqfd

\section{Numerical appendix}

\subsection{S-fBM spectral density representation}
\label{subsec:numerics_sfbm_spectral_density}

\begin{figure}[H]
    \centering
    % --- Top row ---
    \begin{subfigure}[t]{0.49\textwidth}
        \centering
        \includegraphics[width=\textwidth]{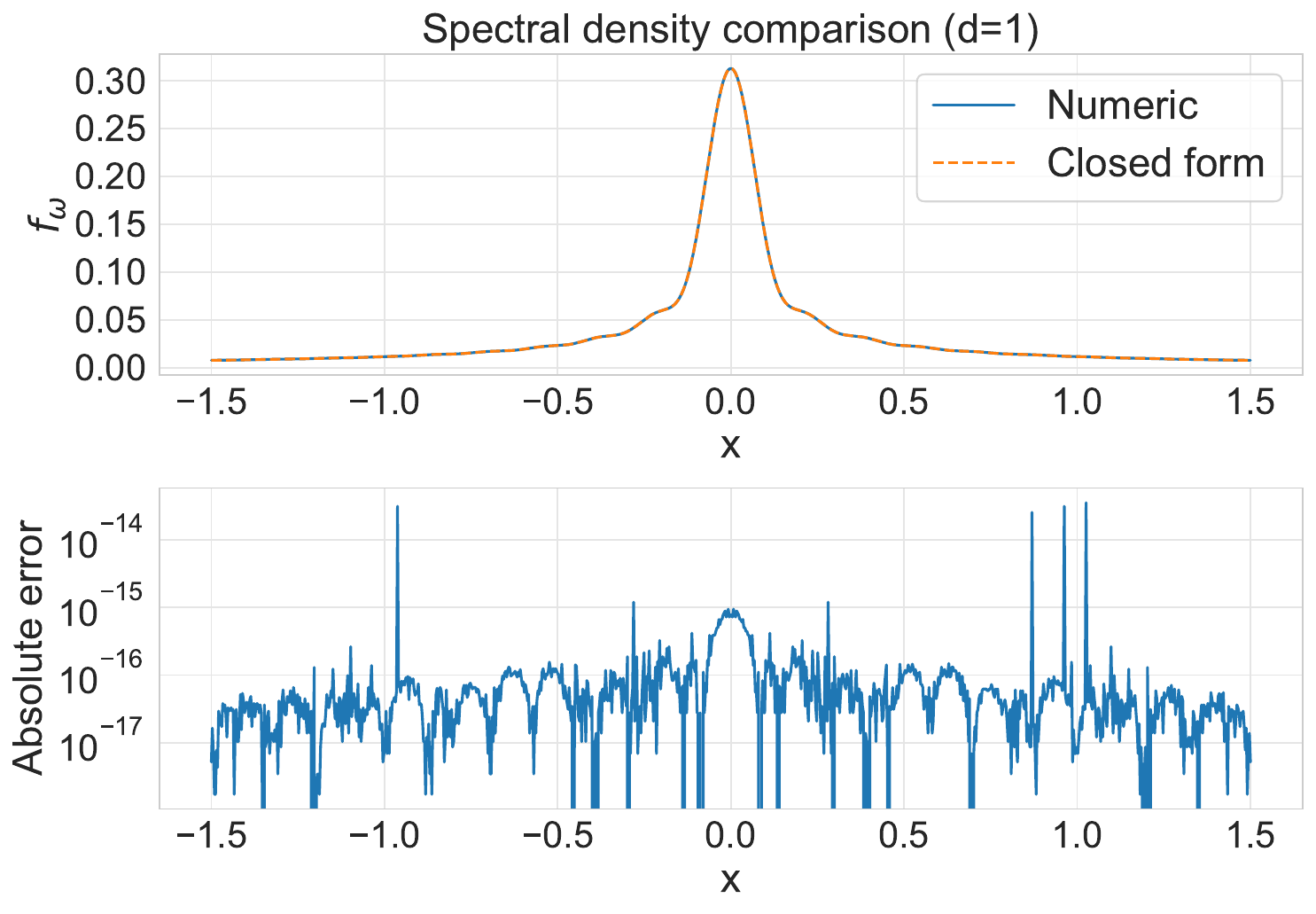}
        \caption{\(T = 40\)}
    \end{subfigure}
    \hfill
    \begin{subfigure}[t]{0.49\textwidth}
        \centering
        \includegraphics[width=\textwidth]{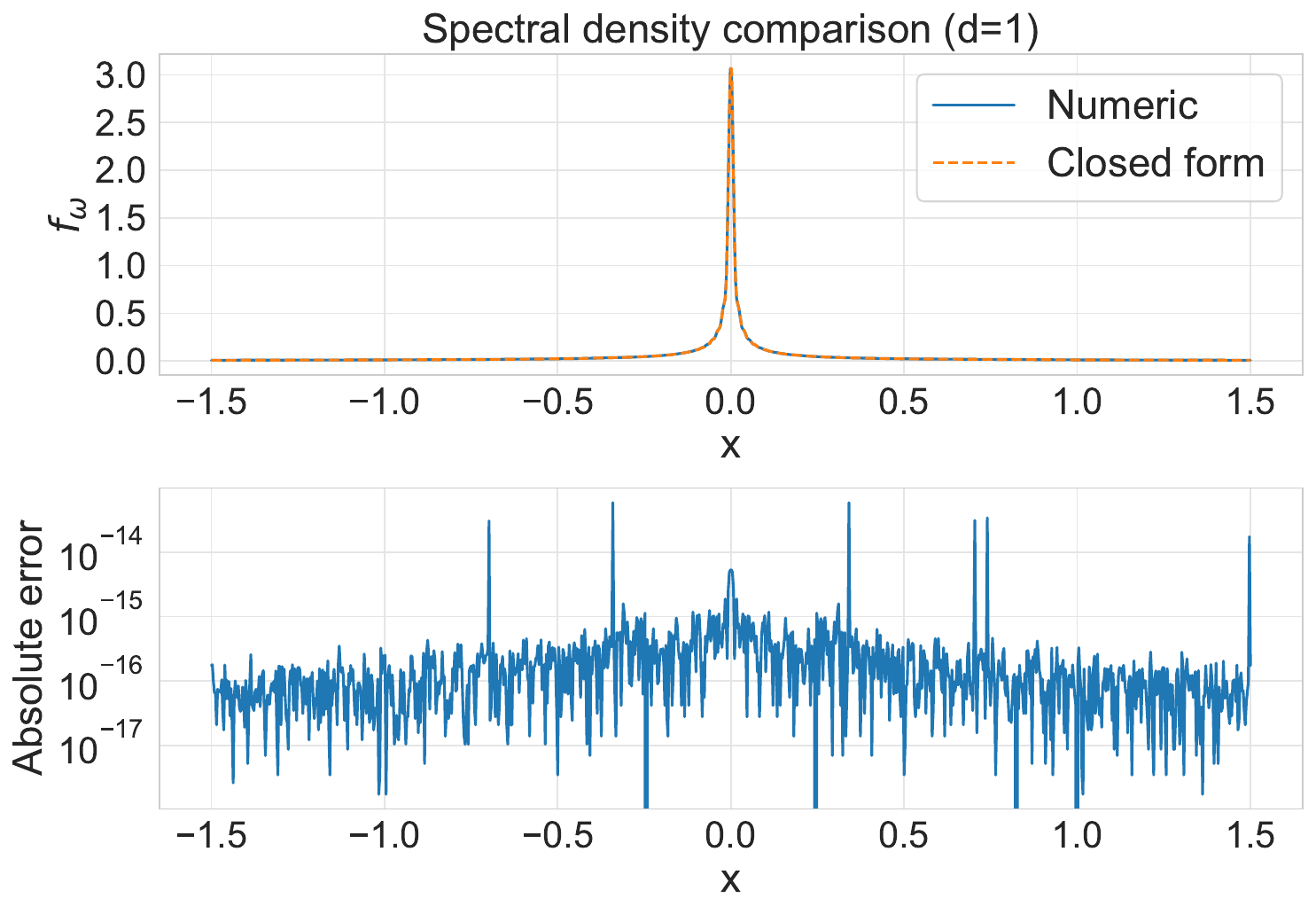}
        \caption{\(T = 400\)}
    \end{subfigure}

    \vspace{0.4cm}

    % --- Bottom row (centered) ---
    \begin{subfigure}[t]{0.6\textwidth}
        \centering
        \includegraphics[width=\textwidth]{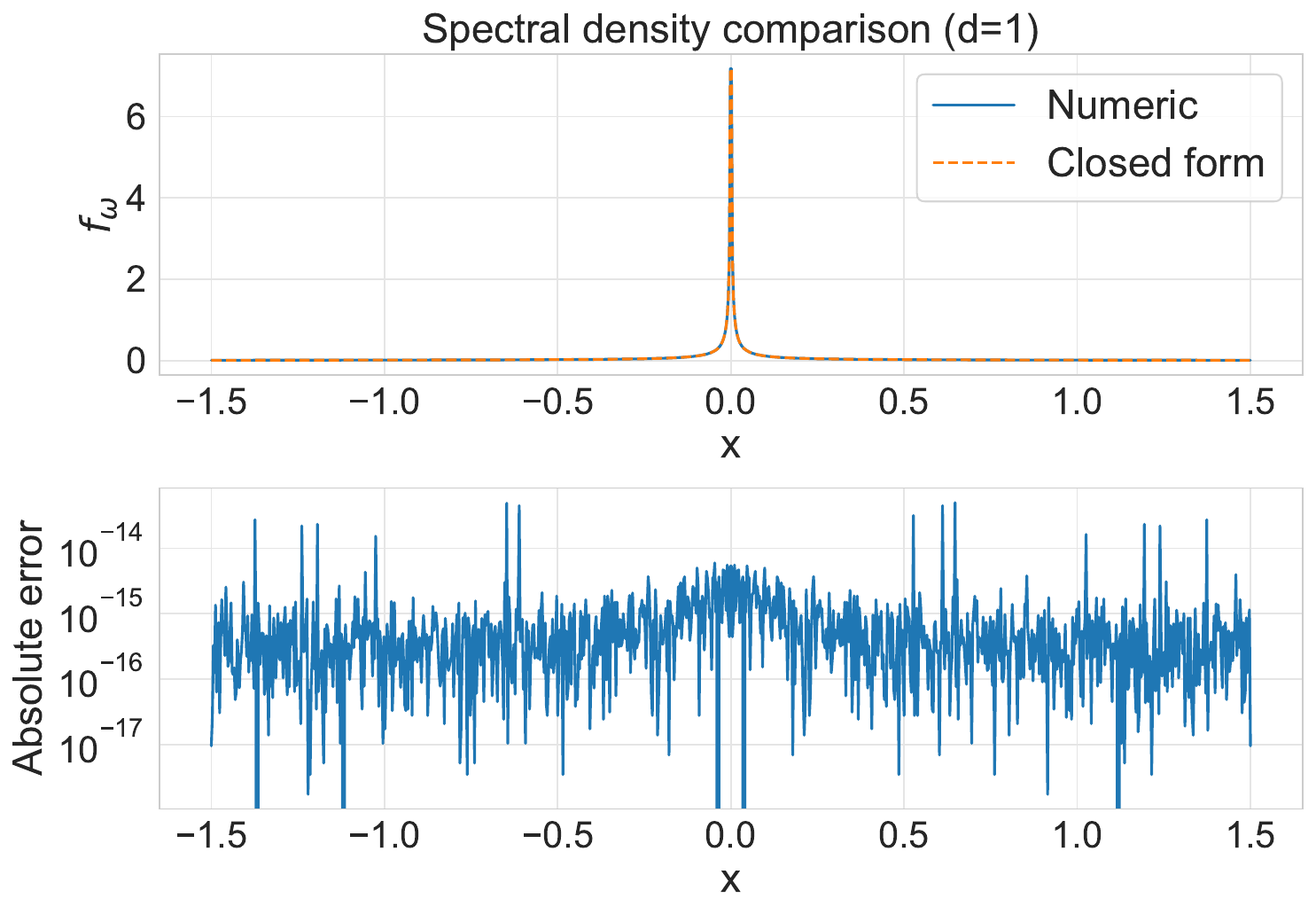}
        \caption{\(T = 4000\)}
    \end{subfigure}

    \caption{Illustration of the estimated quantities for increasing correlation limits \(T\) with fixed parameters \(H = 0.01\) and \(\nu^2 = 1\).
    The top panels correspond to \(T = 40\) and \(T = 400\), while the bottom panel shows the results for \(T = 4000\).}
    \label{fig:spectraldensityhypergeo}
\end{figure}

\begin{figure}[H]
    \centering
    % --- Top: surface plot ---
    \begin{subfigure}[t]{\textwidth}
        \centering
        \includegraphics[width=\textwidth]{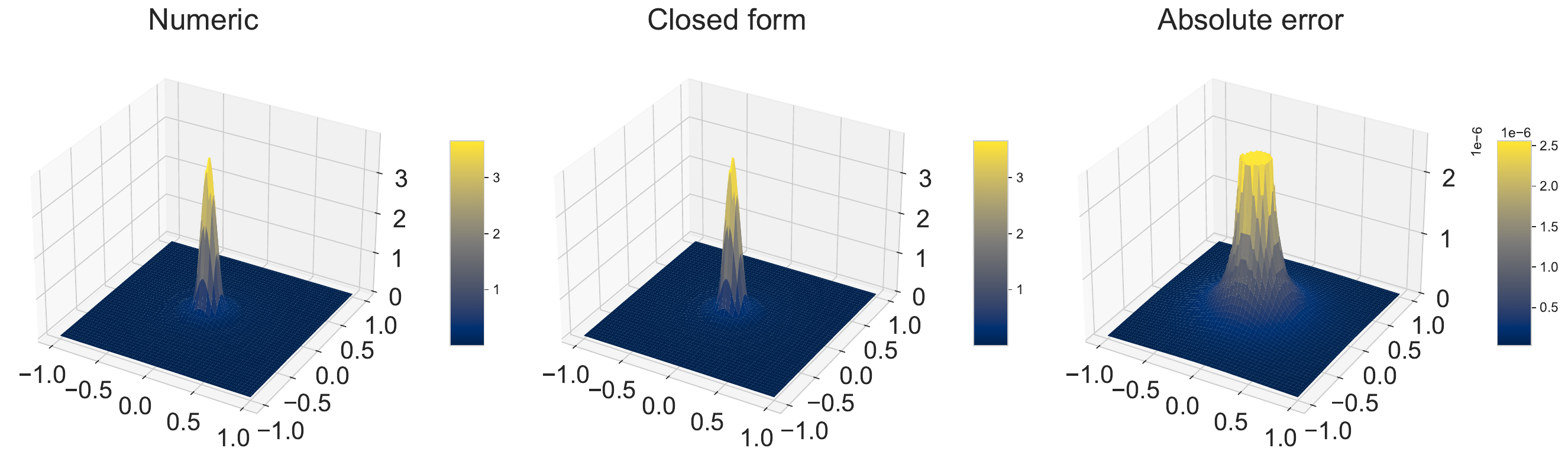}
        \caption{Pointwise error surface}
    \end{subfigure}

    \vspace{0.4cm}

    % --- Bottom row: colormaps ---
    \begin{subfigure}[t]{\textwidth}
        \centering
        \includegraphics[width=\textwidth]{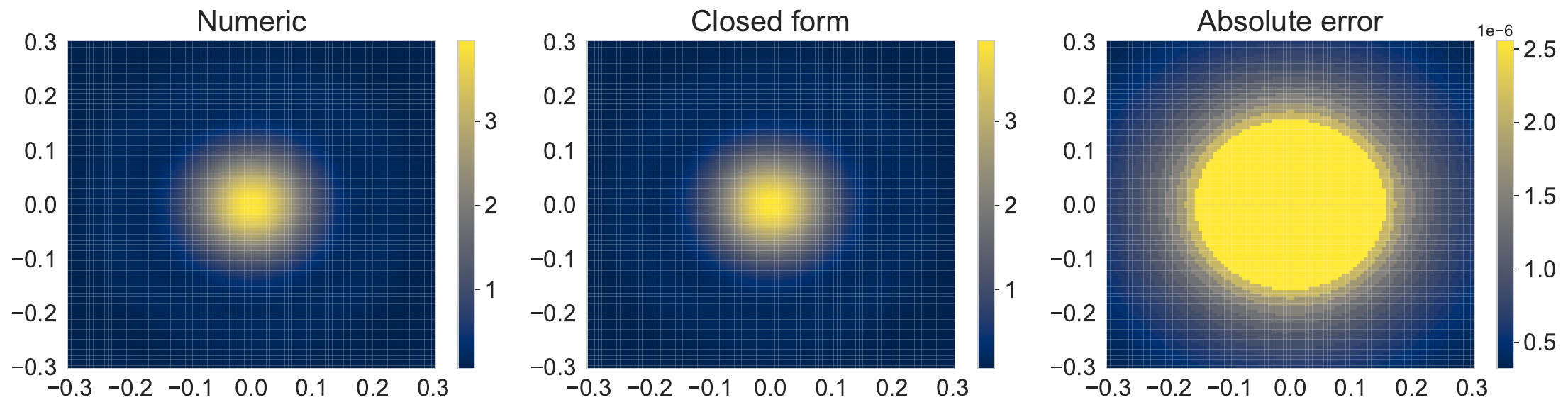}
        \caption{Colormap of the absolute error}
    \end{subfigure}

    \caption{Two dimensional validation of the spectral density representation for \(H = 0.01\) and \(\nu^2 = 2\).
    The top panel shows the surface of the pointwise difference between the analytical spectral density and its numerical Fourier integration reference for \(T = 40\).
    The bottom panels display the corresponding colormap representations, highlighting that the approximation error remains at the level of machine precision across the frequency domain.}
    \label{fig:d2_spectraldensityhypergeo}
\end{figure}

\subsection{Proof of Theorem~\ref{thm:GMMRFF}}
\label{subsec:proof_gmm_rff}

Consider the GMM objective functions associated with the true and approximated discrepancies:
\begin{eqnarray}
    \forall x\in \Xi,\quad  
    \begin{cases}
        \Tau(\bm{x}):=\bm{h}(\bm{x})^{\!\top} W \bm{h}(\bm{x}), \\
        \Hat{\Tau}(\bm{x}):=\Hat{\bm{h}}(\bm{x})^{\!\top} W \Hat{\bm{h}}(\bm{x}),
    \end{cases} \nonumber
\end{eqnarray}
where $\nabla$ and $\nabla^2$ denote the gradient and Hessian operators, respectively.\\ 
By definition of the GMM estimator, $\Hat{\theta}_M$ satisfies the first-order condition:
\begin{eqnarray}
\nabla \Hat{\Tau}\left(\Hat{\theta}_M\right) = 0. \nonumber
\end{eqnarray}
Using a first-order Taylor expansion, there exists $\tilde{\theta}\in \Xi$ such that:
\begin{eqnarray}
\nabla  \Hat{\Tau}\left(\Hat{\theta}_M\right)  
= \nabla  \Hat{\Tau}\left(\theta^{*}\right) + \nabla^2  \Hat{\Tau}\left(\tilde{\theta}\right) \left(\Hat{\theta}_M - \theta^{*}\right). \nonumber
\end{eqnarray}
Hence,
\begin{eqnarray}
\Hat{\theta}_M - \theta^{*} = - \left(\nabla^2  \Hat{\Tau}\left(\tilde{\theta}\right)\right)^{-1} \nabla  \Hat{\Tau}\left(\theta^{*}\right). \nonumber
\end{eqnarray}
Since $\Hat{\Tau}(.)$ is differentiable, we have for any $x\in\Xi$:
\begin{eqnarray}
\begin{cases}
    \nabla  \Hat{\Tau}\left(\bm{x}\right) = 2 \Hat{J}\left(\bm{x}\right)^\top W \Hat{\bm{h}}\left(\bm{x}\right),\\[4pt]
    \nabla^2  \Hat{\Tau}\left(\bm{x}\right) = 2 \Hat{J}\left(\bm{x}\right)^\top W\Hat{J}\left(\bm{x}\right) + 2\left(\nabla^2\Hat{\bm{h}}\left(\bm{x}\right)\right) W \Hat{\bm{h}}\left(\bm{x}\right),
\end{cases} \nonumber
\end{eqnarray}
where $\Hat{J}$ denotes the Jacobian of $\Hat{\bm{h}}$ and the contraction $\left(\nabla^2\Hat{\bm{h}}\left(\bm{x}\right)\right) W \Hat{\bm{h}}\left(\bm{x}\right)$ is defined componentwise as:
\begin{eqnarray}
\left(\nabla^2\Hat{\bm{h}}\left(\bm{x}\right)\right) W \Hat{\bm{h}}\left(\bm{x}\right)
= \sum_{k=1}^Q \left(W\Hat{\bm{h}}\left(\bm{x}\right)\right)_k \nabla^2 \left(\Hat{\bm{h}}\right)_k\left(\bm{x}\right), \nonumber
\end{eqnarray}
where $\left(\Hat{\bm{h}}\right)_k$ is the $k^{th}$ component of $\Hat{\bm{h}}$. One has straightforwardly that:
\begin{eqnarray}
\left\|\Hat{\theta}_M - \theta^{*}  \right\|_{2}^2 \leq\frac{\left\| \nabla  \Hat{\Tau}\left(\theta^{*}\right) \right\|_{2}^2}{\left|\left|\left|\nabla^2  \Hat{\Tau}\left(\tilde{\theta}\right) \right|\right|\right|^2}.\nonumber
\end{eqnarray}
Let us lower bound the denominator. Applying the inverse triangle inequality gives:
\[
\left|\left|\left| \nabla^2 \Hat{\Tau}(\mathbf{x}) \right|\right|\right|\ge 2 \left|\left( \left|\left|\left|\Hat{J}(\mathbf{x})^\top W \Hat{J}(\mathbf{x}) \right|\right|\right|-\left|\left|\left| (\nabla^2 \Hat{\mathbf{h}}(\mathbf{x})) W \Hat{\mathbf{h}}(\mathbf{x}) \right|\right|\right|\right)\right|.
\]
As $W$ is symmetric positive definite, one has:
\[
\left|\left|\left| \Hat{J}(\mathbf{x})^\top W \Hat{J}(\mathbf{x}) \right|\right|\right|\ge \sigma_{\min}(\Hat{J}(\mathbf{x})^\top W \Hat{J}(\mathbf{x})) = \sigma_{\min}(W) \, \sigma_{\min}(\Hat{J}(\mathbf{x}))^2.
\]
On the other hand:
\[
\left|\left|\left| (\nabla^2 \Hat{\mathbf{h}}(\mathbf{x})) W \Hat{\mathbf{h}}(\mathbf{x}) \right|\right|\right|
\le \left\|W\right\|_{\infty}\left\|\Hat{\bm{h}}\left(\bm{x}\right)\right\|_{\infty}\sum_{k=1}^Q    \left|\left|\left|\nabla^2 \left(\Hat{\bm{h}}\right)_k\left(\bm{x}\right) \right|\right|\right|.
\]
This can be rewritten as:
\[
\left|\left|\left| (\nabla^2 \Hat{\mathbf{h}}(\mathbf{x})) W \Hat{\mathbf{h}}(\mathbf{x}) \right|\right|\right|
\le \left\|W\right\|_{\infty}\left\|\Hat{\bm{h}}\left(\bm{x}\right)\right\|_{\infty}\Hat{H}_{Q}(\bm{x}),
\]
where $\Hat{H}_{Q}(\bm{x})=\sum_{k=1}^Q    \left|\left|\left|\nabla^2 \left(\Hat{\bm{h}}\right)_k\left(\bm{x}\right) \right|\right|\right|$. This leads to the following lower bound for any $\bm{x}\in \Xi$:

\[
\begin{aligned}
\left|\left|\left| \nabla^2 \Hat{\Tau}(\mathbf{x}) \right|\right|\right|
&\ge 2 \Big[ \sigma_{\min}(W) \, \sigma_{\min}(\Hat{J}(\mathbf{x}))^2 - \left\|W\right\|_{\infty}\left\|\Hat{\bm{h}}\left(\bm{x}\right)\right\|_{\infty}\Hat{H}_{Q}(\bm{x}) \Big]_+.
\end{aligned}
\]
where $\sigma_{\min}(.)$ is the smallest singular value and $[.]_{+}$ is the positive part.\\
On the other hand, after some algebra one has:
\begin{eqnarray}
\left\| \nabla  \Hat{\Tau}\left(\bm{x}\right) \right\|_{2} \leq 2 \left|\left|\left|\Hat{J}\left(\bm{x}\right)\right|\right|\right| \left|\left|\left|W\right|\right|\right| \left\|\Hat{\bm{h}}\left(\bm{x}\right) \right\|_2\nonumber
\end{eqnarray}
This leads to:
\begin{eqnarray}
\left\|\Hat{\theta}_M - \theta^{*}  \right\|_{2} \leq\frac{\left|\left|\left|\Hat{J}\left(\theta^{*}\right)\right|\right|\right| \left|\left|\left|W\right|\right|\right| \left\|\Hat{\bm{h}}\left(\theta^{*}\right) \right\|_{2}}{\Big[ \sigma_{\min}(W) \, \sigma_{\min}(\Hat{J}(\tilde{\theta}))^2 - \left\|W\right\|_{\infty}\left\|\Hat{\bm{h}}\left(\tilde{\theta}\right)\right\|_{\infty}\Hat{H}_{Q}(\tilde{\theta}) \Big]_+}.\nonumber
\end{eqnarray}
By direct computations:
\begin{eqnarray}
    \E\left[\left\|\Hat{\bm{h}}\left(\theta^{*}\right) \right\|_{2}^2 \right]= \frac{1}{M}\sum_{k=1}^Q\Var\left(\cos\left(\eta \tau_k\right) \right)\leq \frac{Q}{M}.\nonumber
\end{eqnarray}
As a result, there exists an explicit positive constant:
\begin{eqnarray}
    K_{W}:=\frac{\left|\left|\left|\Hat{J}\left(\theta^{*}\right)\right|\right|\right| }{\Big[ \sigma_{\min}(W) \, \sigma_{\min}(\Hat{J}(\tilde{\theta}))^2 - \left\|W\right\|_{\infty}\left\|\Hat{\bm{h}}\left(\tilde{\theta}\right)\right\|_{\infty}\Hat{H}_{Q}(\tilde{\theta}) \Big]_+}\nonumber
\end{eqnarray}
such that:
\begin{eqnarray}
\left\|\Hat{\theta}_M - \theta^{*}  \right\|_{\mathcal{L}^2} \leq K_{W} \left|\left|\left|W \right|\right|\right|\sqrt{\frac{Q}{M}}.\nonumber
\end{eqnarray}\\
\cqfd

\end{document}